\documentclass[ejsv2,authoryear]{imsart}

\RequirePackage{natbib}
\RequirePackage[colorlinks,citecolor=blue,urlcolor=blue]{hyperref}
\RequirePackage{graphicx}

\usepackage{algorithm}
\usepackage{algorithmic}
\usepackage{array}
\usepackage{bm}
\usepackage{bbm}
\usepackage{color}
\usepackage{float}
\usepackage{booktabs}
\usepackage{bigints}
\usepackage{extarrows}
\usepackage{multirow}
\usepackage{subcaption}
\usepackage[bbgreekl]{mathbbol}
\usepackage{mathtools}
\usepackage[section]{placeins}
\usepackage{verbatim}
\usepackage{enumitem}
\usepackage{xcolor}
\usepackage{optidef}
\usepackage{graphicx}
\usepackage{rotating}

\startlocaldefs
\theoremstyle{plain}
\newtheorem{theorem}{Theorem}[section]
\newtheorem{lemma}[theorem]{Lemma}
\newtheorem{remark}{Remark}[section]

\theoremstyle{definition}
\newtheorem{definition}[theorem]{Definition}

\theoremstyle{remark}

\numberwithin{equation}{section}
\numberwithin{theorem}{section}
\numberwithin{proposition}{section}
\numberwithin{corollary}{section}
\numberwithin{remark}{section}
\numberwithin{table}{section}
\numberwithin{figure}{section}

\endlocaldefs

\usepackage[most]{tcolorbox}

\newtcolorbox{highlightremark}{
    colback=white,
    colframe=black,
    boxrule=0.5pt,
    arc=1mm,
    left=2mm,
    right=2mm,
    top=1mm,
    bottom=1mm
}

\newcommand{\mE}{\mathbb{E}}
\newcommand{\mP}{\mathbb{P}}
\newcommand{\bW}{\mathbf{W}}
\newcommand{\bX}{\mathbf{X}}
\newcommand{\bY}{\mathbf{Y}}
\newcommand{\bZ}{\mathbf{Z}}

\newcommand{\bOmega}{\mathbf{\Omega}}
\newcommand{\bS}{\mathbf{S}}
\newcommand{\bP}{\mathbf{P}}
\newcommand{\bF}{\mathbf{F}}
\newcommand{\bG}{\mathbf{G}}
\newcommand{\bU}{\mathbf{U}}
\newcommand{\bD}{\mathbf{D}}

\newcommand{\bK}{\mathbf{K}}
\newcommand{\bJ}{\mathbf{J}}

\newcommand{\calD}{\mathcal{D}}
\newcommand{\calF}{\mathcal{F}}
\newcommand{\calH}{\mathcal{H}}
\newcommand{\calG}{\mathcal{G}}
\newcommand{\calK}{\mathcal{K}}
\newcommand{\calL}{\mathcal{L}}
\newcommand{\calP}{\mathcal{P}}
\newcommand{\TW}{\mathbb{TW}}

\newcommand{\bL}{\mathbf{L}}
\newcommand{\bH}{\mathbf{H}}

\newcommand{\tr}{\operatorname{tr}}

\newcommand{\SNR}{{\rm SNR}}
\newcommand{\LRT}{{\mathcal{T}}}

\renewcommand{\Re}{\operatorname{Re}}
\renewcommand{\Im}{\operatorname{Im}}

\newcommand{\I}{\mathcal{I} }

\begin{document}

\begin{frontmatter}
\title{Adaptable Regularized CCA Tests for Independence of High-Dimensional Random Vectors}
\runtitle{Regularized Test for Independence}

\begin{aug}
\author[A]{\fnms{Haoran}~\snm{Li}\ead[label=e1]{hzl0152@auburn.edu}}
\address[A]{Department of Mathematics and Statistics, Auburn University\printead[presep={,\ }]{e1}}
\runauthor{H. Li}
\end{aug}

\begin{abstract}
We propose an adaptable testing procedure for independence between two high-dimensional random vectors. The method incorporates ridge regularization and principal component-based dimension reduction into the canonical correlation analysis (CCA) framework, thereby stabilizing classical test statistics in high-dimensional settings. Depending on the reduced dimension, we develop both a regularized likelihood ratio test and a regularized largest-root test to accommodate different testing scenarios. We establish the asymptotic behavior of the proposed procedures under both the null hypothesis and representative alternatives, and further develop a data-driven method for selecting the regularization parameter. 
Extensive simulation studies demonstrate favorable finite-sample performance across a broad range of settings.
\end{abstract}

\begin{keyword}
\kwd{canonical correlation analysis}
\kwd{random matrix theory}
\kwd{Tracy-Widom law}
\kwd{ridge-regularization}    
\end{keyword}

\end{frontmatter}

\section{Introduction}\label{sec:introduction}

In many modern applications, variables can naturally be partitioned into several mutually exclusive groups according to their scientific roles or practical interpretations. Examples include collections of genes associated with different biological pathways, groups of financial assets from distinct industrial sectors, and sets of nodes belonging to separate communities in a network. A fundamental statistical question is whether these groups exhibit significant dependence or whether they may reasonably be treated as mutually independent. The validity of an independence assumption often carries important scientific and methodological implications; see, for example, applications in Gaussian graphical models \citep{devijver2018block}, gene expression analysis \citep{xu2019nonparametric}, and financial econometrics \citep{yang2015independence}. From a modeling perspective, independence assumptions can substantially reduce the number of unknown parameters, thereby simplifying estimation and improving interpretability. From a structural perspective, rejecting independence may reveal latent interaction patterns, common driving factors, or underlying information transmission mechanisms across groups of variables. Consequently, testing independence is frequently an essential preliminary step prior to subsequent analyses such as classification, clustering, dimension reduction, and network inference.

In this paper, we focus on a simplified yet representative setting involving two groups of variables, organized into two random vectors $X \in \mathbb{R}^{p_1}$ and $Y \in \mathbb{R}^{p_2}$ with dimensions $p_1$ and $p_2$, respectively. Denote the covariance and cross-covariance matrices as 
\[\Sigma_{x} = \mE (X -\mE X)(X-\mE X)^T, ~~ \Sigma_{y} = \mE (Y -\mE Y)(Y-\mE Y)^T, ~~ \Sigma_{xy} = \mE (X -\mE X)(Y-\mE Y)^T.  \]
Suppose that an i.i.d. sample of size $n$ from these two vectors $X$ and $Y$ is observed jointly, denoted as $(X_i^T, Y_i^T)$, $i=1,\dots, n$.  Moreover, denote the two centered data matrices as
\[ {\bX} = [X_1-\bar{X}, \dots, X_n-\bar{X}] \in \mathbb{R}^{p_1\times n} \quad \text{and} \quad {\bY} = [ Y_1-\bar{Y}, \dots, Y_n-\bar{Y}]  \in \mathbb{R}^{p_2\times n},\]
where $\bar{X} = n^{-1} \sum_{i=1}^n X_i$ and $\bar{Y} = n^{-1} \sum_{i=1}^n Y_i$ are the sample means. Define the sample covariance and cross-covariance matrices as 
\[ \bS_x =  (n-1)^{-1} {\bX}{\bX}^T, \quad \bS_y = (n-1)^{-1}{\bY}{\bY}^T, \quad \bS_{xy} = (n-1)^{-1}{\bX}{\bY}^T.\]

In classical multivariate analysis, the independence test is addressed under the framework of \emph{canonical correlation analysis} (CCA) under Gaussianity. In particular, if $X$ and $Y$ are jointly normal, testing independence between $X$ and $Y$ is equivalent to testing $\Sigma_{xy} = 0$. The traditional test procedures are then designed to measure the eigenvalues of the sample estimator of the normalized cross-covariance matrix $\widetilde{\Sigma}_{xy} = \Sigma_{x}^{-1/2} \Sigma_{xy} \Sigma_y^{-1/2}$.  

To motivate our proposed test procedures, we briefly summarize the traditional CCA tests. Classical literature commonly employs two main types of techniques, both relying on the eigenvalues $\rho_j$ of 
$\bS_{xy}\bS_{y}^{-1}\bS_{xy}^T \bS_{x}^{-1}$. 
The first approach is trace-based and involves aggregating all eigenvalues $\rho_j$ after applying a specific transformation, such as the classical likelihood ratio test (LRT). This approach is expected to be efficient when the departure of $\widetilde{\Sigma}_{xy}$ from zero is spread across all eigendirections under the alternative. The second approach, known as Roy's largest root test, relies exclusively on the largest eigenvalue, say $\max_j \rho_j$. This technique is particularly powerful for detecting concentrated alternatives under which the departure of $\widetilde{\Sigma}_{xy}$ from zero is concentrated on its leading eigendirections. 

To better illustrate the structure underlying these procedures, it is common to reformulate the test matrix $\bS_{xy} \bS_y^{-1} \bS_{xy}^T \bS_{x}^{-1}$ through an associated $F$-matrix representation. Define 
\[ \bP_0 = \bY^T (\bY\bY^T)^{-1}\bY,\]
which is the projection matrix onto the row space of $\bY$. We project $\bX$ onto the subspace associated with $\bP_0$ and its orthogonal complement $I_n-\bP_0$, and define the corresponding cross-product matrices $\bW_{01} = \bX \bP_0 \bX^T$ and $ \bW_{02} = \bX (I_n - \bP_0) \bX^T$.
Further, define the $F$-matrix 
\[ \bF_0 = \bW_{01} \bW_{02}^{-1}. \]
Then, $\rho_j/(1-\rho_j)$ 
is an eigenvalue of $\bF_0$. Since the transformation $\rho \mapsto {\rho}/{1-\rho}$ is strictly increasing on $(0,1)$, inference based on the eigenvalues $\rho_j$ is equivalent to inference based on the eigenvalues of $\bF_0$. 
We refer to Chapter 11 of \cite{muirhead2009aspects} for more details.


While the traditional CCA tests demonstrate reliable performance in the classical regime when the sample size $n$ is substantially larger than the total dimension $(p_1+p_2)$, its reliability diminishes significantly in high-dimensional settings in which $(p_1+p_2)$ is at least comparable to $n$. The main issue is the singularity or near-singularity of the matrices $\bY\bY^T $ and $\bW_{02}$. In particular, when $n< p_1+ p_2$, the matrix $\bW_{02}$ is not invertible; when $n<p_2$, $\bY\bY^T$ is not invertible, rendering the tests ill-defined. Even when $n>p_1+p_2$ but $n/(p_1+p_2) \approx 1$, the presence of small eigenvalues in $\bW_{02}$ can cause instability in $\bW_{02}^{-1}$, thereby compromising the power of the test. Addressing these issues is essential to ensure the validity of the test in high-dimensional regimes.

In this paper, we consider the high-dimensional regime in which both $p_1$ and $p_2$ are comparable to, or may even exceed, the sample size $n$. Our framework is motivated by the premise that, whenever $X$ and $Y$ are dependent, the dependence structure is primarily captured through the leading principal components (PCs) of $Y$. In contrast, the residual components associated with smaller eigenvalues are assumed to contribute negligibly to the dependence structure and to be approximately independent of $X$. 
Specifically, we make the following model assumptions: 
\begin{equation}\label{eq:model_PCR}
\begin{split}
X &= M \Lambda_{m}^T Y + \Sigma_0^{1/2} Z_x,\\
Y &= \Sigma_y^{1/2} Z_y.
\end{split}
\end{equation}
Here, $\Sigma_0 \in \mathbb{R}^{p_1 \times p_1}$ and $\Sigma_y \in \mathbb{R}^{p_2 \times p_2}$ are nonnegative definite covariance matrices; $Z_x \in \mathbb{R}^{p_1}$ and $Z_y \in \mathbb{R}^{p_2}$ are independent random vectors with i.i.d. entries of mean zero and variance one; $M \in \mathbb{R}^{p_1 \times m}$ is a deterministic coefficient matrix; and $\Lambda_{m} \in \mathbb{R}^{p_2 \times m}$ is the matrix consisting of the first $m$ eigenvectors of $\Sigma_y$, where $m \in \mathbb{N}_+$ ($m\leq p_2$)  may depend on $n$, $p_1$, and $p_2$. Moreover, mean shifts may be added to $X$ and $Y$ without affecting the subsequent analysis.

Under this model, the potential dependence between $X$ and $Y$ is assumed to be linear through only the first $m$ PCs of $Y$, expressed as $\Lambda^T_{m}Y$. Testing the independence between $X$ and $Y$ is transformed to testing 
\begin{equation}
    \label{eq:hypothesis_B}
    H_0~:~  M\Lambda_{m}^T = 0 \qquad \text{against}\qquad H_a~:~ M\Lambda_{m}^T\neq 0. 
\end{equation}

Notably, under the proposed framework, $X$ and $Y$ play asymmetric roles. In practical applications where two sets of variables are jointly observed, before applying the proposed procedures, one must therefore determine which random vector should be treated as $X$ and which as $Y$. We recommend a data-driven selection rule based on the spectral properties of the corresponding sample covariance matrices, presented in Section~\ref{sec:practical_guidance}. Throughout the analysis, we assume that this designation has been made.

To construct stable and powerful high-dimensional tests, we propose modifications of the classical trace-based tests and the largest-root test that incorporate ridge regularization together with PC-based dimension reduction. 

Specifically, let a rank parameter $k$ $(\leq p_2)$ be given, whose practical selection is discussed in Section~\ref{sec:practical_guidance}. Rather than projecting $\bX$ onto the entire row space of $\bY$ as in $\bP_0$, we project onto the subspace spanned by the first $k$ sample PCs of $\bY$. Let
$\hat{\Lambda}_k\in\mathbb{R}^{p_2\times k}$
denote the matrix whose columns are the leading $k$ eigenvectors of the sample covariance matrix $\bS_y$. We then define the corresponding projection matrices
\[
\bP_k
=
\bY^T
\hat{\Lambda}_k
\big(
\hat{\Lambda}_k^T
\bY\bY^T
\hat{\Lambda}_k
\big)^{-1}
\hat{\Lambda}_k^T
\bY.
\]
Using these projection matrices, we define
\[
\bW_{k1}
=
\frac{1}{k}
\bX\bP_k\bX^T,
\qquad
\bW_{k2}
=
\frac{1}{n-1-k}
\bX
(I_n-\bP_{k})
\bX^T.
\]



Next, since $\bW_{k2}$ becomes singular or nearly singular when $p_1 \gtrsim n-1-k$, we stabilize its inverse by incorporating ridge regularization. Specifically, for a prescribed regularization parameter $\lambda>0$, we define the ridge-regularized and rank-reduced $F$-matrix
\[
\bF_{k\lambda}
=
\bW_{k1}
(\bW_{k2}+\lambda I_{p_1})^{-1}.
\]

The proposed testing procedures depend on the magnitude of the reduced dimension $k$. When $k$ is relatively small, we consider a regularized trace-based statistic defined by 
\begin{equation}
\label{eq:LRT}
\LRT(k,\lambda)
=
\sum_{j=1}^{k} \ell_j(\bF_{k\lambda}).
\end{equation}
Here, $\ell_j(\cdot)$ denotes the $j$th largest eigenvalue of a matrix with real-valued eigenvalues, while $\ell_{\min}(\cdot)$ and $\ell_{\max}(\cdot)$ are the smallest and the largest one, respectively. 

On the other hand, when $k$ is relatively large, we propose to use the regularized largest-root statistic
\begin{equation}
\label{eq:largest_root}
\ell_{\max}(k,\lambda)
=
\ell_{\max}(\bF_{k\lambda}).
\end{equation}

Both procedures are well defined for arbitrary scenarios of $p_1$, $p_2$, and $n$. An appealing feature of the proposed statistics is their \emph{rotation invariance}. Specifically, the statistics remain unchanged if either $X$ or $Y$ is transformed by an arbitrary orthogonal transformation. This property is particularly desirable in high-dimensional settings where limited structural information, such as coordinate-wise sparsity, is available. In particular, it implies that the performance of the proposed tests does not depend on the coordinate system used to represent the data. Consequently, the procedures are robust to different data representations and broadly applicable across a wide range of high-dimensional problems.


The remainder of the paper is organized as follows. Section~\ref{sec:literature_contribution} reviews the related literature and highlights the contributions of the present work. In Section~\ref{sec:analysis_null}, we establish the asymptotic null distributions of the proposed statistics under different asymptotic regimes of the reduced dimension parameter $k$. Section~\ref{sec:power_analysis} investigates the asymptotic power properties of the proposed procedures under representative alternatives. In Section~\ref{sec:selection_regularization_parameter}, we study the selection of the regularization parameter $\lambda$ and develop a data-driven procedure motivated by a Bayesian--minimax principle. Section~\ref{sec:practical_guidance} provides practical guidance for implementing the proposed methodology.  
Finally, Section~\ref{sec:simulation} presents extensive simulation studies evaluating the finite-sample performance of the proposed methodology. Additional materials and technical proofs are deferred to the Appendix.

\section{Literature Review and Our Contribution}\label{sec:literature_contribution}

High-dimensional tests of independence have been developed along several complementary directions, differing primarily in the structural alternatives they target and the functionals of the cross-covariance matrix they exploit.

One major class consists of trace-based procedures, which aggregate information across the entire spectrum of the cross-covariance structure. Early contributions by \cite{zheng2012central} and \cite{jiang2013central} extended the classical theory of spectral functionals of $F$-matrices to high-dimensional settings. In a related direction, \cite{yang2015independence} proposed a regularized LRT-type procedure based on ridge-regularized sample covariance matrices, thereby extending the classical likelihood ratio framework to settings in which $p_1$ may exceed $n$ while $p_2$ remains smaller than $n$. A comparison between their method and the proposed procedure is provided in Remark~\ref{remark:compare_yang_pan}. Other representative works in this class include \cite{bao2017test,hyodo2015testing,jiang2013testing,yamada2017testing,bodnar2019testing,zheng2019hypothesis}.

In contrast, largest-root-based procedures have mainly been investigated in regimes where $(p_1+p_2)/n \lesssim 1$. In particular, \cite{johnstone2008multivariate}, \cite{johnstone2009approximate}, and \cite{han2018unified} established Tracy--Widom limits for the largest root after appropriate centering and scaling. The behavior of the largest root under low-rank alternatives was further investigated in \cite{bao2019canonical}, with additional developments in \cite{bai2022limiting}. However, to the best of our knowledge, existing largest-root procedures are not applicable to settings in which both $p_1$ and $p_2$ may exceed $n$. The present work addresses this limitation by developing a regularized largest-root framework that remains applicable in such regimes.

Complementary to these CCA-based approaches, which primarily target linear dependence structures, there is also a broad literature on nonparametric procedures for detecting more general forms of dependence. Representative examples include distance covariance based methods \citep{szekely2007measuring}, kernel-based approaches \citep{gretton2005kernel,gretton2010consistent}, rank-based procedures \citep{wang2026testing}, random projection methods \citep{najarzadeh2025testing}, and U-statistic-based tests \citep{lai2023block}.

On the technical side, the current work builds upon the well-established local law and universality framework in the random matrix theory (RMT) literature for various random matrix models. In particular, results concerning sample covariance matrices, spiked covariance models, and signal-plus-noise type matrices are especially relevant to the present setting and constitute an important part of the technical foundation of this paper. Representative references include \cite{Johnstone2001,BaikBenArousPeche2005,el2007tracy,Onatski2008,BaoPanZhou2013LocalEdge,bao2015,PillaiYin2012,PillaiYin2014,bloemendal2014isotropic,lee2016tracy,knowles2017anisotropic,DingYang2018AAP,Yang2019EJP,ding2021spiked,dingYang2021spiked,ding2022edge,ding2023extreme,ding2024eigenvector,ZhangLiuPan2024TWSignalNoise,ding2025tracy,li2025ridge}.
%
 


We summarize the main contributions of the present work as follows. From a methodological perspective, we propose a flexible and statistically stable regularization framework that extends classical CCA-based tests. By incorporating ridge regularization together with PC-based dimension reduction, the proposed procedures remain stable and effective when both $p_1$ and $p_2$ are comparable to, or may exceed, the sample size $n$. 

Second, from a theoretical perspective, we establish the asymptotic null distributions of $\LRT(k,\lambda)$ and $\ell_{\max}(k,\lambda)$ under two distinct asymptotic regimes determined by the reduced dimension parameter $k$. These results provide a theoretical foundation for constructing testing procedures with asymptotically correct type-I error control and are established under mild structural and distributional assumptions. 




Third, we develop an asymptotic power analysis under representative alternatives, explicitly characterizing how the covariance structures of $X$ and $Y$, the cross-covariance structure between them, the reduced dimension parameter $k$, and the regularization parameter $\lambda$ jointly influence the proposed statistics. Building upon this analysis, we further develop a Bayesian--minimax framework for selecting the regularization parameter.

\section{Asymptotic Theory Under Independence}\label{sec:analysis_null}
In this section, we develop the asymptotic null theory for the proposed testing procedures under the assumption that $X$ and $Y$ are independent. Throughout the analysis, the regularization parameter $\lambda$ is treated as fixed, and the rank parameter $k$ is regarded as given.

\paragraph{\textit{Notation.}}
For any $p\times p$ matrix $A$, the quantities $\|A\|_2$ and $\|A\|_F$ denote the spectral norm and Frobenius norm of $A$, respectively. If $A$ is nonnegative definite, we use $A^{1/2}$ to denote its symmetric square root. Furthermore, the notation $\stackrel{P}{\longrightarrow}$ denotes convergence in probability, and $\stackrel{D}{\longrightarrow}$ denotes convergence in distribution, both with respect to the joint distribution of $X$ and $Y$. For two positive quantities $A = A(n, k)$ and $B = B(n, k)$ depending on $n$ and/or $k$, we use the notation $A \asymp B$ to mean $c^{-1} A \leq B \leq c A$ for some positive constant $c$. 

Under Model \eqref{eq:model_PCR}, we make the following technical assumptions. 
\begin{enumerate}[label={\bf C\arabic*}, start=1, leftmargin =2.5em]
\item \label{enum:high_dimension_regime} \textit{(H-D regime)}
We assume that as $n\to\infty$, $p_1 \asymp n$ and $p_2 \asymp n^\zeta$ for some constant $\zeta \geq 1$. Moreover, the rank parameter $k$ is not excessively large in the sense that $n-k \asymp n$. 
\item \label{enum:moments_conditions} \textit{(Moment conditions)} All moments of the entries of $Z_x$ are finite. 
\item \label{enum:boundedness_spectral_norm} \textit{(Bounded spectrum of $\Sigma_0$)}  $0<\liminf \ell_{\min}(\Sigma_0)\leq\limsup \ell_{\max}(\Sigma_0)<\infty$.
\item \label{enum:lower_bound_eigen_Y} \textit{(Sufficient rank)} There exist constants $C>0$ and $\alpha \in(0,1)$ such that $\mP (\ell_{\lfloor \alpha n \rfloor }(\bS_y)> C) \geq 1-\exp(-n)$, for all sufficiently large $n$. We shall always assume that $k/n \leq \alpha$.
\item \label{enum:edger_regularity} \textit{(Edge regularity)} The spectrum of $\Sigma_0$ is asymptotically regular near its largest eigenvalue in the sense of Definition \ref{def:edge_regularity}. 
\end{enumerate}
Condition~\ref{enum:high_dimension_regime} specifies the high-dimensional asymptotic regime under which our theory is developed. Condition~\ref{enum:moments_conditions} requires only moment assumptions on $Z_x$. Condition~\ref{enum:boundedness_spectral_norm} is a standard regularity assumption in the Random Matrix Theory literature that guarantees uniform control of the extreme eigenvalues.  The assumptions on $\Sigma_y$ and $Z_y$ are comparatively mild and are imposed implicitly through Condition~\ref{enum:lower_bound_eigen_Y}, which guarantees that the rank of $\bS_y$ is sufficiently large for the subsequent analysis. Condition~\ref{enum:edger_regularity} is a standard edge regularity condition in RMT, commonly used in the study of the asymptotic fluctuations of extreme eigenvalues. Additional discussion is provided in Section~\ref{subsec:preliminary}.

\subsection{Preliminaries on Random Matrix Theory}\label{subsec:preliminary}
We begin by presenting several fundamental results from random matrix theory (RMT) concerning the asymptotic behavior of $\bW_{k2}$ under the null hypothesis of independence. Results in this section are based on \cite{silverstein1995empirical,bai2004clt, knowles2017anisotropic, li2024analysis, li2025ridge}.

An important quantity in the subsequent analysis is the aspect ratio
\[
q=q(k,n)\coloneqq \frac{p_1}{n-1-k}.
\]
Note that under \ref{enum:high_dimension_regime}, $q\asymp 1$. 

Recall that, for any function $G$ of bounded variation on $\mathbb{R}$, its Stieltjes transform $\varphi_G(\cdot)$ is defined by
\[ \varphi_G(z) = \int \frac{dG(\tau)}{\tau-z}, \qquad z\in \mathbb{C}^+ \coloneqq \{u+iv:\ v>0\}.\]
The Stieltjes transform maps $\mathbb{C}^+$ into itself and plays a role in RMT analogous to that of the Fourier transform in classical probability theory. In particular, it characterizes probability distributions through an inversion formula and provides a convenient tool for studying weak convergence of probability measures via convergence of the corresponding transforms. We refer to \citet{bai2010spectral} for a comprehensive treatment.

\begin{definition}
\label{def:ESD}
For any $p\times p$ matrix $A$ with real-valued eigenvalues, the \emph{empirical spectral distribution} (ESD) of $A$ is defined by
\[
F^A(\tau)
=
\frac{1}{p}\sum_{j=1}^p
\mathbb{I}\bigl(\ell_j(A)\leq \tau\bigr).
\]
Equivalently, $F^A$ is the probability measure that assigns mass $1/p$ to each eigenvalue of $A$.
\end{definition}

Consider the ESD $F^{\Sigma_0}$ of $\Sigma_0$ and the aspect ratio $q$. 
The renowned Mar\v{c}enko-Pastur (M-P) equation in RMT defined on $z\in \mathbb{C}^+$ is expressed as 
\begin{equation}\label{eq:M_P_equation}
z = \frac{-1}{\varphi} + q \int \frac{\tau dF^{\Sigma_0}(\tau)}{ \tau \varphi + 1}.
\end{equation}
Provided that $F^{\Sigma_0}$ is not the Dirac measure at zero and $q>0$, it is known that \eqref{eq:M_P_equation} admits a unique solution $\varphi \equiv \varphi(z, q, F^{\Sigma_0}) \in \mathbb{C}^+$, for any $z\in\mathbb{C}^+$. The function $\varphi(z): \mathbb{C}^+ \mapsto \mathbb{C}^+$ is the Stieltjes transform of a probability measure with bounded support in $[0,\infty)$, denoted by $\calF(F^{\Sigma_0}, q)$. Unfortunately, except in extreme cases where $F^{\Sigma_0}$ has a simple structure, there is no explicit closed-form solution for $\varphi(z)$ or $\calF(F^{\Sigma_0}, q)$. For further details of the M-P equation, see Chapter 3 of \cite{bai2010spectral}. In the sequel, we frequently suppress the dependence on $q$ and/or $F^{\Sigma_0}$ in the notation whenever no ambiguity arises.

In the present paper, we are primarily interested in the behavior of $\varphi(z)$ when $z$ approaches the negative real axis. A standard result in RMT states that $\varphi(z)$ admits a smooth extension to $\mathbb{R}_-$. With a slight abuse of notation, we continue to denote the extended function by
\[
\varphi(x)
\coloneqq
\lim_{\substack{z\in\mathbb{C}^+\mapsto x}}
\varphi(z),
\qquad
x\in\mathbb{R}_-.
\]

Two particularly useful results in our setting are: (i) the existence of a \emph{deterministic equivalent} for the regularized inverse of $\bW_{k2}$; (ii) the availability of a fully data-driven estimator of $\varphi(z)$ that does not require knowledge of the population spectrum. In particular, adapting standard results from RMT to the present framework, define a matrix 
\begin{equation}\label{eq:deteministic_equivalent} 
\calD(z) = \Big( -z \varphi(z) \Sigma_0 - z I_{p_1}\Big)^{-1},\qquad z\in\mathbb{C}^+ \cup \mathbb{R}_-.
\end{equation}
Further, define
    \begin{equation}
    \label{eq:def_hat_varphi}
    \hat{\varphi}(z)
    =
    \frac{1}{n-1-k}
    \sum_{j=1}^{n-1-k}
    \frac{1}{\ell_j(\bW_{k2}) -z}, \qquad z\in\mathbb{C}^+ \cup \mathbb{R}_-,
    \end{equation}
    where $\ell_j(\bW_{k2})$ denotes the $j$th largest eigenvalue of $\bW_{k2}$ for $j\leq p_1$, and $\ell_j(\bW_{k2})=0$ for $j>p_1$. Notably, $\hat{\varphi}(z)$ can be directly computed from the data. 
    
Then, $\hat{\varphi}(z)$ provides a consistent estimator of $\varphi(z)$, while the matrix $\calD(z)$ serves as a deterministic equivalent of
$\big(\bW_{k2}-zI_{p_1}\big)^{-1}$ when $n$ is large, in the sense described below.

\begin{lemma}
\label{lemma:determinist_equivalent}
Suppose that Conditions~\ref{enum:high_dimension_regime}--\ref{enum:lower_bound_eigen_Y} hold. Assume that $X$ and $Y$ are independent and $q = q(k,n) \asymp 1$.
Let $\mathcal{C}$ be any closed and compact subset  of  $\mathbb{C}^+ \cup \mathbb{R}_-$. Then, the following statements hold:
\begin{itemize}
    \item[(i)]
    Let $B$ be any sequence of deterministic matrices satisfying $\|B\|_2<\infty$. Then, 
    \[
    \sup_{z\in \mathcal{C}} n^{2/3}\Big|n^{-1}
    \tr\!\big[
    \big(\bW_{k2}-zI_{p_1}\big)^{-1}B
    \big]
    -
    n^{-1}
    \tr\!\big[
    \calD(z)B
    \big]\Big|
    \stackrel{P}{\longrightarrow}
    0.
    \]

    \item[(ii)]
    Let $\mu$ be any sequence of deterministic vectors satisfying $\|\mu\|_2=1$. Then, 
    \[
    \sup_{z\in\mathcal{C}}n^{2/3}\Big|
    \mu^T
    \big(\bW_{k2}-zI_{p_1}\big)^{-1}
    \mu
    -
    \mu^T
    \calD(z)
    \mu\Big|
    \stackrel{P}{\longrightarrow}
    0.
    \]

    \item[(iii)]
    Let $\hat{\varphi}'(z)$ and $\varphi'(z)$ denote the derivatives of $\hat{\varphi}(z)$ and $\varphi(z)$ with respect to $z$, respectively.
    \[
    \sup_{z\in\mathcal{C}}
    n^{2/3}\big|
    \hat{\varphi}(z)-\varphi(z)
    \big|
    \stackrel{P}{\longrightarrow}
    0
       \qquad \text{and}
        \qquad
    \sup_{z\in\mathcal{C}}
    n^{2/3}\big|
    \hat{\varphi}'(z)-\varphi'(z)
    \big|
    \stackrel{P}{\longrightarrow}
    0.
    \]
\end{itemize}
\end{lemma}

Next, we introduce a variant of a standard edge regularity condition commonly used in the RMT literature for establishing asymptotic fluctuation of extreme eigenvalues. To this end, consider the following master function:
\begin{equation}
    \label{eq:def_g}
    g(h) =  g(h; \lambda, q, F^{\Sigma_0}) 
    =
    h + \Bigg(
    1 + q
    \int \frac{\tau\, dF^{\Sigma_0}(\tau)}{\lambda-\tau h}
    \Bigg)^{-1},
    \qquad
    h \in (-\infty,\eta),
\end{equation}
where $\eta=\lambda/\ell_{\max}(\Sigma_0)$. It is straightforward to verify that the denominators appearing in \eqref{eq:def_g} remain nonzero over the specified domain, and hence the function is well defined.  To alleviate notation complexity, we frequently suppress the dependence on $\lambda$, $q$, and $F^{\Sigma_0}$ in the notation whenever no ambiguity arises. 

A straightforward calculation shows that $g''(h)<0$ on the specified domain and therefore $g(h)$ is strictly concave. Moreover, one can verify that $g'(h)>0$ for sufficiently large negative values of $h$. Consequently, the function $g(h)$ either increases monotonically on $(-\infty,\eta)$ or attains a unique interior maximum within this interval, depending on whether the equation
\[
g'(h)=0
\]
admits a solution in the domain. These two scenarios correspond to two distinct asymptotic phases in the technical analysis. For details, see Section \ref{subsec:proof_main_diverge_k} in the Appendix. The phase transition occurs when $q$, $\lambda$, and $\Sigma_0$ are such that 
\[\lim_{n\to\infty}
\lim_{h\uparrow \eta}
g'(h) = 0.\]
To exclude this boundary case, we impose the following mild regularity condition.
\begin{definition}
\label{def:edge_regularity}
Fix $\lambda>0$ and suppose that $q \asymp 1$. We say that the spectrum of $\Sigma_0$ is asymptotically regular near $\ell_{\max}(\Sigma_0)$ if there exist constants $\epsilon_1>0$ and $\epsilon_2>0$ such that, for all sufficiently large $n$, 
\[
\big|g'(h)\big|
>
\epsilon_1,
\qquad
\text{for all }
h\in
(\eta-\epsilon_2,\eta),
\]
uniformly over $q$.
\end{definition}

The above condition is standard in the RMT literature, although it is often formulated in different but equivalent forms. Related versions appear, for example, in Definition~2.2 of \cite{li2025ridge}, Definition~2.7 of \cite{knowles2017anisotropic}, Condition~(1) of \cite{el2007tracy}, Condition~(2.12) of \cite{lee2016tracy}, and Assumption~2.5 of \cite{DingYang2018AAP}.

\subsection{Asymptotic distribution of the regularized trace-based statistic}\label{subsec:asymptotic_k_fixed}
In this section, we derive the asymptotic distribution of the proposed regularized trace-based statistic $\LRT(k,\lambda)$ in the regime where the reduced dimension parameter $k$ remains fixed as $n\to\infty$.

To begin with, for $k\geq 1$, denote by ${\rm GOE}_k = \big[g_{ij}\big]_{i,j=1}^k$ the \emph{Gaussian Orthogonal Ensemble}, defined by (1) $g_{ij} = g_{ji}$; (2) $g_{ii} \sim N(0,1)$, $g_{ij} \sim N(0,1/2)$, $i\neq j$; (3) $g_{ij}$'s are jointly independent subject to symmetry. Let
\[
\calL_{\rm GOE}(k)
=
\Big(
\ell_1({\rm GOE_k}),
\dots,
\ell_k({\rm GOE_k})
\Big)
\]
denote the ordered eigenvalues of ${\rm GOE}_k$. 

Similarly, denote the vector consisting of the leading $k$ eigenvalues of the proposed regularized $F$-matrix $\bF_{k\lambda}$ after scaling to be 
\[
\calL(k,\lambda)
= \frac{k}{n-1-k}
\Big(
\ell_1(\bF_{k\lambda}),
\dots,
\ell_k(\bF_{k\lambda})
\Big)^T.
\]

\begin{theorem}\label{thm:main_fix_k}
Fix $\lambda>0$. 
Suppose that Conditions~\ref{enum:high_dimension_regime}--\ref{enum:lower_bound_eigen_Y} hold and $k$ remains fixed, as $n\to\infty$. Recall $q = p_1/(n-1-k)$. Assume that $X$ and $Y$ are independent. As $n\to\infty$,
\[ \frac{\sqrt{p_1} }{\Omega_2(\lambda,q)} \Big(\calL(k,\lambda) - \Omega_1(\lambda,q)\mathbb{1}_k  \Big)   \stackrel{D}{\longrightarrow} \calL_{\rm GOE}(k).\]
The normalization parameters $\Omega_1(\lambda,q)$ and $\Omega_2(\lambda, q)$ are given by
\begin{equation}
\label{eq:def_Theta12}
\begin{split}
\Omega_1(\lambda,q)
&= \big[\lambda {\varphi}(-\lambda,q)\big]^{-1} - 1,\\
\Omega_2(\lambda,q)
&= \Big(2q
\frac{{\varphi}'(-\lambda,q)-[{\varphi}(-\lambda,q)]^2}
{\lambda^2[{\varphi}(-\lambda,q)]^4}\Big)^{1/2}.
\end{split}
\end{equation}
\end{theorem}
While our assumptions require the entries of $\bZ_x$ to possess moments of all orders, Theorem~\ref{thm:main_fix_k} is expected to remain valid under the weaker assumption of finite fourth moments. We do not pursue this technical extension here, as our primary objective is to develop a unified theoretical framework for both the trace-based and largest-root-based procedures.

Since $\Omega_1(\lambda,q)$ and $\Omega_2(\lambda,q)$ are smooth functions of $\varphi(-\lambda,q)$ and $\varphi'(-\lambda,q)$, consistent estimators can be constructed by replacing these quantities with $\hat{\varphi}(-\lambda)$ and $\hat{\varphi}'(-\lambda)$, respectively. We henceforth denote the resulting plug-in estimators by $\hat{\Omega}_1(\lambda)$ and $\hat{\Omega}_2(\lambda)$.

As a direct consequence of Theorem~\ref{thm:main_fix_k}, the asymptotic distribution of the regularized trace-based statistic can be obtained by an application of the delta method together with Slutsky's theorem.

\begin{theorem}
\label{thm:main_LRT}
Suppose that the conditions of Theorem~\ref{thm:main_fix_k} hold. Then,
\begin{equation}\label{eq:decision_normal}
\widetilde{\LRT}(k,\lambda) \coloneqq 
\frac{\sqrt{p_1}}{\sqrt{k}\hat{\Omega}_2(\lambda)} \Big(\LRT(k,\lambda)  - k \hat{\Omega}_1(\lambda) \Big)
\stackrel{D}{\longrightarrow}
N(0,1).
\end{equation}
\end{theorem}
In practice, we reject the null hypothesis at significance level $\alpha$ whenever $\widetilde{\LRT}(k,\lambda)>z_{1-\alpha}$, where $z_{1-\alpha}$ denotes the $(1-\alpha)$ quantile of the standard normal distribution.

\textcolor{black}{The simulation results in Section~\ref{subsec:simulation_null} suggest that, over a broad range of settings, the finite-sample distribution of $\widetilde{\LRT}(k,\lambda)$ is well approximated by the standard normal distribution, provided that $k$ is not excessively large relative to $p_1$. In particular, the approximation is found to be accurate when $k\lesssim20$. Consequently, the proposed test provides satisfactory control of the type-I error rate within this regime.
}

\begin{remark}\label{remark:compare_yang_pan}
We compare the proposed statistic $\LRT(k,\lambda)$ with the regularized LRT procedure of \cite{yang2015independence}. Their method is based on the eigenvalues of $n^{-1}\bX^T (\bS_x + t I)^{-1} \bX \bY^T (\bY\bY^T)^{-}\bY$, where $(\bY\bY^T)^{-}$ denotes the Moore--Penrose pseudoinverse of $\bY\bY^T$. Although both procedures employ ridge regularization, the eigenvalues of $\bF_{k\lambda}$ do not admit a simple algebraic relationship with those of $n^{-1}\bX^T (\bS_x + t I)^{-1} \bX \bY^T (\bY\bY^T)^{-}\bY$, except in the unregularized case $\lambda=t=0$. A key distinction is that the proposed method incorporates PC-based dimension reduction through the projection matrix $\bP_k$, whereas the procedure of \cite{yang2015independence} relies on the projection $\bY^T(\bY\bY^T)^{-}\bY$ onto the row space of $\bY$.

When $p_2>n$ and $\bY$ has full column rank, which holds with probability tending to one under mild conditions, we have $\bY^T(\bY\bY^T)^{-}\bY= I_n$. Consequently, the statistic of \cite{yang2015independence} no longer depends on the correlation between $\bX$ and $\bY$, and therefore cannot effectively capture the dependence structure. In contrast, by exploiting the leading principal components of $\bY$, the proposed statistic $\LRT(k,\lambda)$ remains informative and is expected to retain nontrivial power in this regime.

When $p_2<n$, the distinction is more subtle. The asymptotic theory of \cite{yang2015independence} corresponds to the regime in which $k=p_2$ grows proportionally with $n$. In contrast, our analysis establishes the asymptotic distribution of $\LRT(k,\lambda)$ when $k$ remains relatively small. In this regime, the normalization parameters $\Omega_1(\lambda,q)$ and $\Omega_2(\lambda,q)$ admit simple and consistent estimators, leading to a practically implementable testing procedure. By comparison, estimation of the normalization parameters required by \cite{yang2015independence} is considerably more challenging. In Section \ref{sec:simulation}, we numerically compare the power of these tests under representative settings.
\end{remark}

\subsection{Asymptotic distribution of the regularized largest-root statistic}\label{subsec:asymptotic_k_diverge}
In this section, we derive the asymptotic distribution of the proposed regularized largest-root statistic $\ell_{\max}(k,\lambda)$ under the regime in which the reduced dimension parameter $k$ grows proportionally with $n$.

Under Condition~\ref{enum:edger_regularity}, define
\[
h_0=
\begin{cases}
\eta,
& \text{if } g'(h)>0 \text{ for all } h\in(\eta-\epsilon_2,\eta),\\[6pt]
\text{the unique solution of } g'(h)=0,
& \text{if } g'(h)<0 \text{ for } h\in(\eta-\epsilon_2,\eta).
\end{cases}
\]
The uniqueness of the root follows from the strict concavity of $g(h)$. We further define
\begin{equation}
\label{eq:def_rho}
\rho=\lim_{h\uparrow h_0}g(h).
\end{equation}
By construction, the restricted mapping $g: (-\infty,h_0)\to(-\infty,\rho)$ is strictly increasing and therefore admits an inverse mapping, denoted by
\[
h(g):\ (-\infty,\rho)\to(-\infty,h_0).
\]

Next, define the function
\begin{equation}\label{eq:def_s_nlambda_g}
s(g)= \frac{1}{q} \Big( \frac{1}{g - h(g) } -1 \Big), \qquad g\in [0,~\rho).    
\end{equation}
The following lemma summarizes several key properties of $s(g)$. 
\begin{lemma}
\label{lemma:properties_s_fun}
Suppose that Conditions \ref{enum:boundedness_spectral_norm} and \ref{enum:edger_regularity} hold. For all sufficiently large $n$, the function $s(g)$ is analytic on $(0,\rho)$ and satisfies $s'(g) >0$, $s''(g) >0$, $g\in [0,\rho)$. Furthermore, $s'(g)\to\infty$, as $g\uparrow \rho$.
\end{lemma}


\begin{theorem}\label{thm:main_diverge_k}  
Fix $\lambda>0$.  Suppose that Conditions~\ref{enum:high_dimension_regime}--\ref{enum:edger_regularity} hold and that $X$ and $Y$ are independent. Assume that $k/n  \to \gamma \in (0,\alpha)$, $n\to\infty$, where $\alpha$ is defined in Condition~\ref{enum:lower_bound_eigen_Y}. Recall that $q = q(k,n) =p_1/(n-1-k)$. As $n\to\infty$,
\begin{equation}
\frac{p_1^{2/3}}{\Theta_{2}(\lambda,q)} \Big(\ell_{\max}(\lambda, k) - \Theta_{1}(\lambda,q)\Big)~~ \stackrel{D}{\longrightarrow} ~~\TW_1,  
\label{eq:Tracy_Widom}
\end{equation}
where \( \TW_1 \) denotes the Tracy-Widom distribution of type~1.  
The normalization parameters \( \Theta_{1}(\lambda,k)\) and \( \Theta_{2}(\lambda,k) \) are determined as follows. Let \( \beta\) to be the solution in \( (0, \rho) \) to  
\begin{equation}\label{eq:def_beta}
\beta^2 {s}'(\beta) = \frac{k}{p_1}.
\end{equation}  
The existence and uniqueness of $\beta$ follow from Lemma~\ref{lemma:properties_s_fun}. We then have
\begin{align}
&\Theta_{1}(\lambda,q) = \frac{1}{\beta} \Big[1 + \frac{p_1}{k} \beta {s}(\beta)\Big], \label{eq:def_Theta1} \\
&\Theta_{2}(\lambda,q) = \Bigg[\frac{(p_1/k)^3}{2} {s}''(\beta) + \frac{(p_1/k)^2}{\beta^3} \Bigg]^{1/3}. \label{eq:def_Theta2}
\end{align}  
\end{theorem}

\begin{remark}
\label{remark:Tracy_Widom}
For background on Tracy--Widom distributions, we refer readers to \cite{tracy1994fredholm,tracy1996orthogonal,johnstone2008multivariate}. Software implementations for numerical evaluation and simulation of Tracy--Widom laws are available in the R package \textit{RMTstat} \citep{johnstone2022package} and the MATLAB package \textit{RMLab} \citep{dieng2006matlab}.
\end{remark}

\begin{remark}
When $k$ remains fixed as $n\to\infty$, Theorem~\ref{thm:main_fix_k} implies that
\[
\frac{\sqrt{p_1}}{\Omega_2(\lambda,q)}
\Big(
\ell_{\max}(k,\lambda)
-
\Omega_1(\lambda,q)
\Big)
\stackrel{D}{\longrightarrow}
\ell_{\max}\big({\rm GOE}_k\big).
\]
Therefore, Theorem~\ref{thm:main_diverge_k} may be viewed as a large-$k$ extension of Theorem~\ref{thm:main_fix_k} when attention is restricted to the largest eigenvalue. A notable difference between the two regimes is the fluctuation scale. In the fixed-$k$ regime, $\ell_{\max}(k,\lambda)$ fluctuates on the order of $p_1^{-1/2}$, whereas in the diverging-$k$ regime of Theorem~\ref{thm:main_diverge_k}, the fluctuation scale becomes $p_1^{-2/3}$. The latter rate is characteristic of edge eigenvalue fluctuations and arises universally across a broad class of random matrix models. This transition is closely related to the fact that the largest eigenvalue of ${\rm GOE}_k$ converges to the Tracy--Widom distribution after appropriate centering and scaling as $k\to\infty$ \citep{johnstone2008multivariate}.
\end{remark}

\begin{remark}
When $\bY$ is treated as deterministic, Theorem~\ref{thm:main_diverge_k} specializes to the main theorem of \cite{li2025ridge}. Moreover, if $\lambda=0$ and $p_1<n-1-k$, it further reduces to the Tracy--Widom results of \cite{han2016tracy} and \cite{han2018unified}. Thus, Theorem~\ref{thm:main_diverge_k} unifies and extends these earlier results within a common framework that accommodates CCA, PC-based dimension reduction, and ridge regularization.
\end{remark}

To implement the proposed procedure, consistent estimators of $\Theta_1(\lambda,q)$ and $\Theta_2(\lambda,q)$ are required. To this end, we adapt the estimation procedure of \cite{li2025ridge} to the present setting. For ease of exposition, a detailed description of the estimation method is deferred to Section~\ref{sec:estimation} in the Appendix. The procedure relies solely on the eigenvalues of $\bW_{k2}$ and is computationally efficient.

Under mild regularity conditions, and provided that $k$ is not excessively close to $n$ (equivalently, $\gamma$ is not too large), the resulting estimators, denoted by $\hat{\Theta}_1(\lambda)$ and $\hat{\Theta}_2(\lambda)$, satisfy
\[
\hat{\Theta}_1(\lambda)-\Theta_1(\lambda,q)
=
o_p(p_1^{-2/3}),
\qquad
\hat{\Theta}_2(\lambda)-\Theta_2(\lambda,q)
=
o_p(1).
\]
See Section \ref{sec:estimation} for details. Consequently, replacing $\Theta_1(\lambda,q)$ and $\Theta_2(\lambda,q)$ by their estimators does not affect the Tracy--Widom limit in Theorem~\ref{thm:main_diverge_k}.

From a practical perspective, the numerical studies reported in Section~6.1 of \cite{li2025ridge} indicate that the proposed estimation procedure achieves high accuracy provided that $q=p_1/(n-1-k)$ is not excessively large (for example, $q\lesssim5$), although the estimation accuracy deteriorates as $q$ increases. We do not reproduce those numerical results here. Instead, in Section~\ref{subsec:simulation_null}, we directly evaluate the empirical sizes of the proposed procedures using the estimated normalization parameters $\hat{\Theta}_1(\lambda)$ and $\hat{\Theta}_2(\lambda)$. The simulation results demonstrate that the resulting tests maintain satisfactory control of the type-I error rate.

In practice, we reject the null hypothesis at significance level $\alpha$ whenever 
\begin{equation}\label{eq:decision_TW}
\widetilde{\ell}_{\max}(k,\lambda) \coloneqq \frac{p_1^{2/3}}{\hat\Theta_{2}(\lambda)} \Big(\ell_{\max}(\lambda, k) - \hat\Theta_{1}(\lambda)\Big) > \TW_{1,\, 1-\alpha},
\end{equation}
where $\TW_{1,\,1-\alpha}$ denotes the $(1-\alpha)$ quantile of $\TW_1$.


\section{Power Analysis}\label{sec:power_analysis}
In this section, we investigate the asymptotic power properties of the proposed regularized tests under the alternative hypothesis. Recall that under Model~\eqref{eq:model_PCR}, the dependence between $X$ and $Y$ is entirely characterized by the signal component $M\Lambda_m^T Y$. Throughout this section, we assume that this dependence can be fully captured asymptotically by the reduced principal eigenspace in the sense that there exists a sequence of integers $k_0=k_0(n)$ satisfying $(n-k_0) \asymp n$ and
\begin{equation}
\label{eq:consistency_eq1}
\left\|
n^{-1/2}
M\Lambda_m^T\bY(I_n-\bP_{k_0})
\right\|_2
=o_p(1).
\end{equation}
The condition is reasonable to assume in a broad range of applications. Heuristically, it is expected to hold under any of the three scenarios: 
\begin{itemize}
\item[(i)] The reduced dimension $k_0$ is chosen sufficiently large. In the extreme case where $k_0 = p_2 <n$, we have $\bY(I_n-\bP_{k_0}) = 0$, so the condition is satisfied automatically.
\item[(ii)] The spectrum of $\Sigma_y$ decays sufficiently rapidly so that the eigenvalues beyond the $k_0$-th are negligible. For example, the condition is satisfied if $\|M\|_2 = O(1)$ and there exists $k_0 < \min(p_2, n)$ such that $\ell_{k_0}(\bS_y) \stackrel{P}{\longrightarrow} 0$. 
\item[(iii)] The dependence between $X$ and $Y$ is primarily driven by the existence of a small number of common dominant factors. Such a structural assumption is natural in many application areas, including economics, finance, genomics, and psychometrics.
\end{itemize}
We next present a representative example illustrating Scenario~(iii).
\begin{remark}
\label{remark:sufficient_condition_factor_model}
Technically, suppose that $m$ is fixed and that
\[\ell_{1}(\Sigma_y)\ge\cdots\ge\ell_{m}(\Sigma_y)\ge c\log n,\qquad\ell_{m+1}(\Sigma_y)\le C,\]
for some positive constants $c$ and $C$. For simplicity, assume further that $Z_y$ is Gaussian. Under this model, the leading principal components $\Lambda_m^T Y$ may be interpreted as the dominant latent factors driving the dependence between $X$ and $Y$. If $\|M\|_2=O(1)$, Condition~\eqref{eq:consistency_eq1} is satisfied whenever $k_0\ge m$.
\end{remark}

We establish general sufficient conditions for the consistency of the proposed procedures. Define
\[ \mathcal S_{nk} = \frac1n M\Lambda_m^T \bY \bP_k \bY^T \Lambda_m M^T. \]
\begin{theorem}[Consistency of the trace-based test]
\label{thm:consistency_trace_based}
Fix $\lambda>0$. Suppose that Conditions~\ref{enum:high_dimension_regime}--\ref{enum:lower_bound_eigen_Y} hold. Assume that Condition~\eqref{eq:consistency_eq1} is satisfied for some fixed $k_0$, and that the chosen reduced dimension satisfies $k\ge k_0$ and remains fixed as $n\to\infty$. Suppose that the signal under $H_a$ is sufficiently large in the sense that for every $\xi \in(0,1)$, there exist constants $\varepsilon_\xi>0$ and $n_\xi$ such that  
\begin{equation}\label{eq:condition_on_S_nk}
\mP \Big( \|\mathcal S_{nk} \|_2 \ge\varepsilon_\xi \log(n)\Big) \ge \xi, \qquad n\ge n_\xi. 
\end{equation}
Then the proposed trace-based test is asymptotically consistent. Specifically,
\[
\mathbb P\!\big(
\widetilde{\mathcal T}(k,\lambda)>z_{1-\alpha}
\big)
\longrightarrow
1,
\]
as $n\to\infty$, for every significance level $\alpha\in(0,1)$.
\end{theorem}

\begin{theorem}[Consistency of the largest-root test]\label{thm:consistency_largest_root}
    Fix $\lambda>0$. Suppose that Conditions \ref{enum:high_dimension_regime}--\ref{enum:edger_regularity} hold and that $k/n\to \gamma \in (0,\alpha)$, as $n\to\infty$. Assume that Condition \ref{eq:consistency_eq1} is satisfied for some sequence $k_0 = k_0(n)$ with $k_0 \leq k$. Suppose further that the signal satisfies Condition \eqref{eq:condition_on_S_nk} in Theorem~\ref{thm:consistency_trace_based} and the estimated normalization parameters $\hat{\Theta}_1(\lambda)$ and $\hat{\Theta}_2(\lambda)$ are consistent. Then, the proposed largest-root test is asymptotically consistent. Specifically, 
    \[\mP\! \big(\tilde{\ell}_{\max}(k,\lambda)  > \TW_{1,\, 1-\alpha} \big) \longrightarrow 1,\]
    as $n\to \infty$, for every significance level $\alpha \in (0,1)$.
 \end{theorem}
\noindent For the consistency of $\hat{\Theta}_1(\lambda)$ and $\hat{\Theta}_2(\lambda)$ under $H_a$, we refer to Lemma \ref{lemma:consistency_under_Ha} in the Appendix.


While Theorems~\ref{thm:consistency_trace_based} and \ref{thm:consistency_largest_root} provide useful qualitative insight into the proposed procedures, they offer limited understanding of how the test statistics depend on the covariance structures of $X$ and $Y$, the cross-covariance between the two random vectors, and the regularization parameter $\lambda$. In particular, understanding the role of $\lambda$ is essential for developing a principled data-driven procedure for selecting the regularization parameter.

To gain further insight into these effects, we consider a simplified rank-one alternative with a dominating signal. Specifically, we assume that $M\Lambda_m^T$ has rank one. Let
\[ M\Lambda_m^T \Sigma_y^{1/2} = d_n  \mu_n \nu_n^T,\]
be its singular value decomposition, where $d_n>0$ denotes the nonzero singular value, and $\mu_n\in\mathbb R^{p_1}$ and $\nu_n\in\mathbb R^{p_2}$ are the corresponding unit left and right singular vectors, respectively. We further assume that the signal dominates the noise in the sense that $\tr(\Sigma_0) /d_n^2 \to 0$. 

Under this parametrization, the model can be written as
\begin{equation}\label{eq:rank_one_model}
X = d_n\mu_n\nu_n^T Z_y + \Sigma_0^{1/2} Z_0, \qquad Y = \Sigma_y^{1/2}Z_y.
\end{equation}
Thus, the dependence between $X$ and $Y$ is concentrated along a single pair of directions, with signal strength controlled by $d_n$.

Under the assumed conditions, it can be shown that 
\begin{equation}\label{eq:characterize_trace_rankone}
\frac{1}{d_n^2} \mathcal{T}(k,\lambda) = \frac{1}{k}\big(\nu_n^T \bZ_y \bP_k \bZ_y^T \nu_n\big)  \mu_n^T \calD(-\lambda)\mu_n + o_p(1),
\end{equation}
and
\begin{equation}\label{eq:characterize_largest_rankone}
\frac{1}{d_n^2} \ell_{\max}(k,\lambda) =  \frac{1}{k}\big(\nu_n^T \bZ_y \bP_k \bZ_y^T \nu_n\big)  \mu_n^T \calD(-\lambda)\mu_n + o_p(1), 
\end{equation}
where $\mathcal D(-\lambda)$ is the deterministic equivalent of the ridge-regularized inverse $(\bW_{k2}+\lambda I_{p_1})^{-1}$ defined in \eqref{eq:deteministic_equivalent}. Here, under the rank-one alternative, the trace and largest eigenvalue of $\bF_{k\lambda}$ are asymptotically equivalent after scaling, which explains why the two statistics admit the same leading-order characterization. 

This characterization provides useful insight into the roles of the covariance structures of $X$ and $Y$, the cross-covariance between the two random vectors, and the tuning parameters $k$ and $\lambda$. In particular, the overall signal strength is governed by the singular value $d_n$, which reflects the interaction between the dependence component $M\Lambda_m^T$ and the covariance structure of $Y$. The effect of $k$ is manifested through both the rank of the projection $\bP_k$ and the aspect ratio $q$ appearing in $\mathcal D(-\lambda)$. Finally, the effects of $\Sigma_0$ and $\lambda$ are encapsulated by the quadratic form
\[
\mu_n^T\mathcal D(-\lambda)\mu_n,
\]
which measures how the ridge regularization interacts with the covariance structure of $X$ along the signal direction.

\section{Selection of the Regularization Parameter}\label{sec:selection_regularization_parameter}

We now consider the selection of the regularization parameter $\lambda$. Throughout this section, we assume that a candidate parameter space $\mathcal L$ is prescribed, from which the regularization parameter is selected. Recommended choices of $\mathcal L$ are discussed later in this section.

Heuristically, for a given alternative characterized by $M\Lambda_m^T$, a natural strategy is to choose $\lambda$ so as to maximize the power of the proposed test. This, however, requires understanding how the power depends on $\lambda$. In general settings, such dependence is highly intricate and is not directly amenable to analysis. 

By contrast, under the rank-one alternative considered in Section~\ref{sec:power_analysis}, the asymptotic behavior of the proposed statistics admits an explicit characterization. In particular, the influence of $\lambda$ is entirely captured by the quadratic form $\mu_n^T\mathcal D(-\lambda)\mu_n$. Combined with the normalizations in \eqref{eq:decision_normal} and \eqref{eq:decision_TW}, this characterization suggests that the dominant effect of $\lambda$ on the asymptotic power is governed by
\[
\SNR(\lambda,q)
\coloneqq
\frac{
\mu_n^T
\mathcal D(-\lambda)
\mu_n
}{
\Xi(\lambda,q)
},
\]
where $\Xi(\lambda,q)=\Omega_2(\lambda,q)$ for the trace-based test, and $\Xi(\lambda,q)=\Theta_2(\lambda,q)$ for the largest-root test. Accordingly, $\SNR(\lambda,q)$ serves as a proxy for the asymptotic signal-to-noise ratio of the proposed procedures.

Motivated by this characterization, we use $\SNR(\lambda,q)$ as a surrogate objective for selecting the regularization parameter in practice, without explicitly requiring the rank-one structural assumption on the alternative to hold. Accordingly, we choose $\lambda$ by maximizing $\SNR(\lambda,q)$ over the prescribed candidate parameter space $\mathcal L$.


The remaining challenge lies in the empirical treatment of $\SNR(\lambda,q)$, since it depends on the unknown population quantities $\mu_n$ and $\Sigma_0$. To begin with, Result~(ii) of Lemma~\ref{lemma:determinist_equivalent} suggests that, if $\mu_n$ were known, then $\SNR(\lambda,q)$ could be estimated by
\[
\frac{
\mu_n^T
(\bW_{k2}+\lambda I_{p_1})^{-1}
\mu_n
}{
\hat{\Xi}(\lambda)
},
\]
where $\hat{\Xi}(\lambda)$ denotes either $\hat{\Omega}_2(\lambda)$ or $\hat{\Theta}_2(\lambda)$, depending on the decision rule employed. 

It remains to characterize $\mu_n$ in practice. A potential approach is to estimate $\mu_n$ using the leading eigenvector of an estimator of
$\Sigma_{xy}\Sigma_y^{-1}\Sigma_{xy}^T$. 
However, consistency of such an estimator is typically available only under relatively strong structural assumptions. Moreover, even when consistency can be established, the resulting estimator generally remains statistically dependent on $\bW_{k2}$. This dependence substantially complicates the theoretical analysis and makes it difficult to establish consistency of the resulting estimator of $\SNR(\lambda,q)$.

In this paper, rather than attempting to estimate $\mu_n$ directly, we adopt a Bayesian decision-theoretic perspective and select $\lambda$ according to Bayes and minimax principles. Specifically, we place a prior distribution on $\mu_n$ and consider the corresponding prior expectation of $\SNR(\lambda,q)$, rather than the quantity $\SNR(\lambda,q)$ itself. We then aggregate over a class of such prior models through a minimax criterion.

Specifically, for any $\Pi = (\pi_0, \pi_1, \dots, \pi_\theta)$ with $\pi_i \in [0,\infty)$ for $i=0,\dots,\theta$, we consider the prior model on $\mu_n$, denoted by $\mathcal{P}(\Pi)$, defined through
\[
\mu_n = \frac{\tilde{\mu}}{\|\tilde{\mu}\|_2},
\qquad \text{where} \qquad
\tilde{\mu} \sim \mathcal{N}\Big(0, \sum_{i=0}^{\theta} \pi_i \Sigma_0^{i} \Big).
\]
Note that the induced prior distribution of $\mu_n$ is invariant under a common scaling of $\pi_i$'s, namely $\pi_i \mapsto c \pi_i$, $i=0,1,\dots,\theta$, for any $c>0$. Therefore, without loss of generality, we impose the normalization condition 
$\sum_{i=0}^\theta \pi_i = 1$.

We define the Bayes risk function by
\[
\mathfrak{R}(\lambda,\Pi)
=
\limsup_{n\to\infty}
-
\mathbb{E}_{\mathcal{P}(\Pi)}
\big[
\SNR(\lambda,q)
\big] = - \liminf_{n\to\infty} \mE_{\calP(\Pi)}[\SNR(\lambda,q)],
\]
where $\mathbb{E}_{\mathcal{P}(\Pi)}$ denotes expectation with respect to the prior distribution $\mathcal{P}(\Pi)$. 

While the following procedure can be generalized to any finite value of $\theta$, in this paper, we focus on the case of $\theta = 1$, which leads to a family of priors whose covariance is linear in $\Sigma_0$.  
The proposed prior is motivated by two considerations. First, note that
$\mu_n^T\mathcal{D}(-\lambda)\mu_n$
is primarily determined by the projection of $\mu_n$ onto the eigendirections of $\mathcal{D}(-\lambda)$, which coincide with those of $\Sigma_0$. The proposed prior model leads to a natural interpretation of how $\mu_n$ aligns with the eigenspaces of $\Sigma_0$ and provides explicit control over the expected projection of $\mu_n$ along each eigen direction.  
For example, when $\pi_0=1$ and $\pi_i=0$ for all $i\geq1$, the resulting prior assigns equal weight to all directions and therefore induces a uniform distribution over the unit sphere. Conversely, larger values of $\pi_i$, $i\geq 1$ associated with higher powers of $\Sigma_0$ place greater prior mass on directions aligned with the leading eigenvectors of $\Sigma_0$. 
Second, the polynomial covariance structure yields tractable expressions and estimators for the Bayes risk that depend only on the eigenvalues of $\bW_{k2}$. The estimator is especially stable when $\theta = 1$. This substantially simplifies both computation and implementation. 

In particular, under the prior model $\mathcal{P}(\Pi)$,
\begin{align*}
\mathfrak{R}(\lambda,\Pi)
&=
\limsup_{n\to\infty}
\frac{
-\mathbb{E}_{\mathcal{P}(\Pi)}
\mu_n^T
\mathcal{D}(-\lambda)
\mu_n
}{
\Xi(\lambda,q)
}
\\
&=
\limsup_{n\to\infty}
\frac{-
\sum_{i=0}^{\theta}
\pi_i
\Upsilon_i(\lambda)
}{
\Xi(\lambda,q)
\sum_{i=0}^{\theta}
\pi_i
\mathfrak{M}_i
},
\end{align*}
where
\[
\Upsilon_i(\lambda)
=
p_1^{-1}
\tr\!\big[
\mathcal{D}(-\lambda)\Sigma_0^i
\big],
\qquad
\mathfrak{M}_i
=
p_1^{-1}
\tr\!\big[
\Sigma_0^i
\big], \qquad i\geq 0.
\]

The involved quantities can be estimated as follows. First, $\Xi(\lambda,q)$ is estimated by $\hat{\Omega}_2(\lambda)$ or $\hat{\Theta}_2(\lambda)$, depending on the test statistic in use. The resulting estimator is denoted by $\hat{\Xi}(\lambda)$. 

Second, while clearly $\mathfrak{M}_0=1$, we consistently estimate $\mathfrak{M}_1$ by
\[
\hat{\mathfrak{M}}_1
=
p_1^{-1}\tr[\bW_{k2}].
\]
For $\mathfrak{M}_i$ with $i\geq2$, we refer to Lemma~1 of \cite{bai2010estimation} for the construction of consistent estimators, denoted by $\hat{\mathfrak{M}}_i$. Since our primary focus is the case $\theta=1$, further details are omitted.

Third, the quantities $\Upsilon_i(\lambda)$ are estimated  by  $\hat{\Upsilon}_i(\lambda)$, defined through the following recursively formula
\begin{equation}
\label{eq:def_Upsilon}
\begin{split}
\hat{\Upsilon}_{i+1}(\lambda)
&=
\frac{
1
}{
\lambda\hat{\varphi}(-\lambda)
}
\Big\{
\hat{\mathfrak{M}}_i
-
\lambda
\hat{\Upsilon}_i(\lambda)
\Big\},
\qquad
i=0,1,2,\dots,
\\
\hat{\Upsilon}_0(\lambda)
&=
q^{-1}
\Big[
\hat{\varphi}(-\lambda)
-
(1-q)/\lambda
\Big].
\end{split}
\end{equation}
Consistency of these estimators was established in \cite{li2020adaptable}.


It follows that the Bayes risk function can be estimated by
\[
\hat{\mathfrak{R}}(\lambda,\Pi)
=
\frac{-
\sum_{i=0}^{\theta}
\pi_i
\hat{\Upsilon}_i(\lambda)
}{
\hat{\Xi}(\lambda)
\sum_{i=0}^{\theta}
\pi_i
\hat{\mathfrak{M}}_i
}.
\]
\begin{definition}
\label{def:Bayes_choice}
For a given prior specification $\Pi$ and parameter space $\calL$, we define the \emph{Bayes selection} of the regularization parameter under $\mathcal{P}(\Pi)$ as
\[
\lambda_B(\Pi)
=
\arg\min_{\lambda\in\calL}
\mathfrak{R}(\lambda,\Pi).
\]
\end{definition}

Next, we collect all admissible prior specifications into the class
\[
\mathfrak{P}
=
\Bigg\{
\Pi=(\pi_0,\pi_1,\dots,\pi_\theta)
~:~
\pi_j\geq0,
\quad
\sum_{j=0}^{\theta}\pi_j=1
\Bigg\},
\]
where the recommended choice is $\theta=1$.

\begin{definition}
\label{def:minimax_choice}
Given the parameter space $\calL$, we define the \emph{minimax selection} of the regularization parameter over the class $\mathfrak{P}$ as
\[
\lambda_*
=
\arg\min_{\lambda\in\calL}
\sup_{\Pi\in\mathfrak{P}}
\mathfrak{R}(\lambda,\Pi).
\]
\end{definition}

In practice, the above procedure is implemented by replacing $\mathfrak{R}(\lambda,\Pi)$ with its estimator $\hat{\mathfrak{R}}(\lambda,\Pi)$. The resulting choices are denoted by $\hat{\lambda}_B(\Pi)$ and $\hat{\lambda}_*$, respectively.

While the proposed strategy provides a systematic and principled approach for selecting the regularization parameter in a data-driven manner, it also raises potential concerns regarding a ``double-dipping'' phenomenon. Indeed, the same dataset is used both for selecting $\lambda$ and for conducting the final test. In particular, although Theorems~\ref{thm:main_fix_k} and \ref{thm:main_diverge_k} establish the asymptotic distributions of the proposed statistics for any fixed $\lambda$, the additional randomness introduced through the data-driven selection of $\lambda$ may affect the limiting behavior of the resulting test statistic.

To mitigate the variability of $\hat{\lambda}_B$ and $\hat{\lambda}_*$, we recommend restricting the admissible parameter space $\mathcal L$ to a relatively coarse discrete grid over a reasonable range. We suggest choosing $\mathcal L$ as a set of $10$ evenly spaced points ranging from
\[
\bigl((0.1q)\wedge1\bigr)\hat{\mathfrak M}_1
\quad\text{to}\quad
2\hat{\mathfrak M}_1.
\]
Our simulation studies indicate that the performance of the proposed procedures is relatively insensitive to the choice of $\lambda$ within this range. Consequently, discretizing the parameter space in this manner is not expected to result in a substantial loss of efficiency.

The following lemma shows that, provided the candidate parameter space $\mathcal L$ is finite and the population Bayes or minimax regularization parameter is asymptotically identifiable, the corresponding data-driven selection is consistent. Consequently, replacing $\lambda$ by the selected regularization parameter does not affect the asymptotic null distributions established in Theorems~\ref{thm:main_fix_k} and \ref{thm:main_diverge_k}.

\begin{lemma}
    \label{lemma:consistency_Bayes_lambda}
    Suppose that Conditions~\ref{enum:high_dimension_regime}--\ref{enum:edger_regularity} hold. Assume that either $H_0$ holds, or $H_a$ holds and Condition~\eqref{eq:consistency_eq1} is satisfied for some sequence $k_0$ with $k\ge k_0$. Further, assume that
\[
\sup_{\lambda\in\mathcal L}
\big|
\hat{\Xi}(\lambda)-\Xi(\lambda,q)
\big|
\stackrel{P}{\longrightarrow}
0,
\]
where $\mathcal L$ is a finite candidate parameter set.
    \begin{itemize}
        \item[(i)] Fix a prior model $\calP(\Pi)$. Suppose that there exists $\lambda_\infty \in \calL$ and $\epsilon>0 $ such that for all sufficiently large $n$
        \[ \min_{\lambda\in \calL\setminus \{\lambda_\infty\} } \mathfrak{R}(\lambda, \Pi) > \mathfrak{R}(\lambda_\infty, \Pi) + \epsilon. \]
        Then, $\mP(\hat{\lambda}_B(\Pi) = \lambda_\infty ) \longrightarrow 1$. Consequently, the conclusion of Theorem \ref{thm:main_fix_k} or Theorem \ref{thm:main_diverge_k} remain valid when $\lambda$ is replaced by $\hat{\lambda}_B(\Pi)$. 
        \item[(ii)] Fix $\mathfrak{P}$. Suppose that there exists $\lambda_\infty \in \calL$ and $\epsilon >0$ such that for all sufficiently large $n$
        \[ \min_{\lambda\in\calL\setminus\{\lambda_\infty\}} \sup_{\Pi\in \mathfrak{P} } \mathfrak{R} (\lambda, \Pi) > \sup_{\Pi\in \mathfrak{P}} \mathfrak{R}(\lambda_\infty, \Pi) + \epsilon. \]
        Then, $\mP(\hat{\lambda}_*  = \lambda_\infty) \longrightarrow 1$. Consequently,  the conclusion of Theorem \ref{thm:main_fix_k} or Theorem \ref{thm:main_diverge_k} remain valid when $\lambda$ is replaced by $\hat{\lambda}_*$. 
    \end{itemize}
\end{lemma}

\section{Practical Guidance}\label{sec:practical_guidance}  
In this section, we discuss several practical considerations for implementing the proposed methodology. A fundamental feature of the model in \eqref{eq:model_PCR} is that the roles of $X$ and $Y$ are asymmetric. In general, interchanging their roles leads to different testing procedures and may yield different conclusions. Consequently, when two sets of variables are jointly observed, one must first determine which random vector should be designated as $X$ and which as $Y$. 

To address this issue, we propose the following data-driven rule based on the spectral characteristics of the corresponding sample covariance matrices. Factor structures are common in applications such as economics, finance, and psychometrics, where PC-based dimension reduction is particularly effective. Accordingly, if at least one sample covariance matrix exhibits isolated spikes, indicating the presence of strong latent factors, we designate as $Y$ the random vector whose sample covariance matrix contains the larger number of spikes. Otherwise, if neither covariance matrix exhibits isolated spikes, or if both contain the same number of spikes, we designate as $Y$ the random vector with the larger dimension.

The second practical issue concerns the selection of the reduced dimension parameter $k$. Recall that the power analysis under $H_a$ relies on the assumption that the dependence between $X$ and $Y$ can be asymptotically captured by the leading principal eigenspace of $Y$. Since the intrinsic reduced dimension $m$ is generally unknown, we recommend a sequential, data-driven selection procedure.

Specifically, we begin with a relatively small value of $k$ and examine the change in the spectrum of $\bW_{k2}$ when $k$ is increased to $k+1$. We continue increasing $k$ until the spectrum of $\bW_{k2}$ becomes sufficiently stable. Motivated by our numerical experiments, we recommend the stopping criterion
\[ \Big| \tr(\bW_{k+1,\,2}) - \tr(\bW_{k2}) \Big| \leq 0.05\times \frac{1}{p_1} \tr(\bW_{k2}). \]
In practice, we recommend restricting the search to
\[ \min(0.01n,p_2) \le k \le \min(0.5n,p_2).\]

\section{Simulation Studies}\label{sec:simulation}
In this section, we evaluate the finite-sample performance of the proposed procedures and compare them with several representative methods from the existing literature through Monte Carlo experiments.

The proposed procedures are implemented under the following configurations. We consider three choices of the regularization parameter: $\lambda=0.5$, $\lambda=1$, and the data-driven minimax choice $\hat{\lambda}_*$ developed in Section~\ref{sec:selection_regularization_parameter} over the recommended candidate parameter space. To examine the effect of the reduced dimension parameter, we consider four values of $k$, corresponding to $k/p_2\in\{0.01,0.05,0.10,0.30\}$, thereby covering both the small-$k$ and large-$k$ regimes. Under the null hypothesis, the normalization parameters $\Omega_1$, $\Omega_2$, $\Theta_1$, and $\Theta_2$ are replaced by their proposed estimators. In particular, $\Theta_1$ and $\Theta_2$ are estimated using the recommended configuration described in Section~\ref{sec:estimation}.

We compare the proposed procedures with the following CCA-based methods from the existing literature.
\begin{itemize}
\item[(a)] {\bf YP-LRT.} The regularized likelihood ratio test proposed by \cite{yang2015independence}. See Remark~\ref{remark:compare_yang_pan} for a discussion of its relationship to the proposed procedures. This method is applicable only when $p_2<n$.
\item[(b)] {\bf HPY-Largest.} The calibrated version of the classical Roy's largest-root test proposed by \cite{han2018unified} for settings in which $p_1$ and $p_2$ are comparable to $n$. This method is applicable only when $p_1+p_2<n$.

\end{itemize}
While the critical values for {\bf YP-LRT} are difficult to estimate accurately, we use Monte Carlo critical values for all methods to facilitate a fair comparison of empirical power. Consequently, all procedures are calibrated to have the same nominal significance level.

\subsection{Settings}\label{subsec:simulation_setting}
Three distributions are considered for the entries of $Z_x$ and $Z_y$: the standard normal distribution, a standardized $t$-distribution with six degrees of freedom, and a standardized Poisson distribution. Since the results are nearly identical across the three distributions, only those under the Gaussian setting are reported. The dimensions are taken as $p_1,p_2\in\{100,200\}$ with sample size $n\in\{200,400\}$. For the covariance matrix of $X$, we consider
$\Sigma_0\in
\left\{
\Sigma_{\rm ID}(p_1),\,
\Sigma_{\rm Poly}(p_1),\,
\Sigma_{\rm AR}(p_1)
\right\}
$
where the three covariance models are:
\begin{itemize}
\item[(i)] {The identity model:} $\Sigma_{\rm ID}(p)=I_p$;
\item[(ii)] {Polynomial spectral decay model:}  $\Sigma_{\rm Poly}(p) \propto\operatorname{Diag}(\tau_1,\ldots,\tau_p)$, where
$\tau_j = (1+j/p)^{-6}$; 
\item[(iii)] {AR(1)-autocovariance model:} $\Sigma_{\rm AR}(p)\propto \big(0.6^{|i-j|}\big)_{i,j=1}^p$.
\end{itemize}
In each case, the matrix is normalized so that $\tr(\Sigma(p)) =p$.
For the covariance matrix $\Sigma_y$ of $Y$, we generate a spiked covariance model whose spectrum consists of the spectrum of $\Sigma_{\rm Poly}(p_2-5)$ together with five additional spikes. The spike eigenvalues are evenly spaced between $2\tau_{\max}$ and $1.2\tau_{\max}$, where $\tau_{\max}$ denotes the largest eigenvalue of $\Sigma_{\rm Poly}(p_2-5)$.

Under the alternative hypothesis, for each combination of $p_1$, $p_2$, $\Sigma_y$, and $\Sigma_0$, we generate the dependence structure according to
\[
M\Lambda_m^T\Sigma_y^{1/2}
=
SS \times
D_{\rm Haar}
\operatorname{Diag}(d_1,\ldots,d_m)
\Lambda_m^T,
\]
where $SS>0$ is a scalar controlling the signal strength, and, for each Monte Carlo replication, $D_{\rm Haar}\in\mathbb R^{p_1\times m}$ is a random orthonormal matrix drawn according to the Haar measure. We consider the following two configurations for the singular values $\{d_j\}$:
\begin{itemize}
\item[(i)] (Low-rank correlation) We set $m=5$ and choose $d_j\propto j^{-2}$, $j=1,\ldots,5$, followed by normalization so that $\sum_{j=1}^{5}d_j^2=1$.
\item[(ii)] (Exp-Decay correlation) We set $m=\min(p_1,p_2)$ and assume that the squared singular values satisfy $d_j^2\propto0.8^{\,j-1}$, $j=1,\ldots,m$. The singular values are then normalized so that $\sum_{j=1}^{m}d_j^2=1.$
\end{itemize}

\subsection{Empirical sizes}\label{subsec:simulation_null}

Table~\ref{tab:emp_sizes} reports the empirical sizes of the proposed procedures under the investigated settings at the nominal significance level of $5\%$. The results indicate that the trace-based procedure $\LRT(k,\lambda)$ provides satisfactory control of the type-I error rate when the reduced dimension parameter $k$ is relatively small, with the approximation becoming increasingly accurate as the sample size $n$ increases. As expected, the empirical sizes become progressively inflated as $k$ increases. 

For the largest-root procedure $\ell_{\max}(k,\lambda)$, the tests are slightly conservative when $k$ is small, with empirical sizes typically ranging from $1\%$ to $3\%$. As $k$ increases, the conservativeness is reduced, and the empirical sizes become increasingly close to the nominal level. For example, when $k/p_2=0.3$, the empirical sizes are approximately $4\%$ across the simulation settings.

The observed finite-sample behavior is consistent with the asymptotic regimes underlying the proposed approximations. Theorem~\ref{thm:main_fix_k} derives the asymptotic approximation to the finite-sample distribution of $\LRT(k,\lambda)$ under the assumption that $k$ remains fixed, whereas Theorem~\ref{thm:main_diverge_k} derives the corresponding approximation for $\ell_{\max}(k,\lambda)$ under a diverging-$k$ regime. The simulation results confirm that these approximations are effective in the finite-sample settings for which they are intended.

Combining these findings with the theoretical results, our numerical studies suggest using the trace-based procedure $\LRT(k,\lambda)$ when the reduced dimension is relatively small, for example when $k\lesssim 20$, and using the largest-root procedure $\ell_{\max}(k,\lambda)$ when $k$ is larger. This recommendation is consistent with the asymptotic regimes under which the normal and Tracy--Widom approximations are derived, respectively.

For both procedures, the empirical sizes obtained with fixed values of $\lambda$ are very similar to those obtained with the data-driven minimax choice of $\lambda$. This suggests that the additional variability introduced by the data-driven selection of the regularization parameter is negligible in terms of size control.

\begin{table}[htbp]
\centering
\small
\caption{Empirical sizes ($\times 100\%$) of the proposed procedures at the $5\%$ significance level. The innovation vectors $Z_x$ and $Z_y$ follow the standard normal distribution. The normalization parameters are estimated using the proposed estimators $\hat{\Omega}_1$, $\hat{\Omega}_2$, $\hat{\Theta}_1$, and $\hat{\Theta}_2$.}
\label{tab:emp_sizes}
\resizebox{\linewidth}{!}{
\begin{tabular}{lrrr|ccc|ccc|ccc|ccc}
\toprule
 &  &  & \multicolumn{1}{c}{} & \multicolumn{6}{c}{$n=200$} & \multicolumn{6}{c}{$n=400$} \\
 &  &  &\multicolumn{1}{c}{}  & \multicolumn{3}{c}{Trace-Based} & \multicolumn{3}{c}{Largest-Root} & \multicolumn{3}{c}{Trace-Based} & \multicolumn{3}{c}{Largest-Root} \\  \cmidrule(lr){5-7}\cmidrule(lr){8-10} \cmidrule(lr){11-13} \cmidrule(lr){14-16}
$\Sigma_0$ & $p_1$ & $p_2$ & \multicolumn{1}{c|}{$k/p_2$} & $\lambda=0.5$ & $\lambda=1$ & $\lambda=\hat\lambda_*$ & $\lambda=0.5$ & $\lambda=1$ & $\lambda=\hat\lambda_*$ & $\lambda=0.5$ & $\lambda=1$ & $\lambda=\hat\lambda_*$ & $\lambda=0.5$ & $\lambda=1$ & $\lambda=\hat\lambda_*$ \\
\midrule
\multirow{20}{*}{ID} & \multirow{10}{*}{100} & \multirow{5}{*}{100} & 0.01 & 5.86 & 5.86 & 5.68 & 2.02 & 2.12 & 2.18 & 6.22 & 6.20 & 5.80 & 2.12 & 2.06 & 2.16 \\
 &  &  & 0.05 & 6.56 & 6.70 & 6.46 & 3.02 & 2.98 & 3.14 & 5.80 & 5.68 & 5.44 & 3.44 & 3.42 & 3.32 \\
 &  &  & 0.1 & 6.64 & 6.58 & 6.20 & 3.94 & 3.98 & 3.86 & 5.90 & 5.76 & 5.80 & 3.72 & 4.02 & 3.84 \\
 &  &  & 0.2 & 7.96 & 7.24 & 6.94 & 4.36 & 4.36 & 4.28 & 6.92 & 6.34 & 6.46 & 3.76 & 3.72 & 3.72 \\
 &  &  & 0.3 & 8.82 & 8.18 & 7.66 & 4.80 & 4.52 & 4.46 & 6.50 & 6.36 & 6.14 & 4.16 & 4.22 & 4.24 \\
  \cline{3-16}
 &  & \multirow{5}{*}{200} & 0.01 & 6.90 & 6.64 & 6.48 & 2.72 & 2.62 & 2.72 & 6.22 & 6.16 & 5.96 & 2.88 & 2.94 & 3.04 \\
 &  &  & 0.05 & 6.84 & 6.48 & 6.20 & 3.48 & 3.42 & 3.24 & 6.14 & 5.98 & 5.70 & 3.70 & 3.62 & 3.56 \\
 &  &  & 0.1 & 8.20 & 7.58 & 7.06 & 4.36 & 4.30 & 4.08 & 6.18 & 6.10 & 5.86 & 4.14 & 4.20 & 3.96 \\
 &  &  & 0.2 & 10.12 & 9.32 & 8.36 & 5.10 & 4.82 & 4.48 & 7.12 & 6.68 & 6.60 & 4.30 & 4.30 & 4.26 \\
 &  &  & 0.3 & 13.38 & 11.00 & 9.88 & 5.16 & 4.74 & 4.42 & 8.32 & 7.74 & 7.40 & 4.48 & 4.48 & 4.86 \\
  \cline{2-16}
 & \multirow{10}{*}{200} & \multirow{5}{*}{100} & 0.01 & 6.28 & 5.92 & 5.64 & 1.98 & 2.04 & 1.86 & 5.92 & 5.78 & 5.86 & 1.84 & 1.80 & 1.68 \\
 &  &  & 0.05 & 5.94 & 5.62 & 5.50 & 3.56 & 3.58 & 3.54 & 6.08 & 5.96 & 5.68 & 3.30 & 3.24 & 3.50 \\
 &  &  & 0.1 & 6.60 & 6.08 & 5.92 & 3.60 & 3.64 & 3.68 & 5.90 & 5.74 & 5.38 & 3.52 & 3.46 & 3.76 \\
 &  &  & 0.2 & 7.92 & 7.20 & 6.86 & 3.86 & 4.06 & 3.92 & 6.28 & 5.72 & 5.74 & 4.54 & 4.84 & 4.58 \\
 &  &  & 0.3 & 9.52 & 8.62 & 7.64 & 4.88 & 4.88 & 4.60 & 7.18 & 6.84 & 6.42 & 4.30 & 4.08 & 4.16 \\
  \cline{3-16}
 &  & \multirow{5}{*}{200} & 0.01 & 5.56 & 5.80 & 5.78 & 2.68 & 2.64 & 2.54 & 5.70 & 5.60 & 5.40 & 2.60 & 2.62 & 2.42 \\
 &  &  & 0.05 & 6.66 & 6.48 & 6.04 & 3.56 & 3.26 & 3.38 & 6.22 & 5.96 & 6.00 & 3.96 & 3.82 & 3.54 \\
 &  &  & 0.1 & 8.30 & 7.44 & 7.08 & 4.42 & 4.10 & 3.90 & 6.42 & 6.12 & 5.86 & 4.42 & 4.56 & 4.58 \\
 &  &  & 0.2 & 10.94 & 9.74 & 8.72 & 4.98 & 5.18 & 4.92 & 7.58 & 7.08 & 6.58 & 4.24 & 4.12 & 3.96 \\
 &  &  & 0.3 & 14.06 & 12.46 & 11.44 & 5.04 & 4.84 & 4.52 & 8.60 & 7.44 & 6.96 & 4.14 & 4.12 & 4.04 \\
  \cline{1-16}
\multirow{20}{*}{AR} & \multirow{10}{*}{100} & \multirow{5}{*}{100} & 0.01 & 6.64 & 6.58 & 6.52 & 1.98 & 2.24 & 2.12 & 6.62 & 6.66 & 6.48 & 2.06 & 2.22 & 2.10 \\
 &  &  & 0.05 & 6.38 & 6.28 & 6.46 & 3.08 & 2.96 & 3.20 & 5.58 & 5.58 & 5.56 & 2.56 & 2.64 & 2.62 \\
 &  &  & 0.1 & 7.04 & 6.78 & 7.08 & 3.36 & 3.42 & 3.34 & 5.50 & 5.52 & 5.54 & 2.92 & 3.02 & 3.28 \\
 &  &  & 0.2 & 7.18 & 7.02 & 7.28 & 3.88 & 3.66 & 4.04 & 6.00 & 5.62 & 5.86 & 3.56 & 3.42 & 3.60 \\
 &  &  & 0.3 & 9.20 & 8.52 & 9.18 & 4.18 & 4.00 & 4.18 & 6.98 & 6.82 & 6.92 & 4.00 & 4.06 & 4.26 \\
  \cline{3-16}
 &  & \multirow{5}{*}{200} & 0.01 & 6.06 & 5.78 & 6.00 & 2.06 & 2.14 & 2.10 & 5.96 & 6.10 & 5.72 & 2.20 & 2.12 & 2.40 \\
 &  &  & 0.05 & 6.00 & 5.86 & 6.26 & 3.18 & 3.30 & 3.26 & 5.68 & 5.38 & 5.66 & 3.00 & 3.16 & 2.98 \\
 &  &  & 0.1 & 7.34 & 6.98 & 7.40 & 3.96 & 4.02 & 3.96 & 6.88 & 6.80 & 7.10 & 3.86 & 4.16 & 3.90 \\
 &  &  & 0.2 & 10.06 & 9.36 & 10.12 & 4.30 & 4.02 & 4.32 & 6.70 & 6.56 & 6.96 & 3.78 & 3.62 & 3.98 \\
 &  &  & 0.3 & 12.20 & 11.00 & 12.24 & 4.78 & 4.22 & 4.62 & 7.68 & 7.56 & 7.98 & 3.94 & 4.10 & 4.02 \\
  \cline{2-16}
 & \multirow{10}{*}{200} & \multirow{5}{*}{100} & 0.01 & 5.78 & 5.90 & 5.90 & 1.84 & 1.82 & 1.82 & 6.00 & 6.02 & 6.32 & 1.90 & 1.72 & 1.98 \\
 &  &  & 0.05 & 6.72 & 6.26 & 6.62 & 2.68 & 2.62 & 2.72 & 5.84 & 5.70 & 5.92 & 3.14 & 3.12 & 3.28 \\
 &  &  & 0.1 & 6.42 & 6.28 & 6.30 & 2.94 & 3.12 & 3.00 & 6.36 & 6.18 & 6.38 & 3.52 & 3.38 & 3.62 \\
 &  &  & 0.2 & 7.98 & 7.22 & 7.40 & 4.00 & 4.00 & 3.98 & 6.54 & 6.32 & 6.66 & 3.96 & 3.76 & 4.04 \\
 &  &  & 0.3 & 9.16 & 8.56 & 8.80 & 4.10 & 3.88 & 3.86 & 6.52 & 6.34 & 6.60 & 3.76 & 3.72 & 3.68 \\
  \cline{3-16}
 &  & \multirow{5}{*}{200} & 0.01 & 6.04 & 6.04 & 6.10 & 2.54 & 2.52 & 2.66 & 6.36 & 6.38 & 6.30 & 2.76 & 2.68 & 2.90 \\
 &  &  & 0.05 & 6.84 & 6.36 & 6.66 & 3.88 & 3.82 & 3.72 & 6.16 & 6.14 & 6.26 & 3.32 & 3.44 & 3.64 \\
 &  &  & 0.1 & 8.16 & 7.54 & 7.80 & 3.62 & 3.66 & 3.86 & 6.46 & 6.30 & 6.56 & 3.54 & 3.74 & 3.60 \\
 &  &  & 0.2 & 10.54 & 9.54 & 9.86 & 4.02 & 3.80 & 3.58 & 7.20 & 6.86 & 7.18 & 4.30 & 4.34 & 4.26 \\
 &  &  & 0.3 & 13.38 & 11.86 & 11.84 & 4.98 & 4.72 & 4.24 & 8.54 & 8.00 & 8.54 & 4.56 & 4.44 & 4.64 \\
  \cline{1-16}
\multirow{20}{*}{Poly} & \multirow{10}{*}{100} & \multirow{5}{*}{100} & 0.01 & 5.90 & 5.84 & 6.02 & 1.50 & 1.58 & 1.80 & 6.38 & 6.16 & 6.24 & 1.92 & 1.74 & 1.88 \\
 &  &  & 0.05 & 6.52 & 6.46 & 6.58 & 3.06 & 3.02 & 3.26 & 6.00 & 6.08 & 5.78 & 3.50 & 3.28 & 3.40 \\
 &  &  & 0.1 & 6.10 & 6.58 & 6.44 & 3.32 & 3.48 & 3.68 & 5.82 & 5.80 & 5.94 & 3.30 & 3.10 & 3.44 \\
 &  &  & 0.2 & 7.54 & 7.30 & 7.84 & 3.62 & 3.56 & 3.52 & 6.22 & 6.16 & 6.16 & 3.50 & 3.32 & 3.88 \\
 &  &  & 0.3 & 9.28 & 8.84 & 9.90 & 4.42 & 4.66 & 4.16 & 7.08 & 6.86 & 7.26 & 3.58 & 3.48 & 3.66 \\
  \cline{3-16}
 &  & \multirow{5}{*}{200} & 0.01 & 6.72 & 6.84 & 6.68 & 2.04 & 2.04 & 2.22 & 5.34 & 5.18 & 5.54 & 1.86 & 1.76 & 1.98 \\
 &  &  & 0.05 & 6.98 & 6.74 & 7.18 & 3.54 & 3.42 & 3.70 & 6.48 & 6.28 & 6.24 & 3.68 & 3.60 & 3.60 \\
 &  &  & 0.1 & 7.44 & 7.04 & 7.48 & 3.50 & 3.98 & 3.54 & 5.78 & 5.62 & 5.86 & 3.74 & 3.50 & 4.08 \\
 &  &  & 0.2 & 9.04 & 8.34 & 9.36 & 3.96 & 4.22 & 3.72 & 7.04 & 6.62 & 7.22 & 3.62 & 3.32 & 4.02 \\
 &  &  & 0.3 & 12.94 & 12.08 & 13.94 & 4.74 & 4.58 & 4.58 & 7.28 & 7.24 & 7.74 & 3.88 & 4.00 & 3.76 \\
  \cline{2-16}
 & \multirow{10}{*}{200} & \multirow{5}{*}{100} & 0.01 & 6.66 & 6.58 & 6.66 & 2.10 & 2.02 & 2.10 & 6.46 & 6.72 & 6.22 & 1.70 & 1.56 & 1.82 \\
 &  &  & 0.05 & 6.22 & 6.06 & 6.18 & 2.92 & 2.80 & 2.80 & 5.24 & 5.14 & 5.24 & 2.58 & 2.66 & 2.72 \\
 &  &  & 0.1 & 7.00 & 6.68 & 6.76 & 3.70 & 3.76 & 3.74 & 5.52 & 5.44 & 5.84 & 3.46 & 3.46 & 3.40 \\
 &  &  & 0.2 & 7.42 & 7.42 & 7.42 & 3.70 & 3.70 & 3.70 & 6.50 & 6.56 & 6.82 & 3.52 & 3.64 & 3.54 \\
 &  &  & 0.3 & 9.14 & 8.88 & 8.92 & 4.06 & 4.04 & 3.90 & 7.28 & 6.98 & 7.56 & 4.16 & 3.80 & 4.48 \\
  \cline{3-16}
 &  & \multirow{5}{*}{200} & 0.01 & 6.02 & 6.24 & 6.16 & 2.34 & 2.22 & 2.20 & 5.98 & 5.98 & 5.84 & 2.26 & 2.26 & 2.46 \\
 &  &  & 0.05 & 7.18 & 6.72 & 7.04 & 3.86 & 3.96 & 4.00 & 5.62 & 5.56 & 5.64 & 3.44 & 3.52 & 3.38 \\
 &  &  & 0.1 & 8.26 & 7.82 & 8.16 & 4.02 & 4.18 & 4.20 & 6.16 & 5.94 & 6.22 & 3.80 & 3.92 & 3.90 \\
 &  &  & 0.2 & 11.20 & 9.68 & 9.96 & 3.90 & 4.52 & 4.40 & 7.68 & 7.32 & 7.90 & 4.30 & 4.36 & 4.46 \\
 &  &  & 0.3 & 14.20 & 12.68 & 12.70 & 4.42 & 4.10 & 3.96 & 8.18 & 7.80 & 8.94 & 4.10 & 3.86 & 3.94 \\
\bottomrule
\end{tabular}}
\end{table}

\subsection{Empirical power}\label{subsec:simulation_power}

Figures~\ref{fig:power_poly_exp_decay}--\ref{fig:power_ID_exp_decay} show the size-adjusted empirical power as a function of the signal strength parameter $SS$ under a representative subset of the simulation settings. The remaining settings are presented in Figures~\ref{fig:power_AR_exp_decay}--\ref{fig:power_AR_low_rank} in Appendix~\ref{sec:addition_simulation_result}. To improve readability, we display only a representative subset of the combinations of $k$ and $\lambda$, thereby avoiding excessive overlap among the curves. Consistent with the practical recommendations in Section~\ref{sec:selection_regularization_parameter}, only the results obtained using the data-driven minimax choice $\lambda=\hat{\lambda}_*$ are reported. For the trace-based procedure $\LRT(k,\lambda)$, which is recommended when the reduced dimension is relatively small, we report the results for $k_1=0.01p_2$ and $k_2=0.05p_2$. As shown in Table~\ref{tab:emp_sizes}, these choices of $k$ also provide satisfactory control of the type-I error rate. For the largest-root procedure $\ell_{\max}(k,\lambda)$, which is recommended for larger values of $k$, we report the results for $k_3=0.10p_2$ and $k_4=0.30p_2$.

Several conclusions can be drawn from the simulation results. First, the proposed trace-based procedure exhibits the highest empirical power across all simulation settings considered. The largest-root procedure performs comparably under the Low-rank correlation model, but is generally less powerful under the Exp-Decay correlation model. This behavior is expected. Under the low-rank model, the signal is concentrated in a few leading eigendirections, allowing the largest-root statistic to capture most of the dependence through its dominant eigenvalue. By contrast, under the Exp-Decay model, the signal is distributed across many eigendirections. In this case, the trace-based statistic aggregates information over multiple directions, whereas the largest-root statistic utilizes only the leading one, resulting in a loss of power.

\begin{figure}[htbp]
    \centering
    \begin{subfigure}[t]{0.23\linewidth}
        \centering
        \includegraphics[width=\linewidth]{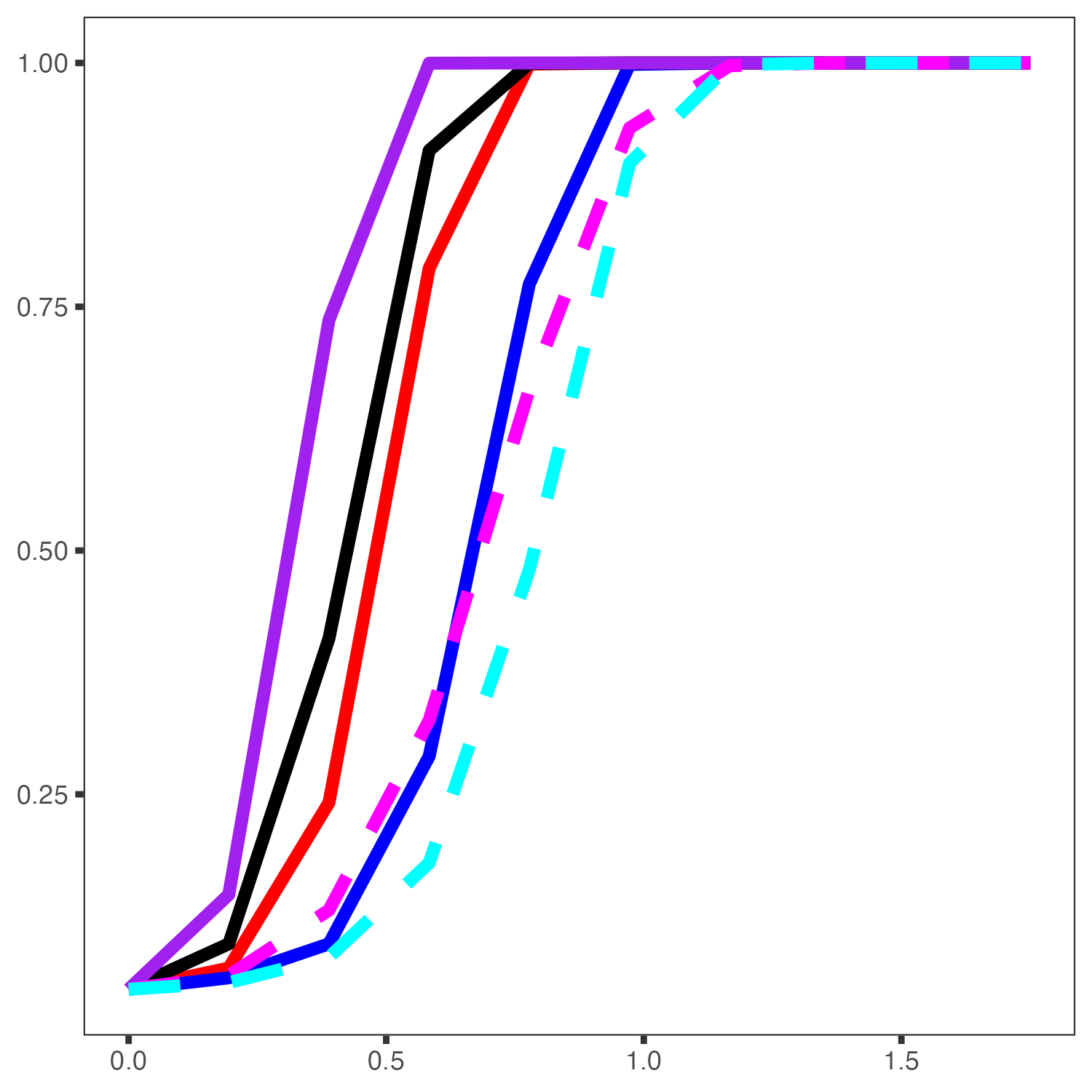}
    \end{subfigure}%
    \begin{subfigure}[t]{0.23\linewidth}
        \centering
        \includegraphics[width=\linewidth]{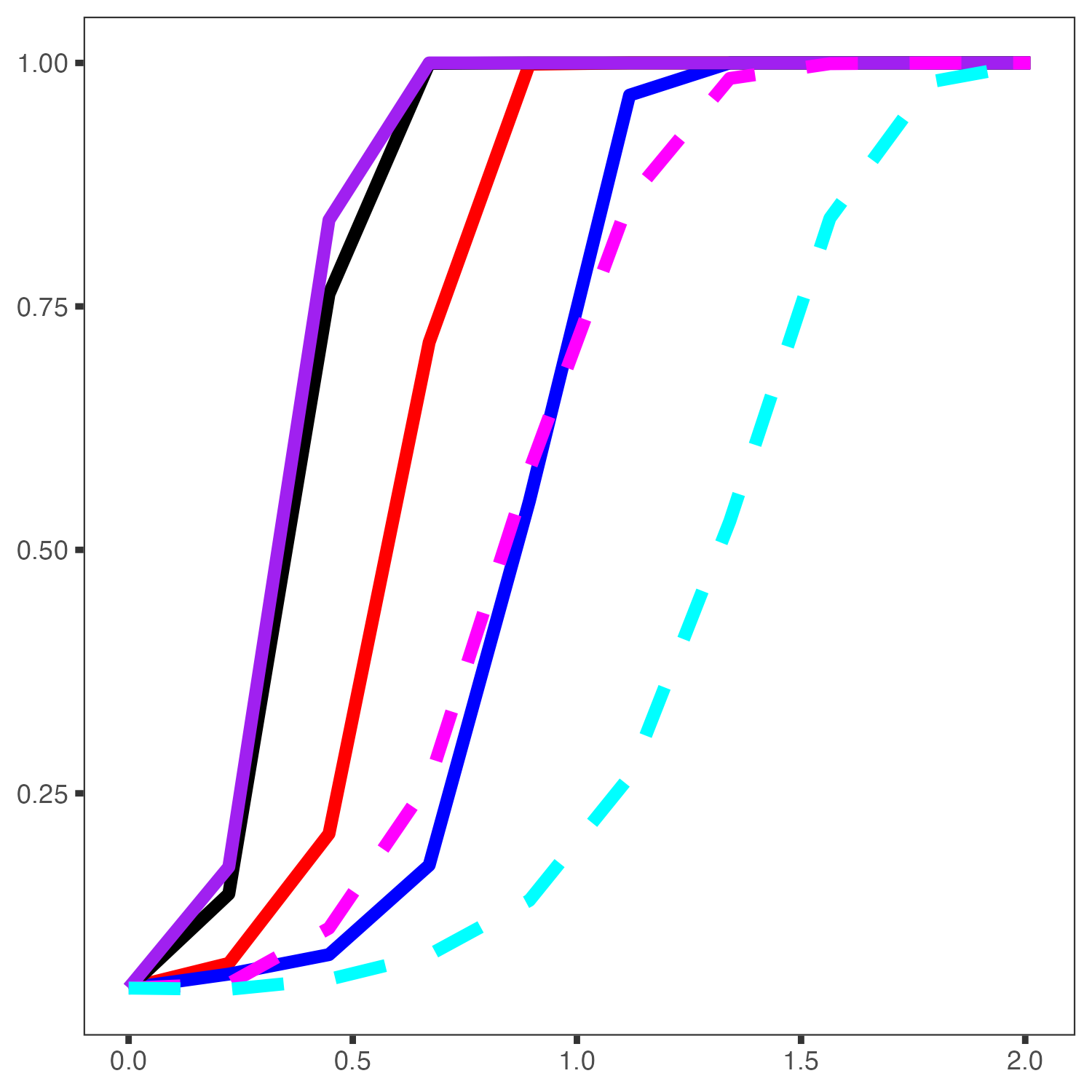}
    \end{subfigure}
    \begin{subfigure}[t]{0.23\linewidth}
        \centering
        \includegraphics[width=\linewidth]{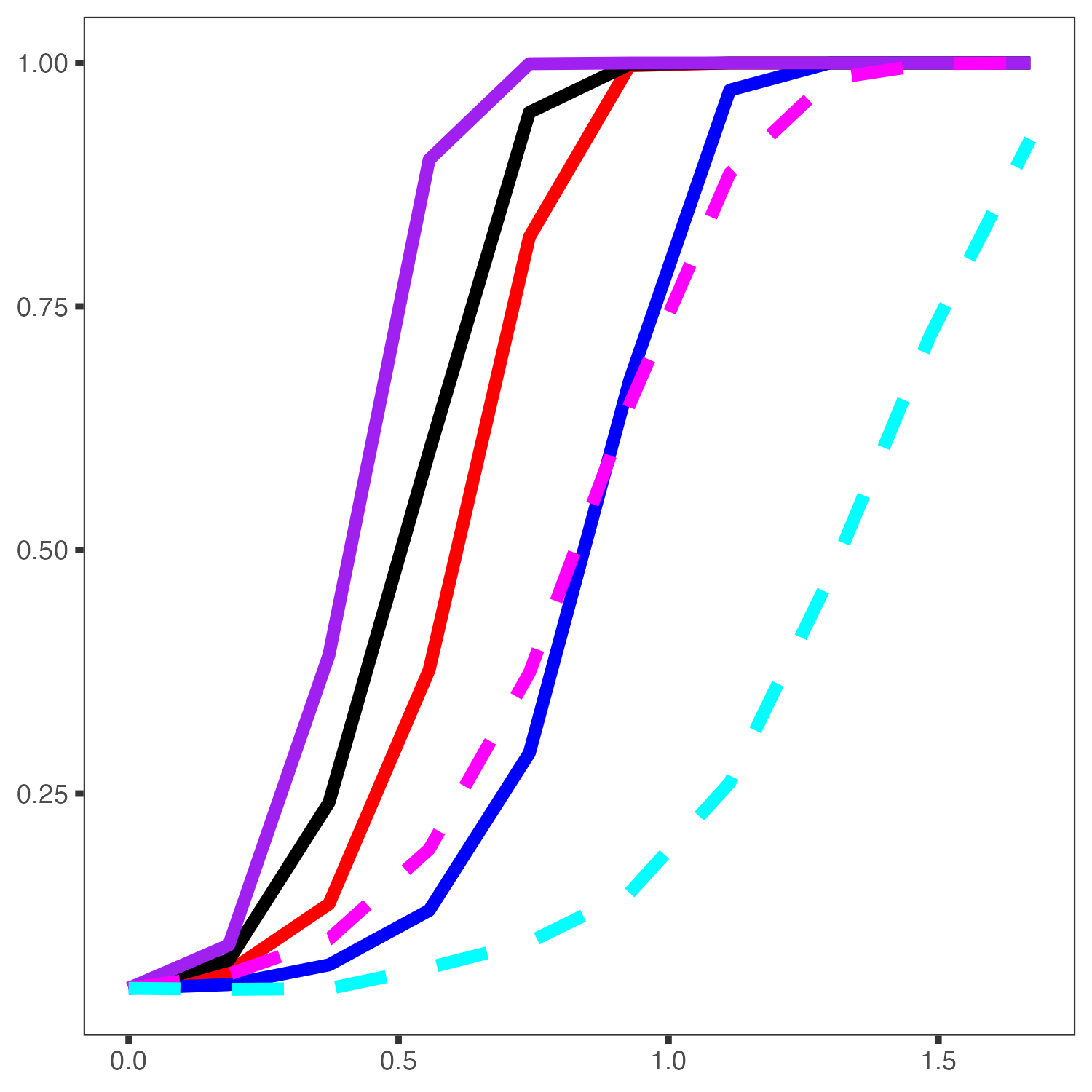}
    \end{subfigure}
    \begin{subfigure}[t]{0.23\linewidth}
        \centering
        \includegraphics[width=\linewidth]{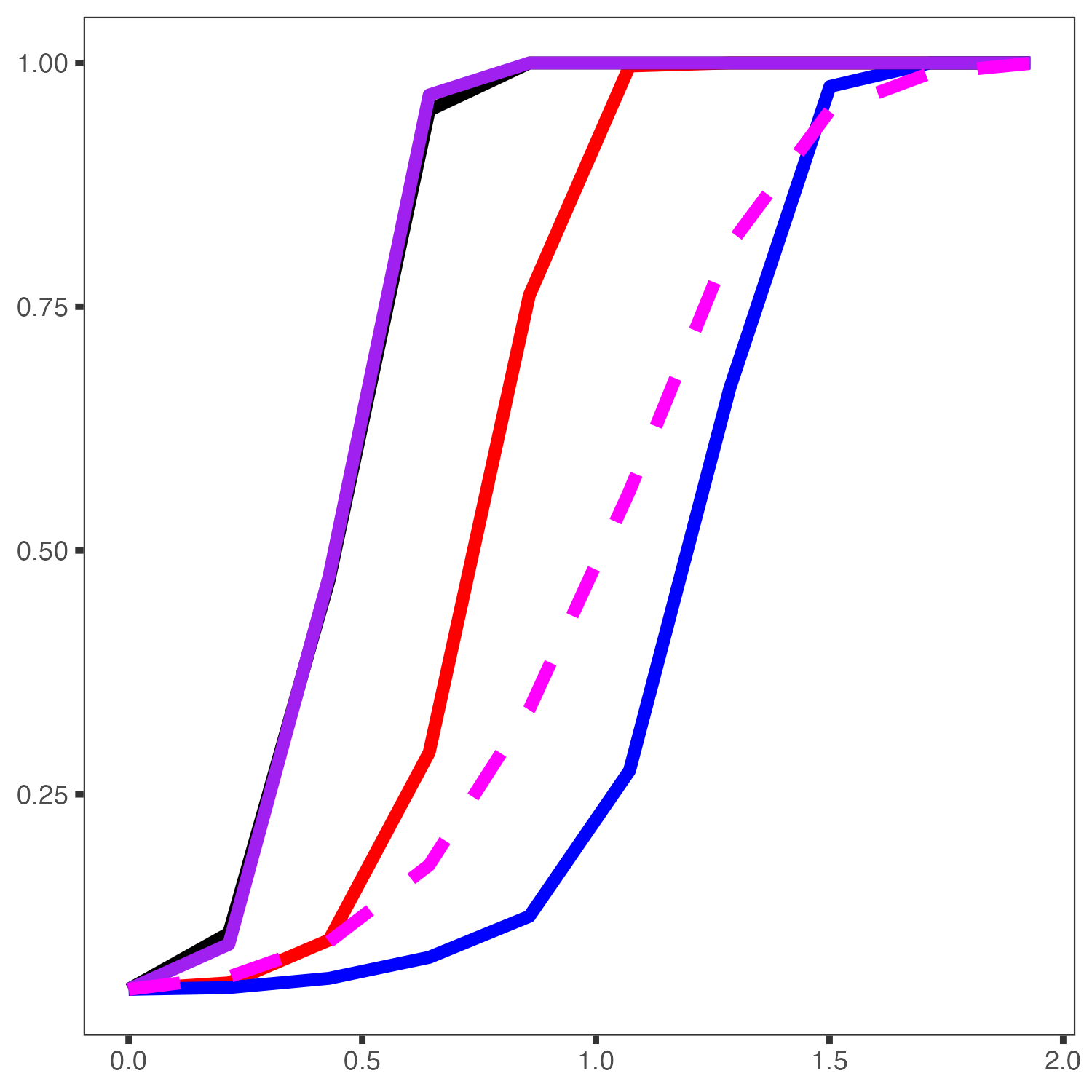}
    \end{subfigure}
    \vfill
    \begin{subfigure}[t]{0.23\linewidth}
        \centering
        \includegraphics[width=\linewidth]{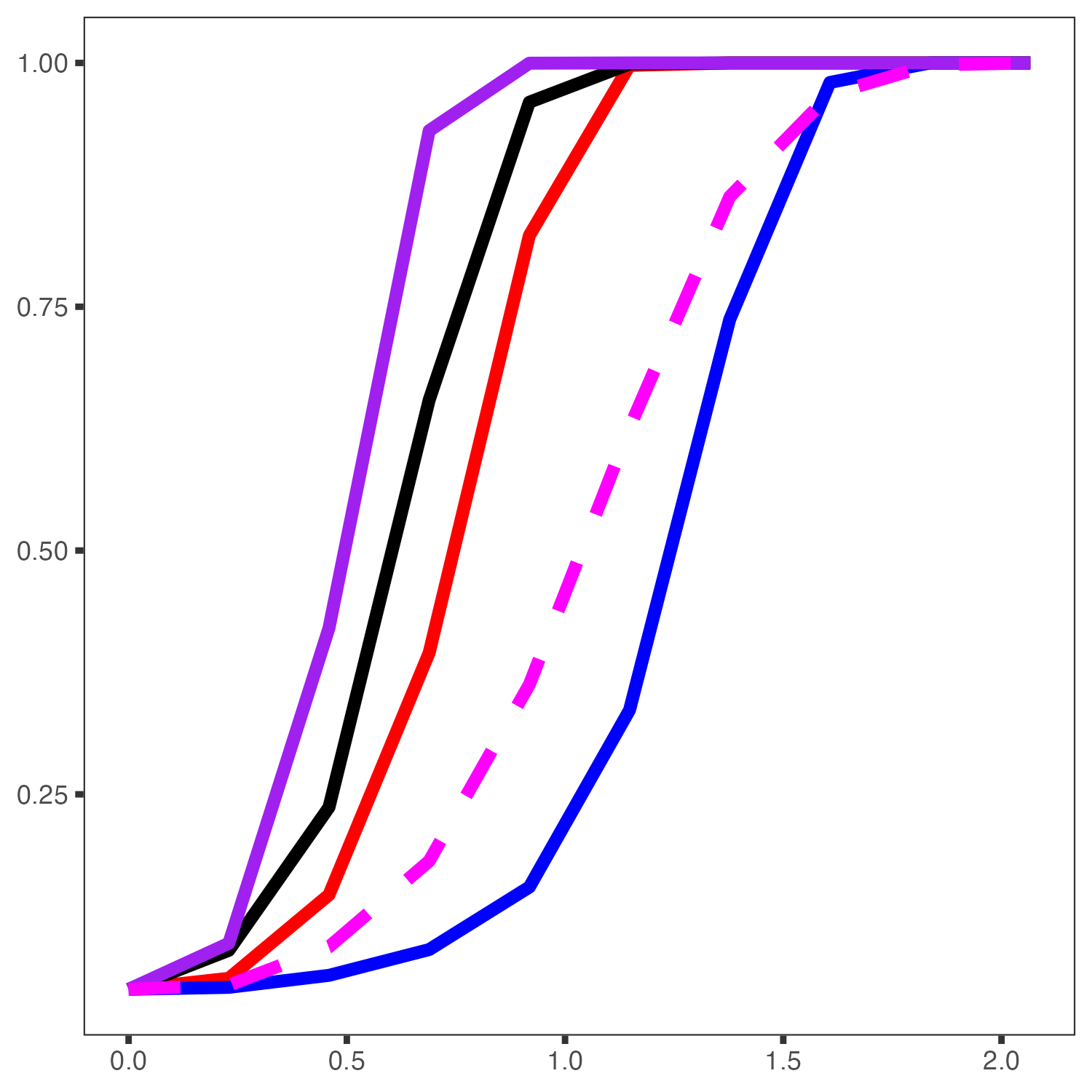}
    \end{subfigure}%
    \begin{subfigure}[t]{0.23\linewidth}
        \centering
        \includegraphics[width=\linewidth]{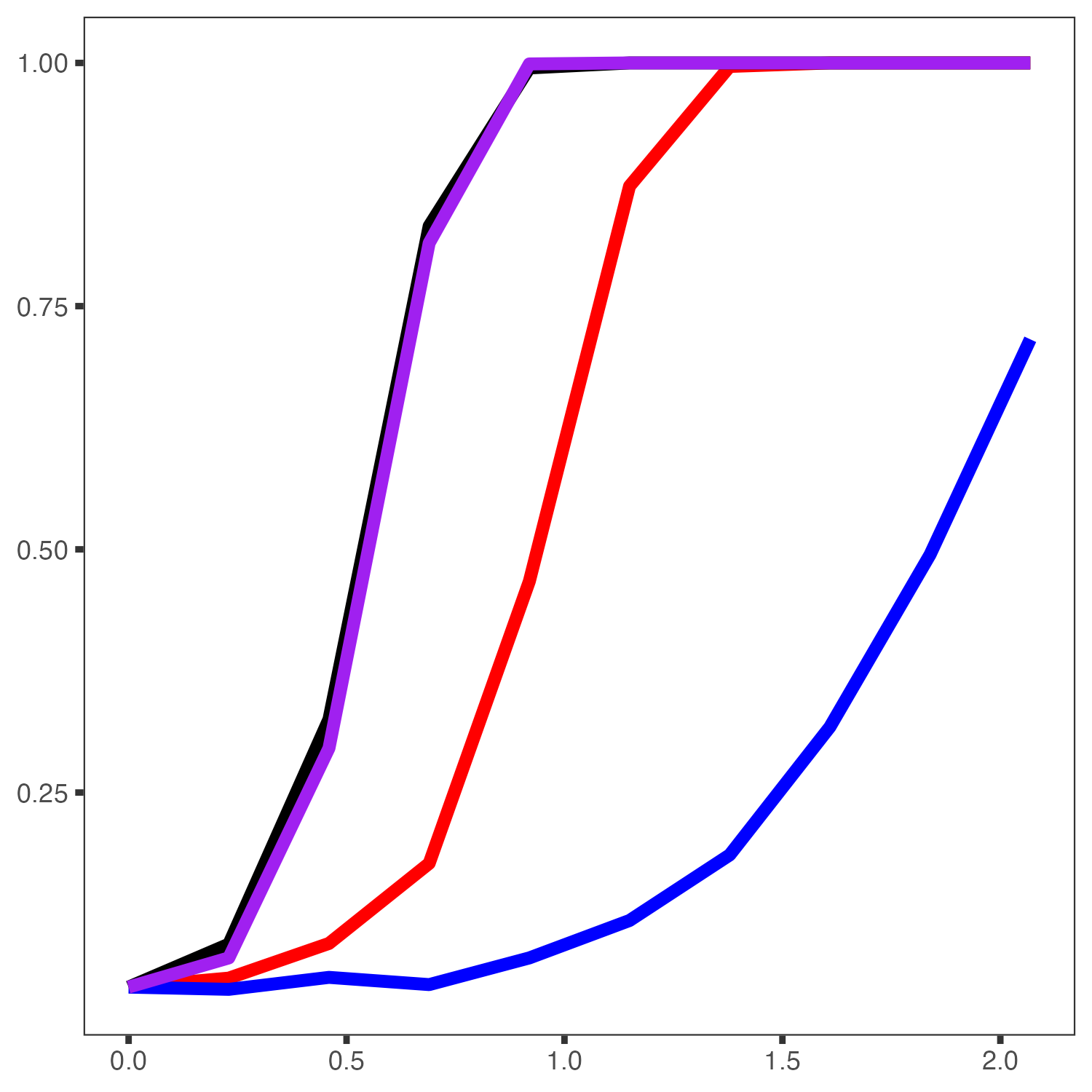}
    \end{subfigure}
    \begin{subfigure}[t]{0.23\linewidth}
        \centering
        \includegraphics[width=\linewidth]{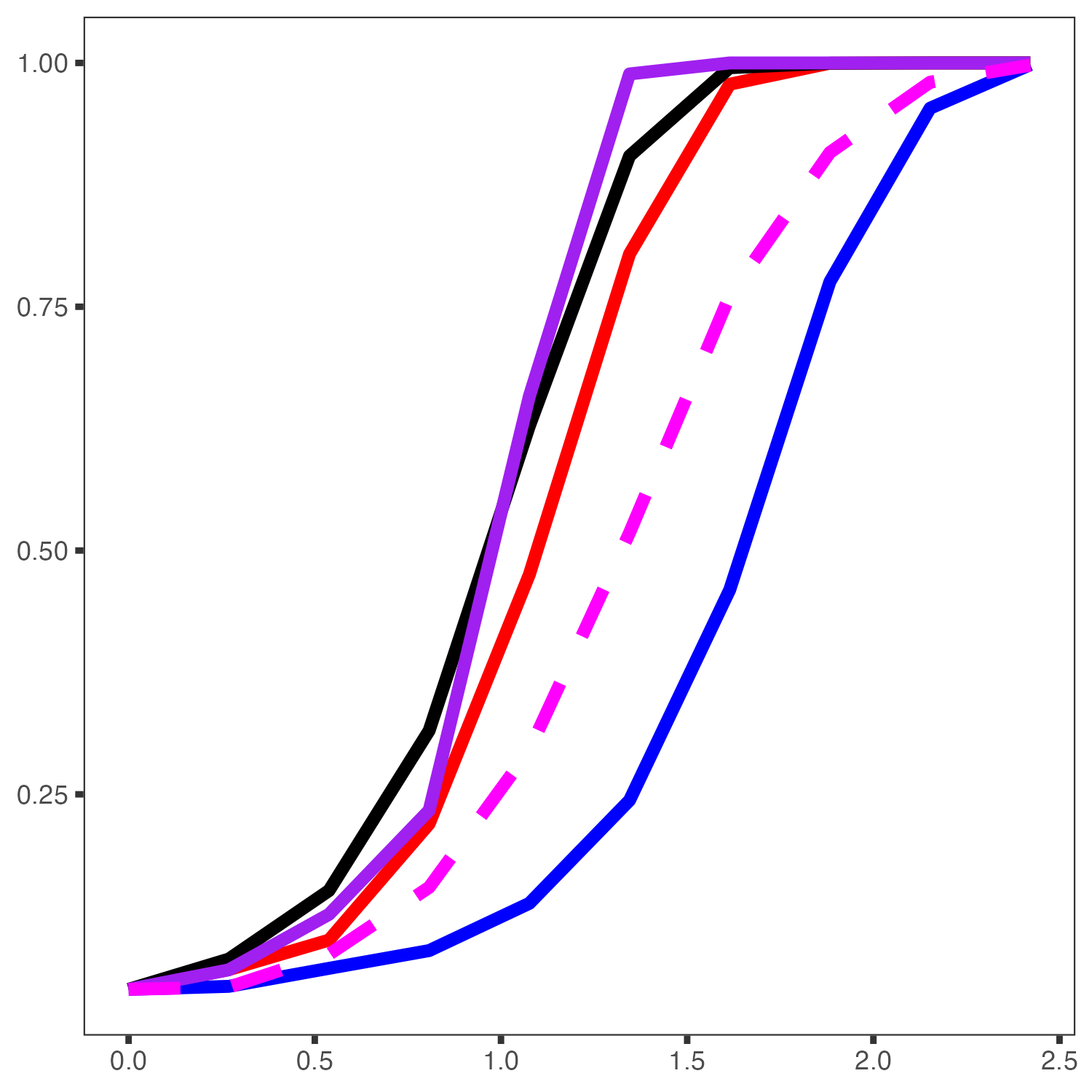}
    \end{subfigure}
    \begin{subfigure}[t]{0.23\linewidth}
        \centering
        \includegraphics[width=\linewidth]{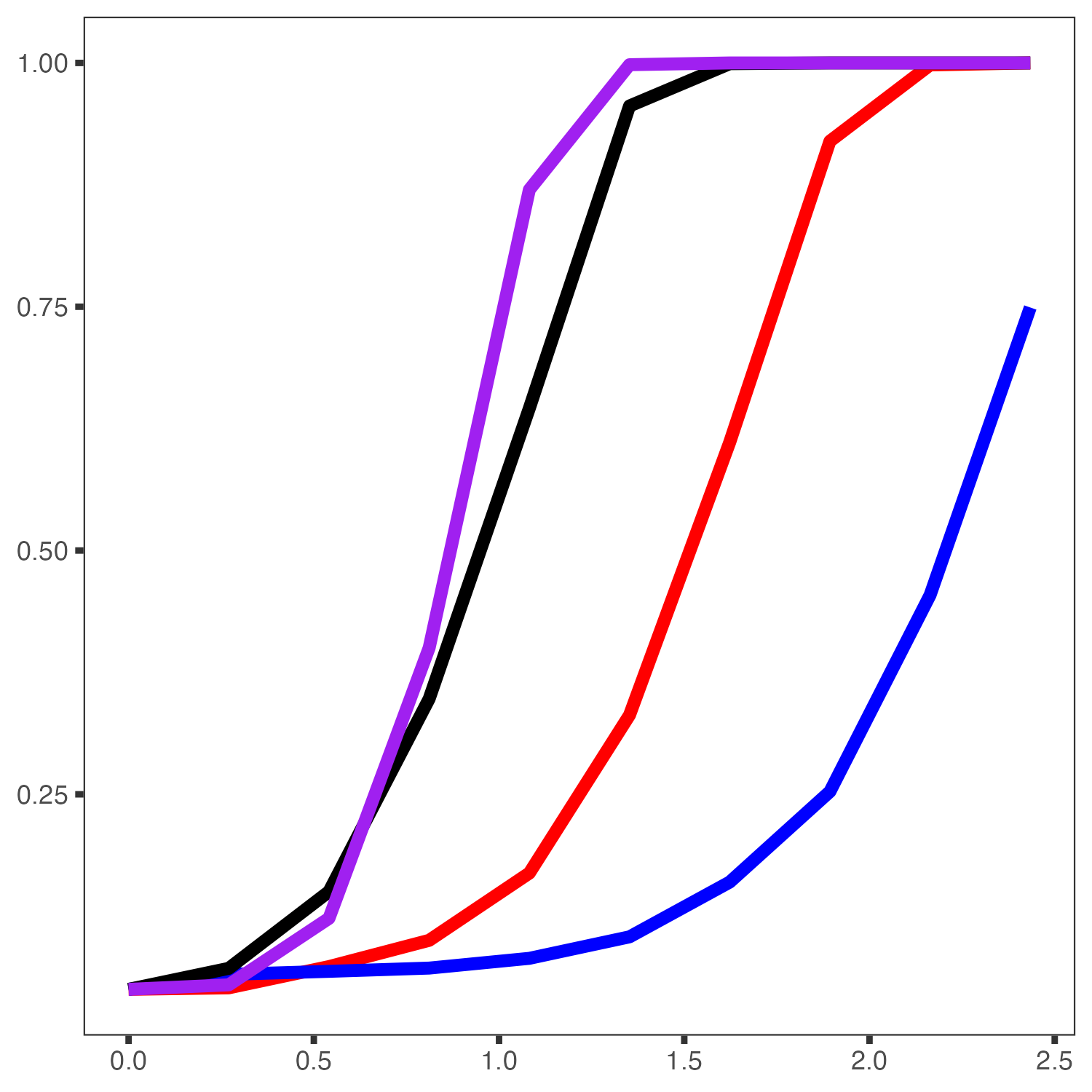}
    \end{subfigure}
    \caption{
    Size-adjusted empirical power under the setting of $\Sigma_0=\Sigma_{\rm Poly}$, and the Exp-Decay correlation model. Panel columns (left to right) correspond to $(p_1,p_2)=(100,100)$, $(100,200)$, $(200,100)$, and $(200,200)$. Panel rows (top to bottom) correspond to $n=400, 200$. The curves represent $\LRT(k_1,\hat{\lambda}_*)$ (black, solid), $\LRT(k_2,\hat{\lambda}_*)$ (purple, solid), $\ell_{\max}(k_3,\hat{\lambda}_*)$ (red, solid), $\ell_{\max}(k_4,\hat{\lambda}_*)$ (blue, solid), YP-LRT (magenta, dashed; available only when $p_2<n$), and HPY-Largest (cyan, dashed; available only when $p_1+p_2<n$). 
    }
    \label{fig:power_poly_exp_decay}
\end{figure}

Nevertheless, the largest-root procedure remains advantageous when a relatively large value of $k$ is required. As demonstrated in Table~\ref{tab:emp_sizes}, the largest-root procedure maintains accurate control of the type-I error rate for substantially larger values of $k$ than the trace-based procedure. Consequently, in applications where a large reduced dimension is needed to adequately remove the dependence between $\bX$ and $\bY$ in $\bW_{k2}$, the largest-root procedure provides a more reliable choice.

The effect of $k$ on empirical power reflects the trade-off between signal accumulation in $\bW_{k1}$ and the residual degrees of freedom available for estimating $\bW_{k2}$. Increasing $k$ allows more principal components, and hence potentially more signal, to be incorporated into $\bW_{k1}$. At the same time, however, it reduces the degrees of freedom associated with $\bW_{k2}$, thereby increasing the variability of the test statistic. This trade-off is clearly illustrated by comparing $\ell_{\max}(k_3,\hat{\lambda}_*)$ and $\ell_{\max}(k_4,\hat{\lambda}_*)$. Under both correlation models, increasing $k$ from $k_3$ to $k_4$ contributes little additional signal, while substantially reducing the degrees of freedom of $\bW_{k2}$. As a result, $\ell_{\max}(k_3,\hat{\lambda}_*)$ consistently achieves higher empirical power than $\ell_{\max}(k_4,\hat{\lambda}_*)$ across all simulation settings.

\begin{figure}[htbp]
    \centering
    \begin{subfigure}[t]{0.23\linewidth}
        \centering
        \includegraphics[width=\linewidth]{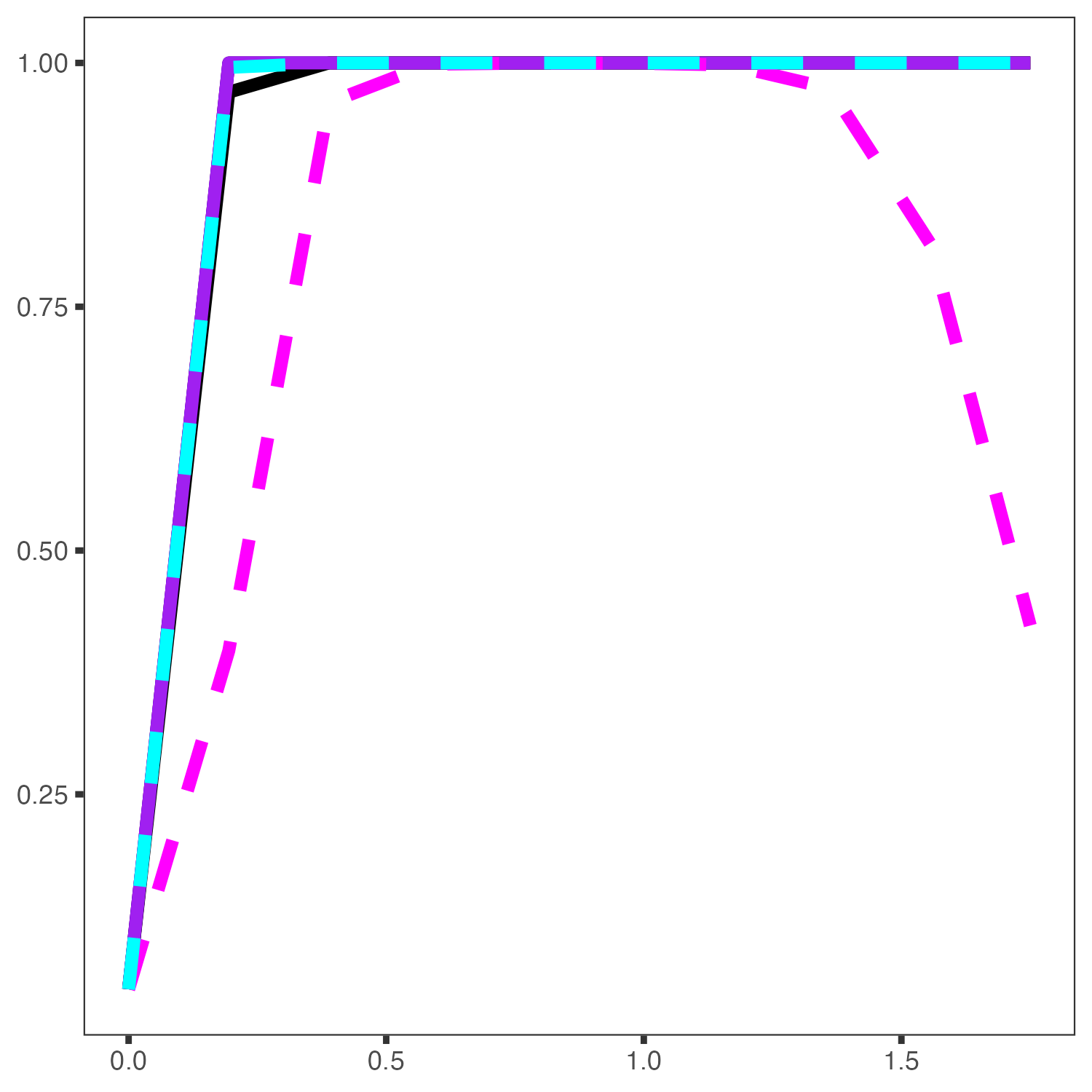}
    \end{subfigure}%
    \begin{subfigure}[t]{0.23\linewidth}
        \centering
        \includegraphics[width=\linewidth]{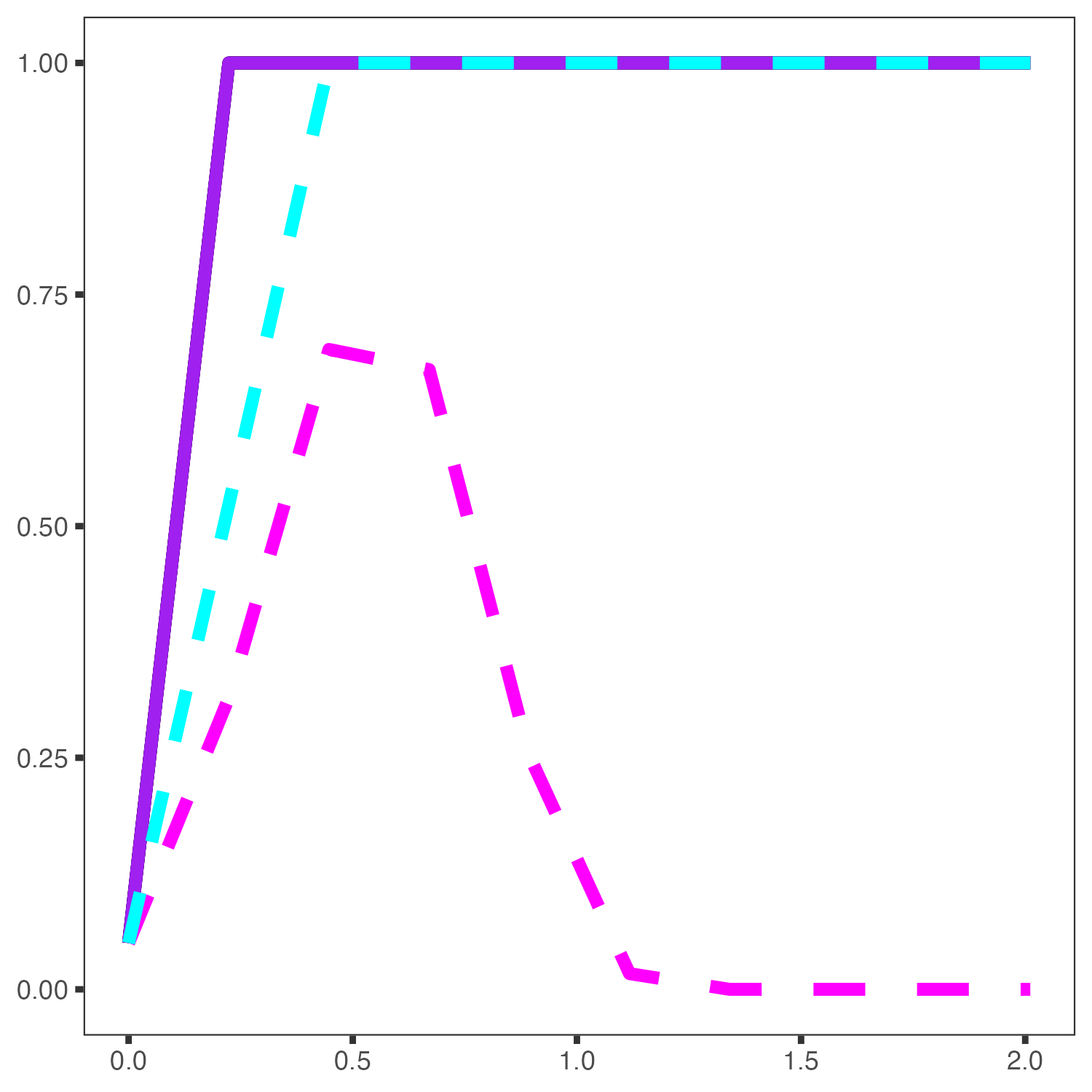}
    \end{subfigure}
    \begin{subfigure}[t]{0.23\linewidth}
        \centering
        \includegraphics[width=\linewidth]{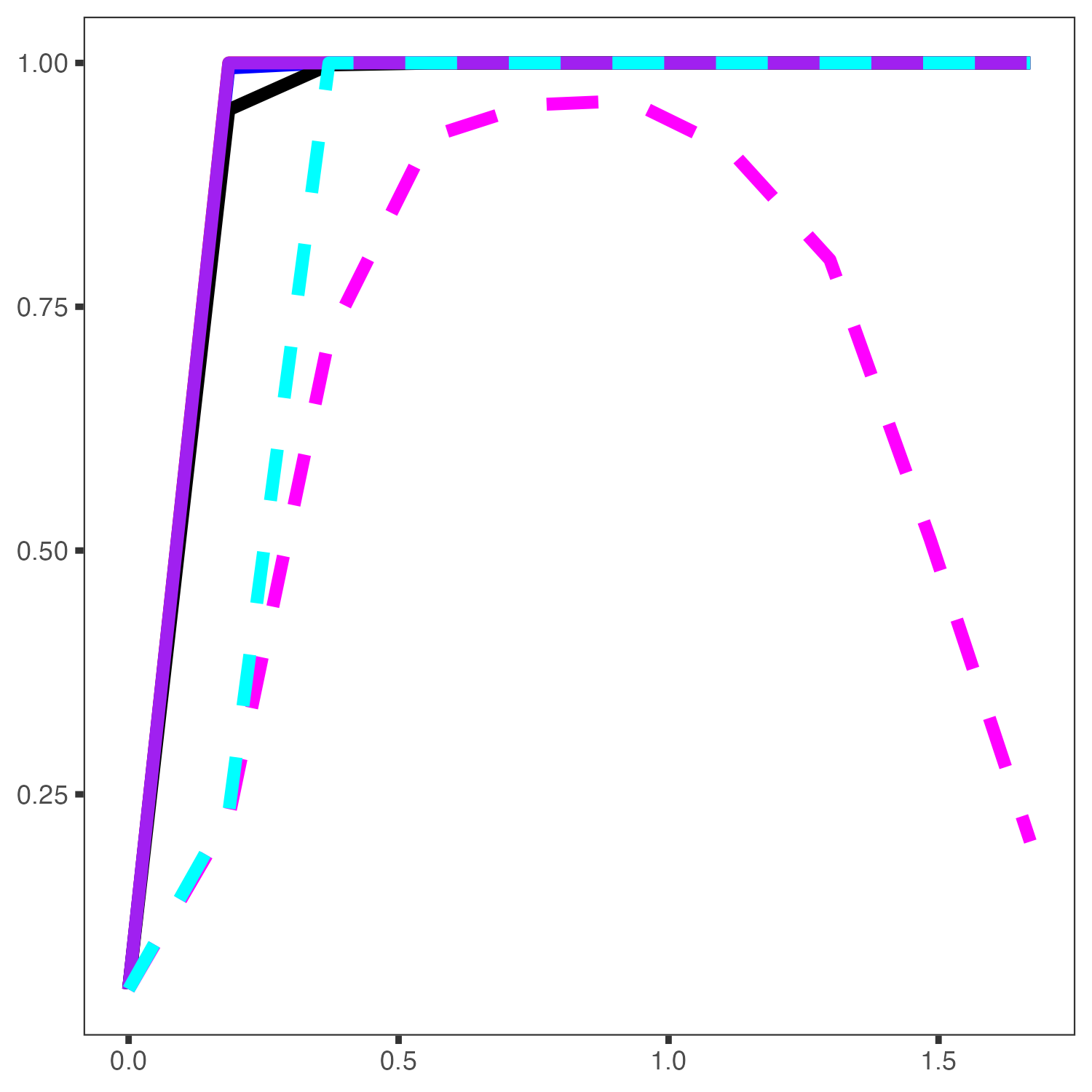}
    \end{subfigure}
    \begin{subfigure}[t]{0.23\linewidth}
        \centering
        \includegraphics[width=\linewidth]{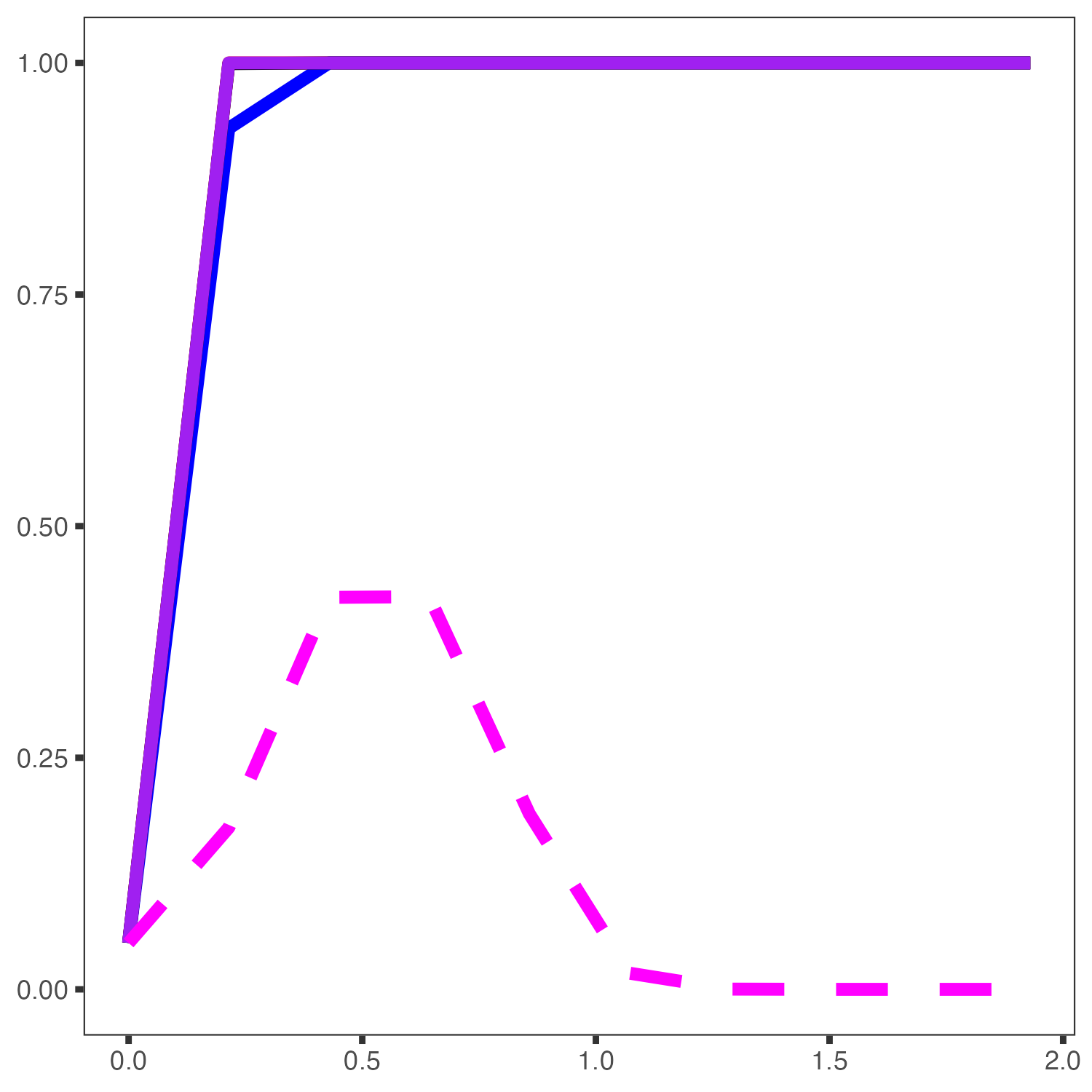}
    \end{subfigure}
    \vfill
    \begin{subfigure}[t]{0.23\linewidth}
        \centering
        \includegraphics[width=\linewidth]{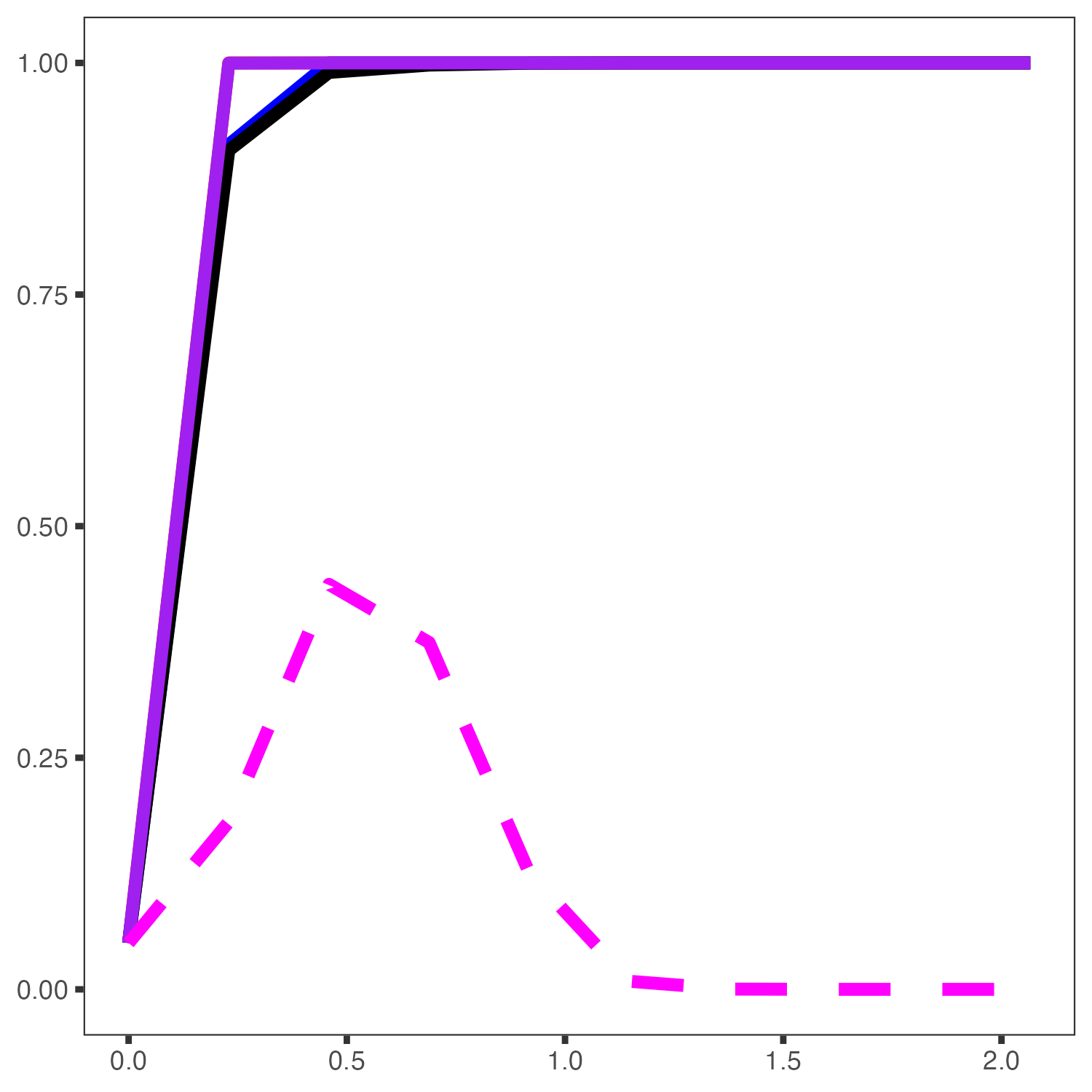}
    \end{subfigure}%
    \begin{subfigure}[t]{0.23\linewidth}
        \centering
        \includegraphics[width=\linewidth]{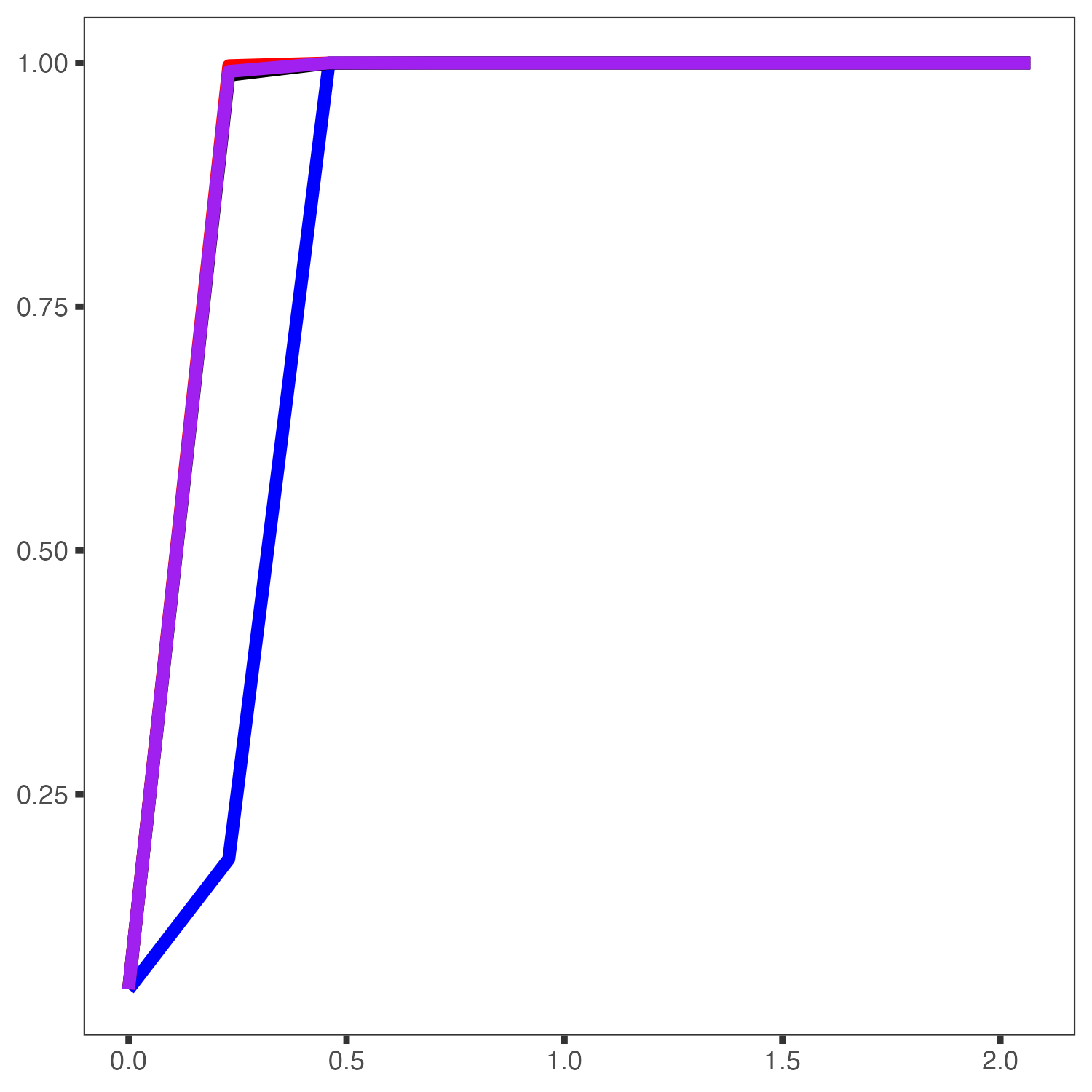}
    \end{subfigure}
    \begin{subfigure}[t]{0.23\linewidth}
        \centering
        \includegraphics[width=\linewidth]{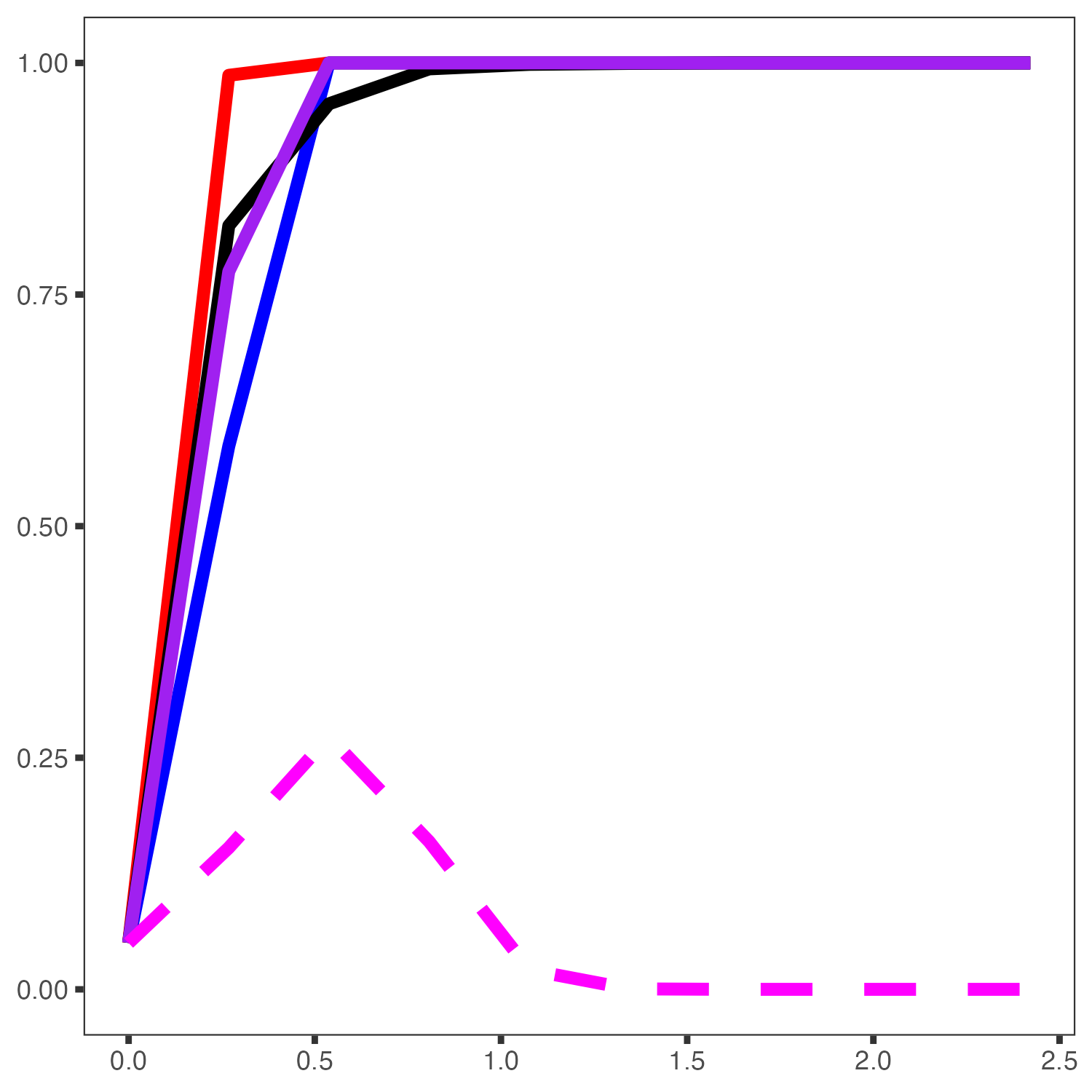}
    \end{subfigure}
    \begin{subfigure}[t]{0.23\linewidth}
        \centering
        \includegraphics[width=\linewidth]{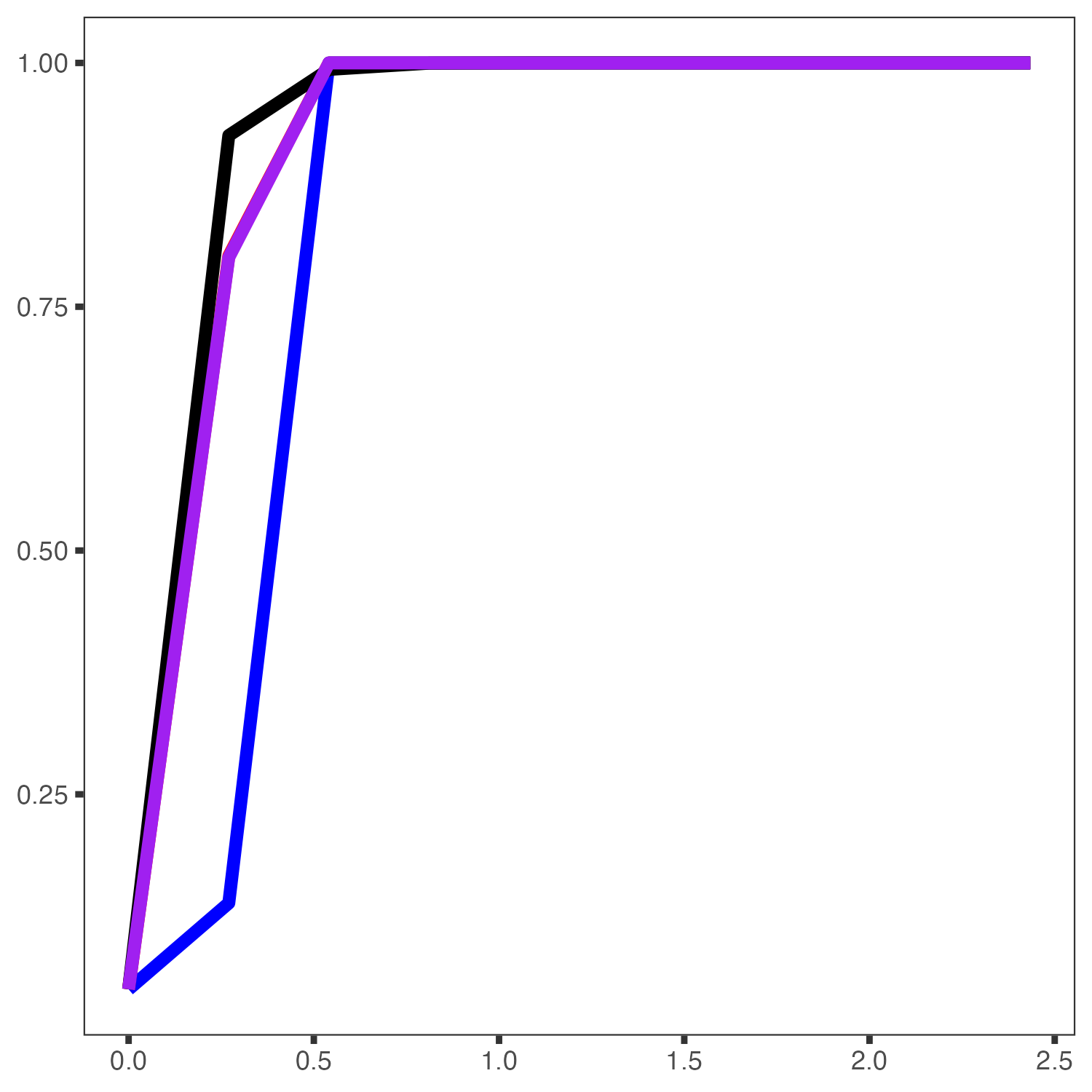}
    \end{subfigure}
    \caption{
    Same as Figure \ref{fig:power_poly_exp_decay} but under $\Sigma_0=\Sigma_{\rm Poly}$ and the Low-rank correlation model. 
    }
    \label{fig:power_poly_low_rank}
\end{figure}

\begin{figure}[htbp]
    \centering
    \begin{subfigure}[t]{0.23\linewidth}
        \centering
        \includegraphics[width=\linewidth]{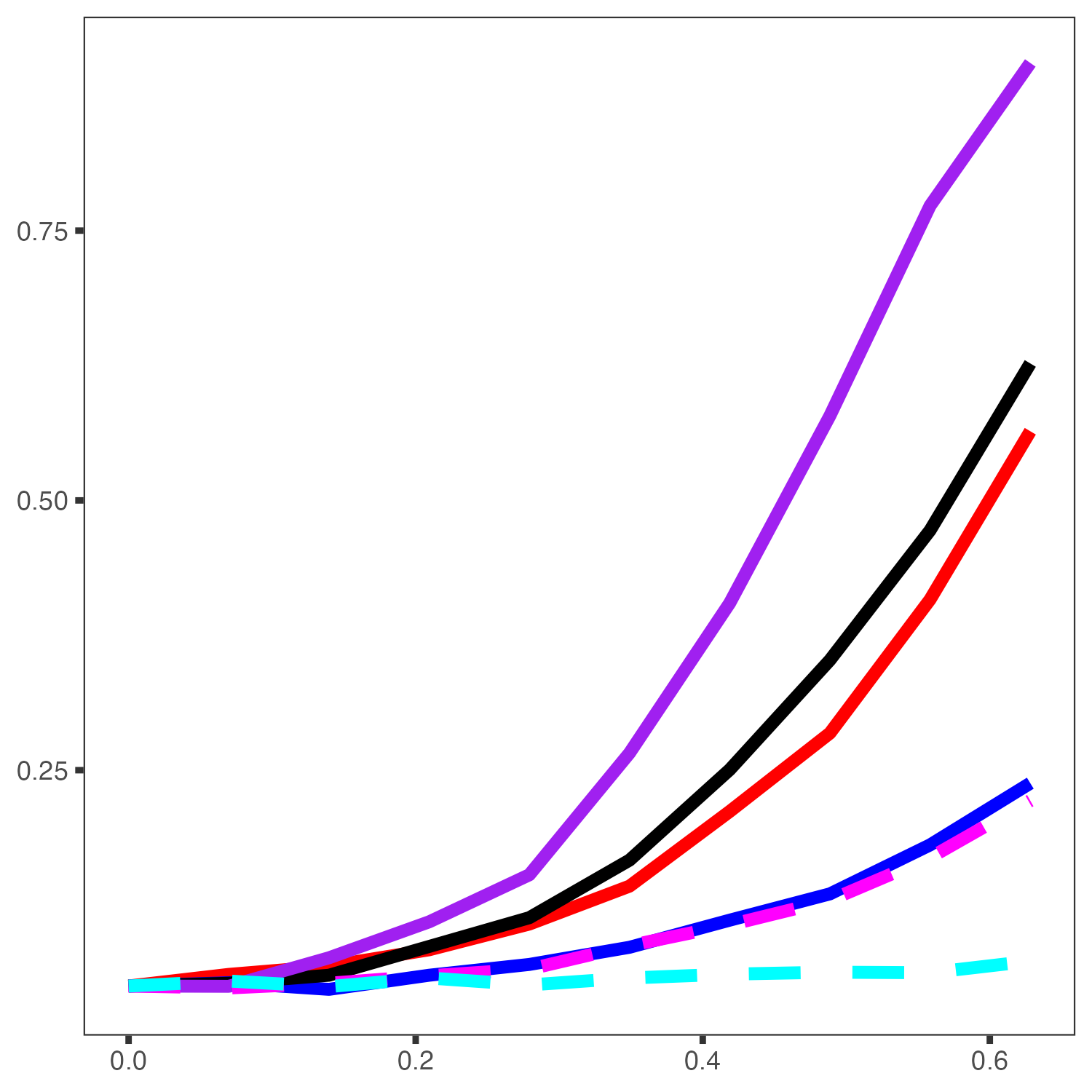}
    \end{subfigure}%
    \begin{subfigure}[t]{0.23\linewidth}
        \centering
        \includegraphics[width=\linewidth]{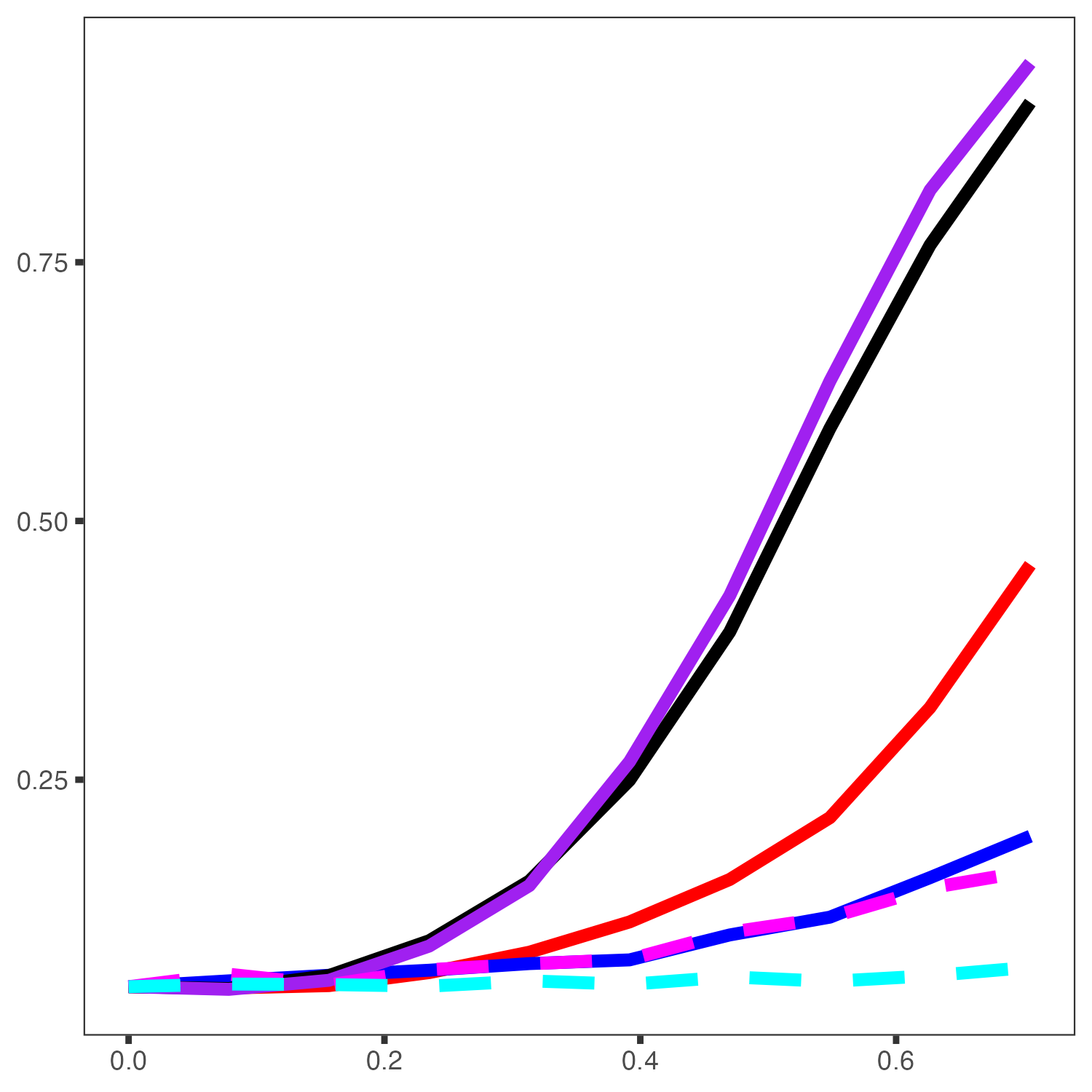}
    \end{subfigure}
    \begin{subfigure}[t]{0.23\linewidth}
        \centering
        \includegraphics[width=\linewidth]{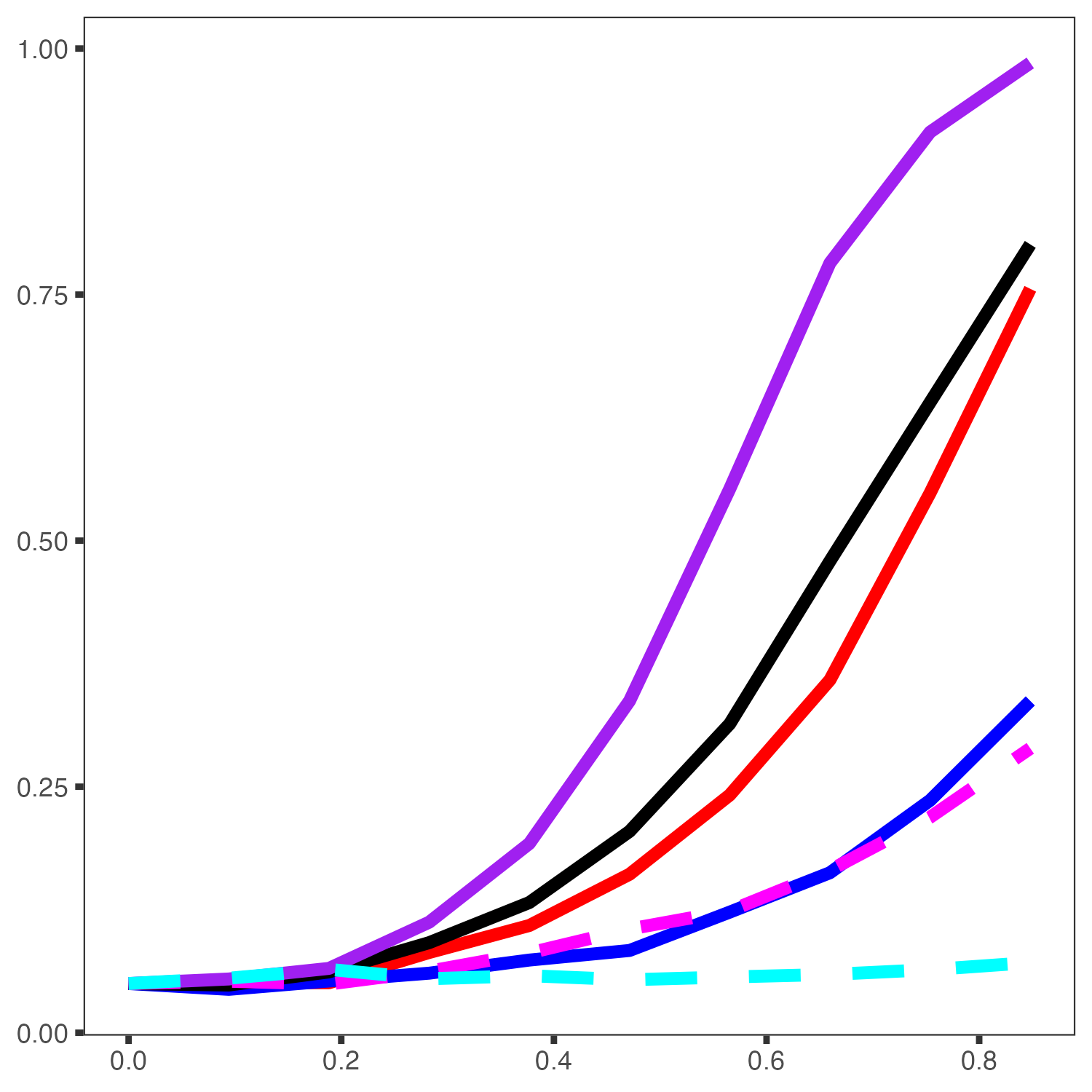}
    \end{subfigure}
    \begin{subfigure}[t]{0.23\linewidth}
        \centering
        \includegraphics[width=\linewidth]{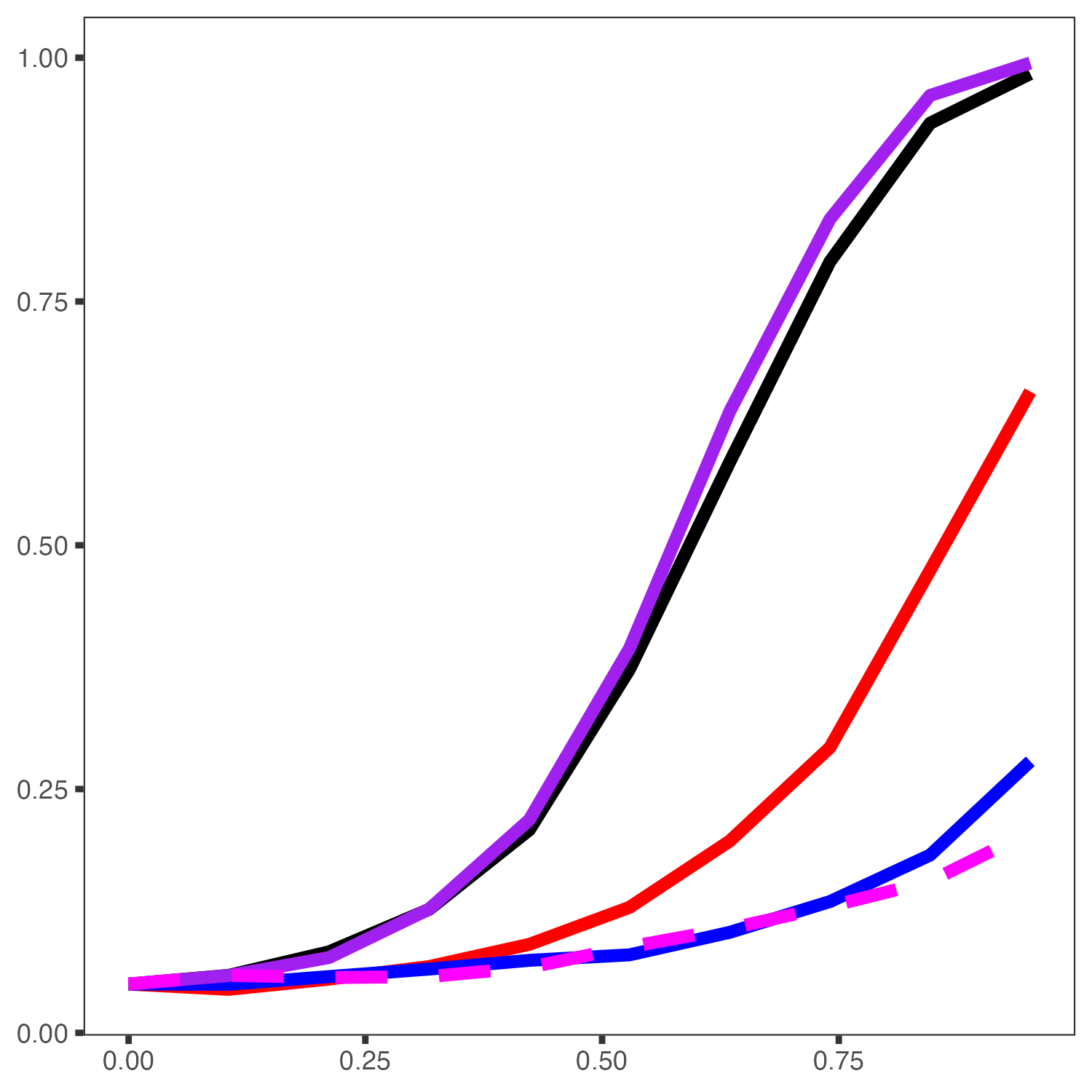}
    \end{subfigure}
    \vfill
    \begin{subfigure}[t]{0.23\linewidth}
        \centering
        \includegraphics[width=\linewidth]{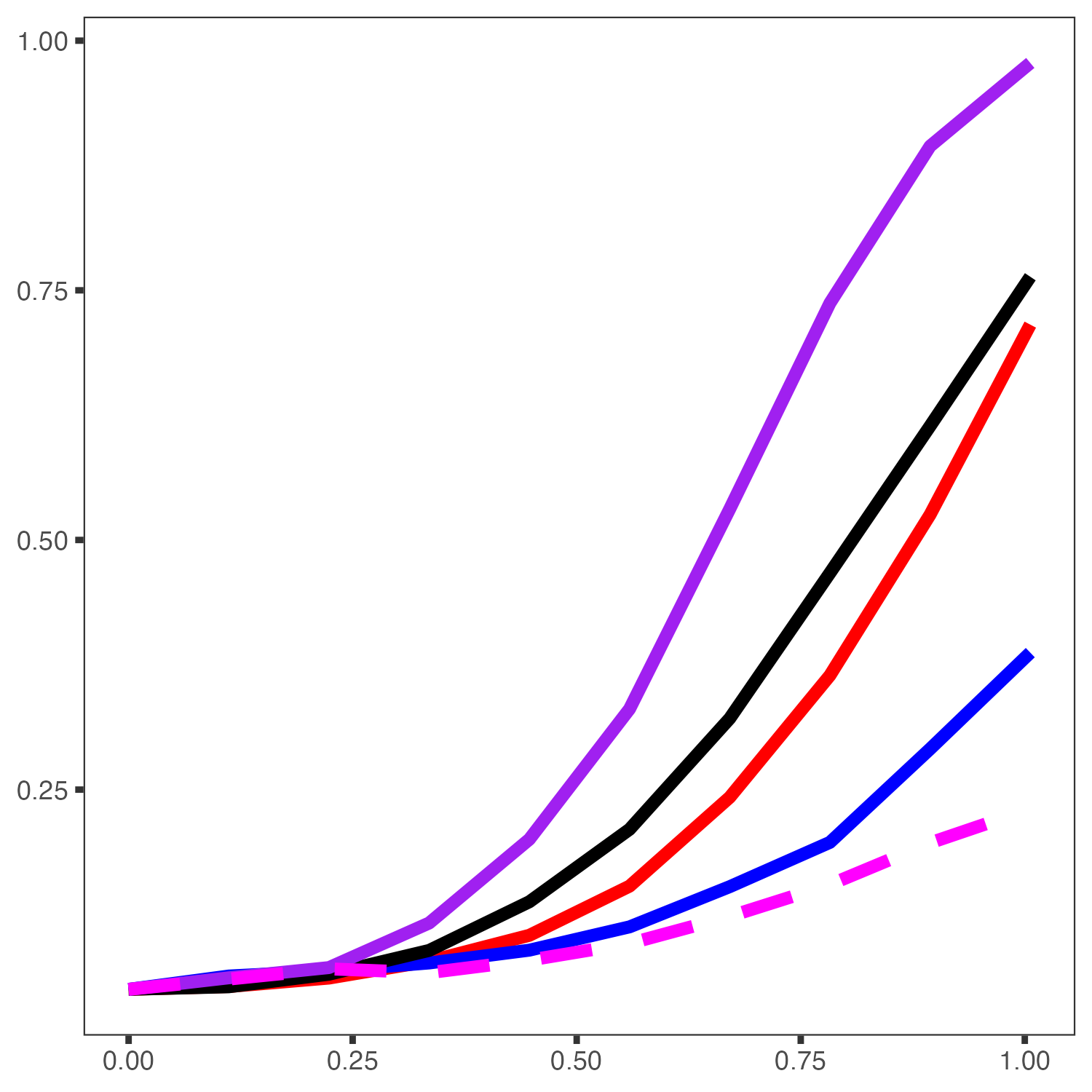}
    \end{subfigure}%
    \begin{subfigure}[t]{0.23\linewidth}
        \centering
        \includegraphics[width=\linewidth]{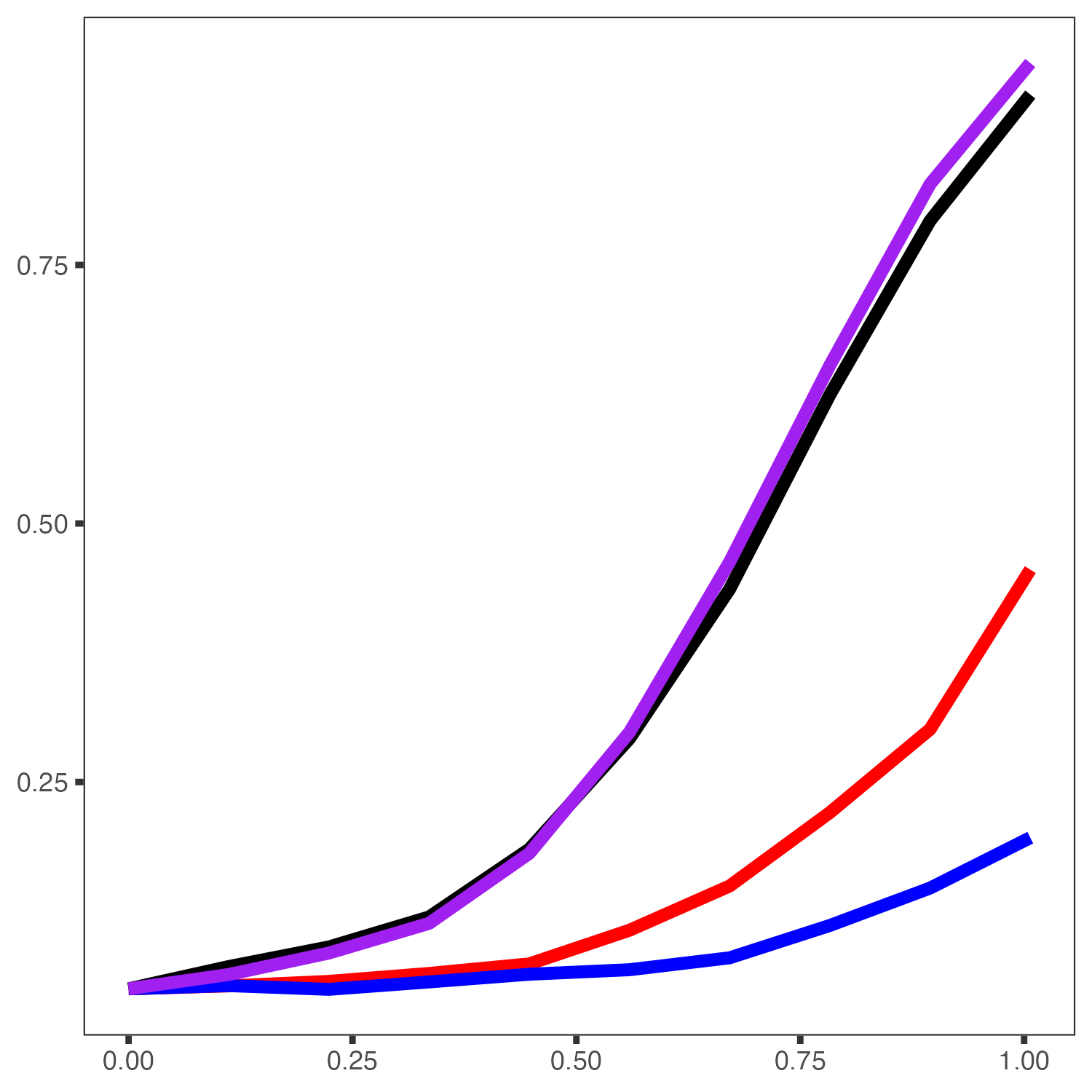}
    \end{subfigure}
    \begin{subfigure}[t]{0.23\linewidth}
        \centering
        \includegraphics[width=\linewidth]{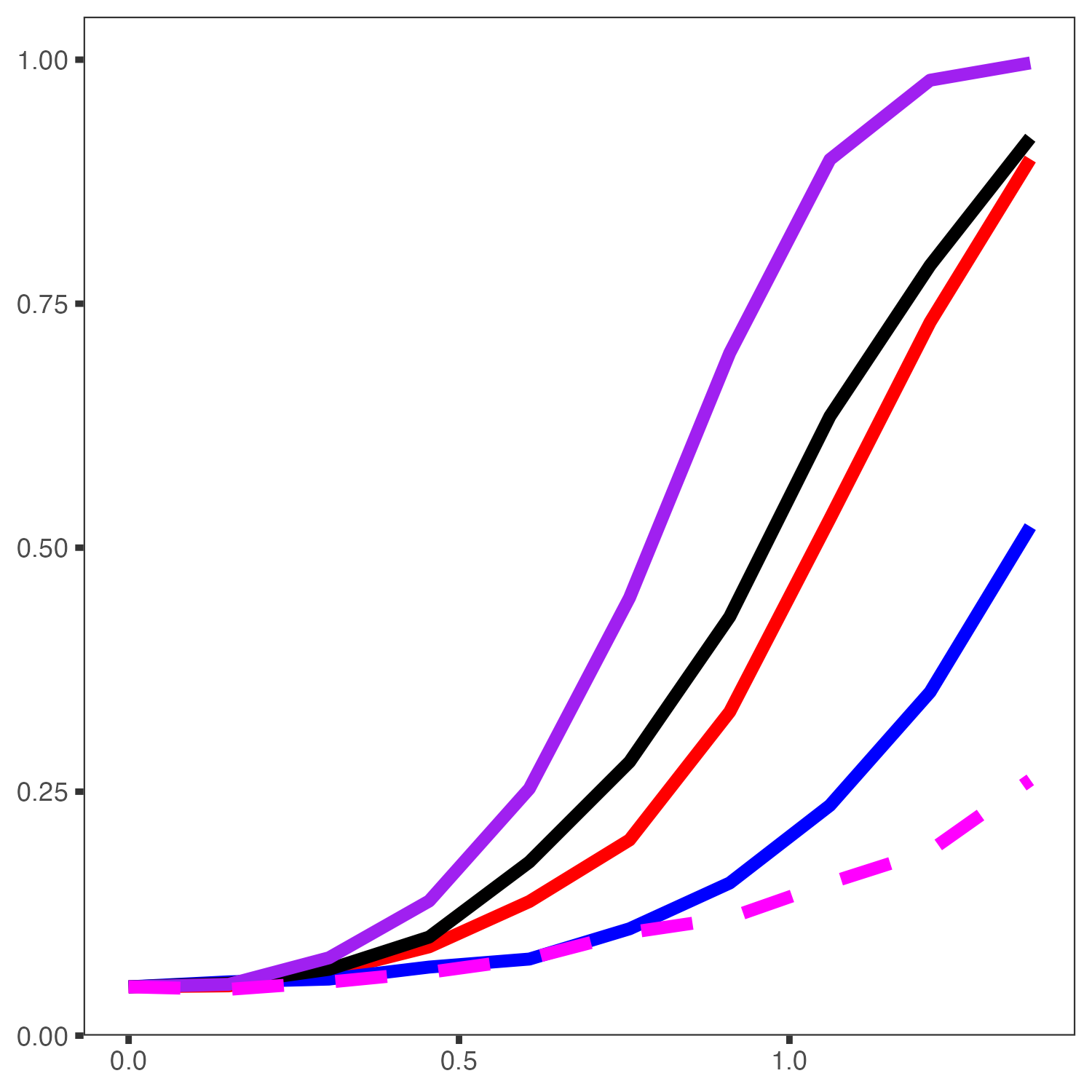}
    \end{subfigure}
    \begin{subfigure}[t]{0.23\linewidth}
        \centering
        \includegraphics[width=\linewidth]{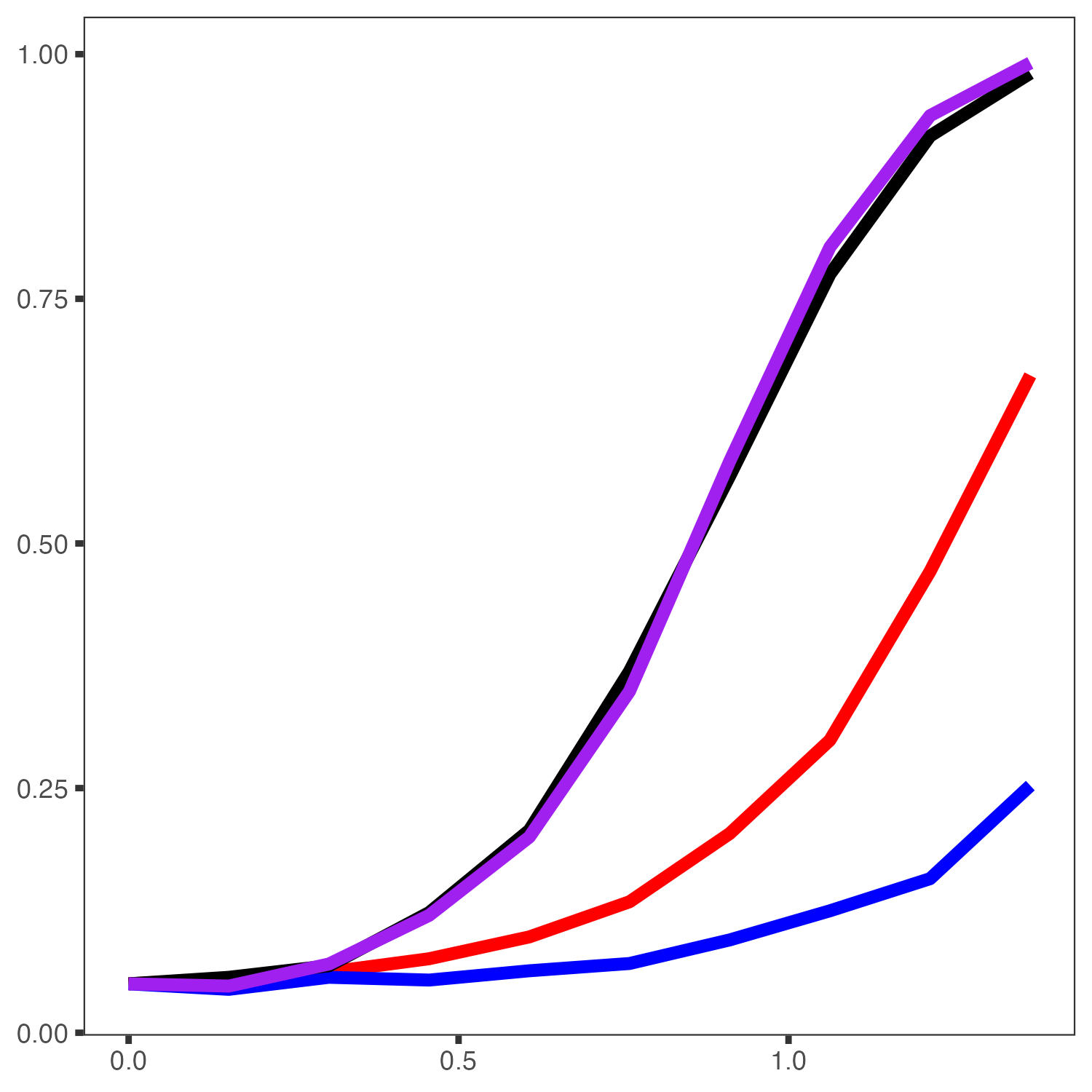}
    \end{subfigure}
    \caption{
    Same as Figure \ref{fig:power_poly_exp_decay} but under $\Sigma_0 = \Sigma_{ID}$ and the Exp-Decay correlation model.
    }
    \label{fig:power_ID_exp_decay}
\end{figure}

The results indicate that {YP-LRT} is asymptotically consistent under the Exp-Decay correlation model when the signal strength is sufficiently large. However, its performance deteriorates substantially under the Low-rank correlation model, where it fails to effectively exploit the concentrated signal; see Figure~\ref{fig:power_poly_low_rank}. This demonstrates that the power of {YP-LRT} depends strongly on the underlying correlation structure. 

Finally, we consider {HPY-Largest}. This procedure is based on the largest eigenvalue of $\bW_{k1}\bW_{k2}^{-1}$ with $k=p_2$ and can therefore be viewed as the unregularized counterpart $(\lambda=0)$ of the proposed largest-root procedure $\ell_{\max}(p_2,\lambda)$. The simulation results indicate that {HPY-Largest} performs comparably to the proposed procedure when the total dimension $p_1+p_2$ is relatively small compared with the sample size $n$. As the dimension increases and $p_1+p_2$ approaches $n$, however, the proposed regularized procedure exhibits a noticeable power advantage. This behavior is consistent with the motivation for ridge regularization. When the residual degrees of freedom associated with $\bW_{k2}$ become limited, the inverse $\bW_{k2}^{-1}$ becomes increasingly variable, which adversely affects the signal-to-noise ratio of the largest-root statistic. By incorporating a ridge regularization term, the proposed procedure stabilizes the inverse of $\bW_{k2}$ and consequently achieves improved finite-sample power in high-dimensional settings.

\section{Discussion}\label{sec:discussion}

We have proposed a regularized framework for testing the independence between two high-dimensional random vectors by combining ridge regularization with principal-component-based dimension reduction. The resulting procedures remain statistically stable when the dimensions are comparable to, or exceed, the sample size. Under two asymptotic regimes determined by the reduced dimension parameter, we established the limiting distributions of the proposed trace-based and largest-root statistics, developed consistent estimators of the required normalization parameters, and investigated their asymptotic power properties. Numerical studies further demonstrate that the proposed procedures provide satisfactory finite-sample performance over a broad range of covariance structures and dependence patterns.

An important feature of the proposed framework is its flexibility. By varying the reduced dimension parameter $k$ and the regularization parameter $\lambda$, the procedures can adapt to different signal configurations and dimensional regimes. Our theoretical and numerical results suggest that the trace-based procedure is preferable when the dependence can be adequately represented by a relatively small reduced dimension, whereas the largest-root procedure is more suitable when a larger reduced dimension is required. We also proposed practical data-driven strategies for selecting both tuning parameters, making the methodology readily applicable in practice.

More broadly, the methodology developed in this paper demonstrates how principal-component-based dimension reduction can be combined to stabilize classical multivariate procedures in high-dimensional settings. We expect that the same principle may prove useful for other high-dimensional multivariate inference problems, including multivariate regression, MANOVA,  and general linear hypotheses. Developing such extensions constitutes an interesting direction for future research. 

Several aspects of the present work also merit further investigation. First, although the proposed data-driven procedures for selecting $k$ and $\lambda$ perform well in our numerical studies, their theoretical optimality remains open. Second, the power analysis under the alternative is carried out under representative structural assumptions that permit explicit asymptotic characterizations. Extending these results to more general alternatives would provide a deeper theoretical understanding of the proposed procedures. Finally, our asymptotic analysis assumes independent observations, whereas many modern applications involve temporal, spatial, or network dependence. Extending the proposed methodology to accommodate such dependence structures is another promising direction for future research.


\appendix

\section{Estimation of $\Theta_1(\lambda,q)$ and $\Theta_2(\lambda,q)$}
\label{sec:estimation}

This section describes the estimation method of the normalization parameters $\Theta_1(\lambda,q)$ and $\Theta_2(\lambda,q)$. The procedure is adapted from \cite{li2025ridge}, with several modifications tailored to the present framework.

\subsection{Reformulation}
Recall that as in Theorem \ref{thm:main_diverge_k}, the quantities $\Theta_1(\lambda,q)$ and $\Theta_2(\lambda,q)$ are defined through the function $s(\cdot)$ and its derivatives, where $s(\cdot)$ is specified implicitly via an inverse mapping. To facilitate estimation, we first reformulate these quantities in a more tractable form. To this end, define the auxiliary functions
\[
\calH_j(h)
=
\int
\frac{\tau^j\, dF^{\Sigma_0}(\tau)}
{(\lambda-\tau h)^j},
\qquad
h\in(-\infty,\eta),
\qquad
j=1,2,3.
\]
Recall that $q = p_1/(n-1-k)$ and the definition of $g(h)$ in \eqref{eq:def_g}. With this notation, the function $g (h)$ and its derivative may be written as
\begin{equation}\label{eq:def_g_H}
\begin{split}
&g (h) = h + \Bigl[1 + q \calH_1(h)\Bigr]^{-1},\\ 
&g' (h) = 1 - q \Bigl[1 + q \calH_1(h)\bigr]^{-2} \calH_2(h).
\end{split}
\end{equation}
We express $s (g)$ and its derivatives in terms of the functions $\calH_j(h)$, $j=1,2,3$. 
\begin{lemma}
    \label{thm:determine_s}
    For any $g_0 \in [0, ~\rho )$,  
   \[ s (g_0) = q^{-1} \Big(\frac{1}{g_0- h(g_0) } -1\Big), \]
   where $h= h(g_0)$ is the unique solution in $(-\infty, ~\eta )$ to 
   \begin{equation}\label{eq:determine_s}
   \begin{cases}
       &g_0 = h + \big[ 1+ q \calH_1(h)\big]^{-1},\\[5pt]
       &1 - q \big[ 1+ q\calH_1(h)\big]^{-2} \calH_2(h) >0. 
   \end{cases} 
   \end{equation}
   Furthermore, the derivatives $s' (g_0)$ and $s'' (g_0)$ are determined as:  
\begin{equation}\label{eq:characterization_s2s3}
\begin{aligned}
&b(g_0) = g_0 - \frac{1}{1 + q {s} (g_0)},\\
&\zeta(g_0) = \frac{q}{( 1+ q s (g_0))^{2}},\\
&s' (g_0) =  \calH_2(b(g_0))\left(\zeta(g_0) s' (g_0) + 1 \right),\\
&s'' (g_0) = 2 \left(\zeta(g_0) s' (g_0) + 1 \right)^2 \calH_3(b(g_0))\\
&\qquad + \calH_2(b(g_0)) \Big[ \frac{- 2 (q s' (g_0))^2 }{(1+q s (g_0))^3} + \zeta(g_0)s'' (g_0) \Big].
\end{aligned}
\end{equation}
\end{lemma}

Notably, the equations in \eqref{eq:characterization_s2s3} define a system of ordinary differential equations. Therefore, to estimate $s(g)$, $s'(g)$, and $s''(g)$ over the entire domain, it suffices to obtain accurate estimates of the initial value $s(0)$ and the functions $\calH_j(h)$, $j=1,2,3$. The estimate of $s(0)$ can be obtained by solving \eqref{eq:determine_s} at $g_0 =0$. The remaining estimation can then be carried out using standard ODE solvers. The resulting procedure is computationally efficient. In particular, the third equation in \eqref{eq:characterization_s2s3} is linear in $s'(g_0)$ once $s(g_0)$ and $\calH_2(\cdot)$ are specified. Similarly, the last equation in \eqref{eq:characterization_s2s3} is linear in $s''(g_0)$ once $s(g_0)$, $s'(g_0)$, $\calH_2(\cdot)$, and $\calH_3(\cdot)$ are specified.

We summarize the estimation procedure, given estimators of $\rho$ and $\calH_j$, $j=1,2,3$, in Algorithm~\ref{algo:ODE}.

\begin{algorithm}[htbp]
\caption{ODE Formulation for Estimation of $s(g)$, $s'(g)$, and $s''(g)$}
\label{algo:ODE}
\begin{algorithmic}
\STATE \textbf{Input:} $q = p_1/(n-1-k)$, ~$\rho= \hat{\rho}$, ~$\calH_1(\cdot) =  \hat{\calH}_1(\cdot)$,~ $\calH_2(\cdot) = \hat{\calH}_2(\cdot)$, ~$\calH_3(\cdot)= \hat{\calH}_3(\cdot)$;
\STATE \textbf{Initial:} Find  $s_0$ as the unique solution to 
\[\calH_1 \Big( \frac{-1}{1+q s}\Big)  = s \qquad \text{and}\qquad \calH_2\Big(\frac{-1}{1+q s}\Big) < q^{-1} (1+q s)^2.\]
\STATE \textbf{ODE}: On $g \in [0,~\rho)$, solve the following ODE with the initial $s(0) =s_0$:
\begin{equation*}
\begin{aligned}
&s'(g) =  \frac{\calH_2( b(g) )}{ 1 - \calH_2(b(g)) \zeta(g) },\\
&s''(g) = \frac{2[\zeta(g)s'(g) +1]^2 \calH_3(b(g)) - 2{[1+q s(g)]^{-3} } {\calH_2(b(g)) (q s'(g))^2}  }{1 -  \zeta(g) \calH_2(b(g))},\\
&b(g) = g - \frac{1}{1+q s(g)}, ~~\mbox{and } ~~ \zeta(g) = \frac{q}{ (1+q s(g))^2} .
\end{aligned}
\end{equation*}
\STATE \textbf{Output:} Estimated functions: $\hat{s}(g) =  s(g)$, ~$\hat{s}'(g) = s'(g)$, ~$\hat{s}''(g) = s''(g)$ on $[0,~\hat{\rho})$.
\end{algorithmic}
\end{algorithm}

Following Algorithm~\ref{algo:ODE}, we estimate $\beta$ by $\hat{\beta}$, defined as the unique solution to
\[
\hat{\beta}^2
\hat{s}'(\hat{\beta})
=
\frac{k}{p_1},
\qquad
\hat{\beta}\in(0,\hat{\rho}).
\]
The normalization parameters $\Theta_1(\lambda, q)$ and $\Theta_2(\lambda,q)$ are then estimated by
\begin{equation}
\label{eq:estimators_Theta1Theta2}
\hat{\Theta}_1(\lambda)
=
\frac{1}{\hat{\beta}}
\left[
1+
\frac{p_1}{k}
\hat{\beta}\,
\hat{s}(\hat{\beta})
\right],
\qquad
\hat{\Theta}_2(\lambda)
=
\left[
\frac{(p_1/k)^3}{2}
\hat{s}''(\hat{\beta})
+
\frac{(p_1/k)^2}{\hat{\beta}^3}
\right]^{1/3}.
\end{equation}

\subsection{Estimation of $\calH_j$}

The remaining tasks are to estimate $\calH_j(\cdot)$ and $\rho$. To this end, we explore the M-P equation \eqref{eq:M_P_equation} and the convergence of $\hat{\varphi}$ as shown in Result (iii) of Lemma \ref{lemma:determinist_equivalent}. 

First of all, by differentiating both sides of \eqref{eq:M_P_equation} with respect to $z$, we obtain that 
\begin{equation*}
    \begin{split}
        \frac{z}{\lambda q} + \frac{1}{\lambda q \varphi(z)} = \calH_1( - \lambda \varphi(z)),\\
        \frac{1}{\lambda^2 q \varphi^2(z)}  -\frac{1}{\lambda^2 q \varphi'(z)} = \calH_2(-\lambda \varphi(z)).
    \end{split}
\end{equation*}
Here, we extend the definition of $\calH_j(h)$ to the complex domain by allowing $h$ to be complex-valued. 
Since $\varphi(z)$ and $\varphi'(z)$ can be precisely estimated by $\hat{\varphi}(z)$ and $\hat{\varphi}'(z)$,  we define  
\begin{equation}\label{eq:def_Q1Q2}
    Q_1(z)
=
\frac{z}{\lambda q}
+
\frac{1}{\lambda q \hat{\varphi}(z)},
\qquad
Q_2(z)
=
\frac{1}{\lambda^2 q \hat{\varphi}^2(z)}
-
\frac{1}{\lambda^2 q \hat{\varphi}'(z)},
\quad z \in \mathbb{C}^{+}.
\end{equation}
It follows then uniformly for $z$ in any closed and compact subset $\mathcal{C}$ of $\mathbb{C}^+$,
\begin{equation}\label{eq:Q12_approximate_H12}
Q_1(z)  -  \calH_1(-\lambda \hat{\varphi}(z))  \stackrel{P}{\longrightarrow} 0, \qquad   Q_2(z) - \calH_2(-\lambda \hat{\varphi}(z)) \stackrel{P}{\longrightarrow} 0.
\end{equation}

Note that $F^{\Sigma_0}$ can be well-approximated by a weighted discrete measure of the form
\[
F^{\Sigma_0}(\tau)
\simeq
\sum_{b=1}^B
w_b\,\mathbb{I}(\tau>\sigma_b),
\]
where $\{\sigma_b\}_{b=1}^B$ is a user-specified grid chosen to densely cover the interval 
\[\bigl[\liminf \ell_{\min}(\Sigma_0),~\limsup \ell_{\max}(\Sigma_0)\bigr]\] and the weights $\{w_b\}_{b=1}^B$ satisfy $w_b\geq 0$ and $\sum_{b=1}^B w_b=1$. Accordingly, the target functions admit the approximation
\[
\calH_j(h)
=
\int
\frac{\tau^j\, dF^{\Sigma_0}(\tau)}
{(\lambda-\tau h)^j}
\simeq
\sum_{b=1}^B
\frac{
w_b \sigma_b^j
}{
(\lambda-\sigma_b h)^j
}
\coloneqq
\tilde{\calH}_j(h;w_1,\dots,w_B),
\qquad
j=1,2.
\]
In the following estimation procedure, we treat the grid points $\{\sigma_b\}_{b=1}^B$ as fixed, while the weights $\{w_b\}_{b=1}^B$ are treated as unknown model parameters. The resulting estimators are generally insensitive to the particular choice of $\{\sigma_b\}_{b=1}^B$, provided that the grid is chosen sufficiently dense.

Motivated by the relationship \eqref{eq:Q12_approximate_H12}, we consider a grid of points \(z\) in $\mathbb{C}^+$, denoted by \(\{z_i\}_{i=1}^I\). We select the weights \(w_b\) such that
\[Q_1(z) \simeq  \tilde{\calH}_1(-\lambda\hat{\varphi}(z); w_1, \dots, w_B),  \quad Q_2(z) \simeq \tilde{\calH}_2(-\lambda\hat{\varphi}(z); w_1, \dots, w_B), \quad z \in \{z_i\}_{i=1}^I.  \]
In particular,  given the grids $\{\sigma_b\}_{b=1}^B$ and $\{z_i\}_{i=1}^I$,  define the loss function as 
\begin{equation} \label{eq:L_infty_loss}
\begin{split}
L_\infty&(w_1,\dots, w_B)= \max_{i=1,\dots,I} \max_{j=1,2} \Big\{ |\Re(e_{ij} )|, |\Im(e_{ij})|\Big\},~~\mbox{where} \\
    e_{ij} &= \frac{Q_j(z_i) - \tilde{\calH}_j(-\lambda\hat\varphi(z_i), w_1,\dots, w_B)}{|Q_j(z_i)|}, \quad i = 1,\dots, I; j=1,2.
\end{split}
\end{equation}
We choose the weights as
\begin{equation}\label{eq:arg_min_weights}
(w_1^*, \dots, w_B^*) = \arg\min_{w_b \geq 0,\ \sum w_b = 1} L_\infty(w_1, \dots, w_B).
\end{equation}
To mitigate the impact of mass points associated with negligible weights, we incorporate a truncation step applied to $(w_1^*,\dots,w_B^*)$. Specifically, for some small constant $a>0$, we set $w_b^*=0$ whenever $w_b^*<a$, and subsequently renormalize the remaining weights so that they sum to one. With a slight abuse of notation, we continue to use $B$ to denote the number of retained nonzero weights after truncation. The resulting positive weights and associated mass points are denoted by $\{\hat{w}_b: b=1,\dots,B\}$ and $\{\hat{\sigma}_b: b=1,\dots,B\}$, respectively.  Without loss of generality, we assume that the mass points are ordered such that
\[
\hat{\sigma}_1>\hat{\sigma}_2>\cdots>\hat{\sigma}_B.
\]
The recommended choices of the grids $\{\sigma_b\}$ and $\{z_i\}$ are given in Section \ref{sec:additional_details_estimation}. 

The choice of the $L_\infty$-type loss is motivated primarily by computational considerations. In particular, once the support points $\{\sigma_b\}_{b=1}^B$ are specified, the resulting optimization problem can be reformulated as a linear program. This formulation is presented in Algorithm \ref{algo:linear_program}. 
\begin{algorithm}[h]
\caption{Linear Programming Formulation for Selecting Weights}
\label{algo:linear_program}
\begin{algorithmic}
\STATE \textbf{Input:} A grid of masses \(\sigma_1, \dots, \sigma_B\); a grid of points \(z_1, \dots, z_I\); eigenvalues $\ell_j$'s of $\bW_{k2}$; $\lambda>0$. 
\STATE \textbf{Calculation:}  Compute
\[e_{ij} =  \frac{Q_j(z_i)}{|Q_j(z_i)| }- \frac{1}{|Q_j(z_i)|}\sum_{b=1}^B \frac{\sigma_b^jw_b}{(\sigma_b  \lambda\hat\varphi(z_i) +\lambda)^j},~ i =1,\dots, I; j=1,2,\]

\STATE \textbf{Optimization:} 
Find the optimal $(w_1,\dots, w_B, \theta)$ through the linear programming
\renewcommand{\arraystretch}{0.3}  
\[
\begin{array}{ll}
\underset{\theta,w_1,\dots,w_B}{\text{minimize}} & \phantom{-\theta\leq}\theta\\[1pt]
\text{subject to} &  -\theta\leq \Re(e_{ij}) \leq \theta, \quad \forall~ i = 1,\dots, I; j=1,2, \\[1pt]
&  -\theta\leq \Im(e_{ij}) \leq \theta, \quad \forall~ i = 1,\dots, I; j=1,2, \\[1pt]
& w_b \geq 0, \quad \forall b = 1,\dots, B, \\[1pt]
&  \mbox{and } w_1 + w_2 + \dots + w_B = 1.
\end{array}
\]
\STATE \textbf{Truncation}: To mitigate the effect of floating-point errors, truncate the optimal weight $w_b$ to $w_b \times \mathbb{I}(w_b>B^{-1}10^{-d})$ and rescale the truncated weights to be of sum one. The recommended choice of $d$ is $d =2$.  
\STATE \textbf{Output:} All positive weights \(\hat{w}_1, \dots, \hat{w}_B\), the associated masses $\hat{\sigma}_1, \dots, \hat{\sigma}_B$ and loss \(\theta\).
\end{algorithmic}
\end{algorithm}

The functions $\calH_j(\cdot)$ are then estimated by
\[
\hat{\calH}_j(h)
=\int \frac{ \tau^j d\hat{F}(\tau)}{ (\lambda - \tau h)^j }   = 
\sum_{b=1}^{B}
\frac{\hat{w}_b \hat{\sigma}_b^j}{(\lambda - \hat{\sigma}_b h)^j},
\qquad
j=1,2,3,
\quad
0<h<\lambda/\hat{\sigma}_1.
\]
Here, $\eta$ is estimated by $\lambda/\hat{\sigma}_1$. Further, the function $g(h)$ is estimated by 
\[ \hat{g}(h) = h + \Big[ 1+ q \hat{\calH}_1(h) \Big]^{-1}.\]
Further, recall the definition of $\rho$ as in \eqref{eq:def_rho}. We replace $g(h)$ by $\hat{g}(h)$ and estimate $\rho$ as follows. 
(i) If $q \hat{w}_1 \geq 1$, then $\hat{g}'(h) >0$ for all $h < \lambda/\hat{\sigma}_1$ and we estimate $\rho$ by 
\[ \hat{\rho} = \lim_{h \to \lambda/\hat{\sigma}_1} \hat{g}(h) = \lambda/\hat\sigma_{1}.\]
(ii) If $q\hat{w}_1 <1$, then there exists a unique root $\hat{h}_0< \lambda/\hat{\sigma}_1$ to 
\[ q \hat{\calH}_2(\hat{h}_0) =  \Big( 1+ q \hat{\calH}_1(\hat{h}_0)  \Big)^2.\]
Then, we estimate $\rho$ by 
\[  \hat{\rho}  = \hat{g}(\hat{h}_0) =  \hat{h}_0 + \Big[  1+ q \hat{\calH}_1(\hat{h}_0) \Big]^{-1}. \]


\subsection{Consistency}\label{subsec:consistency}

In this section, we establish the consistency of the proposed estimators. The results are adapted from Theorem 3.2, Lemma 3.3, and Lemma 3.4 of \cite{li2025ridge} whose proofs are therefore omitted. For a discrete grid of complex values $J = \{t_1, t_2, \dots\}$, we define its grid size by 
$\operatorname{size}(J) = \sup_{k\geq 2} |t_k - t_{k-1}|$.
Suppose that $\mathcal{C}$ is a compact subset of $\mathbb{C}$ and $J$ is a discrete grid. We define the one-sided Hausdorff distance between $J$ and $\mathcal{C}$ by  
\[ d_H(J, \mathcal{C}) = \sup_{z\in\mathcal{C}} \inf_{t\in J} |t - z|. \]


\begin{theorem}
\label{thm:consistency_estimators}
Suppose that Conditions~\ref{enum:high_dimension_regime}--\ref{enum:edger_regularity} hold. Let $\hat{\Theta}_1(\lambda)$ and $\hat{\Theta}_2(\lambda)$ denote the estimators constructed from the grids $R_n=\{\sigma_b\}$ and $J_n=\{z_i\}$.  Let $\varphi(J_n)$ be the image set $\{\varphi(z): z\in J_n\}$.  We impose the following additional conditions. As in Theorem \ref{thm:main_diverge_k}, assume that $k/n \to \gamma \in (0,\alpha)$. 
Define
\[
\mathfrak{S}
=
\Big[
\liminf_{n\to\infty}\ell_{\min}(\Sigma_0),
\,
\limsup_{n\to\infty}\ell_{\max}(\Sigma_0)
\Big].
\]
We assume that the grid $R_n$ provides an increasingly dense approximation to $\mathfrak{S}$ in the sense that
\[
d_H(R_n,\mathfrak{S})
=
o(n^{-2/3}).
\]
Then, depending on the design of the grid $J_n$, the following conclusions hold.

\begin{itemize}
    \item[(i)]
    Suppose there exists a compact and open subset $\mathcal{C}\subset\mathbb{C}^+$ such that
    $d_H\bigl(\varphi(J_n),\mathcal{C}\bigr)=o(1)$. 
    Then
    \[
    \hat{\Theta}_1(\lambda)-\Theta_1(\lambda,q)
    =
    o_p(1)
    \qquad
    \text{and}
    \qquad    
    \hat{\Theta}_2(\lambda)-\Theta_2(\lambda,q)
    =
    o_p(1).
    \]

    \item[(ii)]
    Let $h_\beta$ denote the solution to \eqref{eq:determine_s} with $g_0=\beta$. Suppose that, for some constant $\varepsilon>0$,
    \[
    d_H\Big(
    \varphi(J_n),
    [h_\beta/\lambda-\varepsilon,\,
    h_\beta/\lambda+\varepsilon]
    \Big)
    =
    o(n^{-2/3}).
    \]
    Then
    \[\hat{\Theta}_1(\lambda)-\Theta_1(\lambda,q)= o_p(n^{-2/3})
    \qquad
    \text{and}
    \qquad
    \hat{\Theta}_2(\lambda)-\Theta_2(\lambda,q)=o_p(1).\]
\end{itemize}
\end{theorem}

To further illustrate the two conditions imposed on $J_n$, note that for any $F^{\Sigma_0}$ and $q$, one can construct a sequence of grids $\{J_n\}$ satisfying the condition in part~(i). In this case, the proposed estimators are consistent, although the convergence rate may be slower than $n^{-2/3}$. By contrast, part~(ii) provides a sufficient condition for achieving the sharper $o_p(n^{-2/3})$ convergence rate. The condition in part~(ii), however, may fail depending on the location of $\beta$. In particular, when $k$ is excessively large, the corresponding value of $h_\beta$ may lie too close to the boundary point
$\eta={\lambda}/{\ell_{\max}(\Sigma_0)}$,
making it impossible to construct a suitable sequence $\{J_n\}$ satisfying the required approximation condition. Conversely, the following lemma shows that the proposed estimators achieve $o_p(n^{-2/3})$ accuracy provided that the selected $k$ is not excessively large. Recall that as in Theorem \ref{thm:main_diverge_k}, $k/n\to\gamma\in(0,1)$.
\begin{lemma}
\label{lemma:consistency_condition3}
Under the conditions of Theorem~\ref{thm:consistency_estimators}, there exists a constant $\gamma_{0}>0$ such that, for any
$\gamma < \gamma_{0}$, 
a sequence of grids $\{J_n\}$ can be constructed to satisfy the condition in part~(ii) of Theorem~\ref{thm:consistency_estimators}.
\end{lemma}

It is also of interest to investigate the behavior of the proposed estimators under the alternative hypothesis. Since the estimation procedure relies solely on the eigenvalues of $\bW_{k2}$ together with the prescribed grids $\{\sigma_b\}$ and $\{z_i\}$, its behavior under $H_a$ depends on the extent to which the signal component perturbs the spectrum of $\bW_{k2}$. Under Condition~\eqref{eq:consistency_eq1}, the perturbation induced by the signal term $M\Lambda_m^T\bY$ is asymptotically negligible. Consequently, the proposed estimators continue to be consistent under the alternative hypothesis.

\begin{lemma}
\label{lemma:consistency_under_Ha}
Under $H_a$, suppose that Condition~\eqref{eq:consistency_eq1} holds. Assume further that the conditions in Part~(i) of Theorem~\ref{thm:consistency_estimators} are satisfied. Then,
\[
\hat{\Theta}_1(\lambda,q)-\Theta_1(\lambda,q)=o_p(1),
\qquad
\hat{\Theta}_2(\lambda,q)-\Theta_2(\lambda,q)=o_p(1).
\]
\end{lemma}


\subsection{Implementation details}\label{sec:additional_details_estimation}

\emph{Implementation details for Algorithm \ref{algo:ODE}}. The proposed ODE requires high accuracy on the initial condition ${s}_0$. We propose using Newton-Raphson method to find $s_0$, starting from the initial point $[(n-1-k)/p_1] [  (\lambda {\varphi}(-\lambda))^{-1} - 1]$. The ODE can then be solved by traditional numerical methods such as the fourth-order Runge-Kutta method (RK4). In our simulation studies, the \texttt{R} package \texttt{deSolve} is used for implementation.

\emph{Implementation details for Algorithm \ref{algo:linear_program}}. 
The recommended choice of the grid \(\{\sigma_b\}_{b=1}^B\) is equally spaced on \([{\ell}_{\min},~ {\ell}_{\max}]\), with \(B \approx 500\). Here, $\ell_{\min}$ and $\ell_{\max}$ are the smallest and largest eigenvalues of $\bW_{k2}$, respectively. The choice of \(B\) strikes a balance between computational efficiency and estimation accuracy.
The recommended choice of \(\{z_i\}_{i=1}^I\) is such that \(\Re({\hat\varphi}(z_i))\) are evenly spaced on \([{\hat\varphi}(1.05{\ell}_{\max} ), {\hat\varphi}(-\lambda)]\), and $\Im({\hat\varphi}(z_i)) = 10^{-2}\times {\ell}_{\max}^{-1}$. We can first select the value of ${\hat\varphi}(z_i)$ and numerically find the corresponding $z_i$ using optimization methods like Newton-Raphson. While more points can potentially improve accuracy, we recommend setting \(I \approx 500\) to balance computational efficiency and estimation precision. 
Traditional linear programming solvers can be used for implementation. In our simulation, we utilize the \texttt{R} package \texttt{Rglpk}, which implements the GNU Linear Programming Kit. With the recommended settings, the problem can be efficiently handled on a typical PC (\(\leq 10\) seconds).

\subsection{Further discussion}
Algorithm~\ref{algo:linear_program} is inspired by the approaches of \citet{el2008spectrum} and \citet{ledoit2012nonlinear}, but differs substantially in both the choice of basis representation and the design of the loss function. The methods of \citet{el2008spectrum} and \citet{ledoit2012nonlinear} are primarily designed to produce smooth estimators of $F^{\Sigma_0}$ through combinations of smooth basis functions and point masses. In contrast, our objective is not to recover $F^{\Sigma_0}$ itself, but rather to accurately estimate the functionals $\calH_j(h)$. From this perspective, smoothness of the estimator of $F^{\Sigma_0}$ is not essential. 
Accordingly, we adopt a simpler representation that approximates $F^{\Sigma_0}$ using only mixtures of point masses. Our numerical experiments indicate that, provided the grid $\{\sigma_b\}$ is sufficiently dense, this approach performs comparably to, and in some cases better than, approaches based on smooth basis functions for estimating $\calH_j(h)$. This suggests that the additional complexity introduced by smooth basis representations provides limited practical benefit in the present setting. Consequently, we focus exclusively on point-mass mixture models throughout the proposed estimation procedure.

The proposed method also differs from existing approaches in the construction of the loss function. The procedures of \citet{el2008spectrum} and \citet{ledoit2012nonlinear} primarily control the discrepancy between $Q_1$ and $\tilde{\calH}_1$. By contrast, our method explicitly penalizes discrepancies involving both $Q_1$ and $Q_2$, since accurate estimation of $\calH_2$ is directly required in the subsequent analysis.

It has been observed in several studies that the procedure of \citet{el2008spectrum} may perform unsatisfactorily in finite samples when the goal is accurate recovery of $F^{\Sigma_0}$ itself; see, for example, \citet{li2013estimation}, \citet{ledoit2015spectrum,ledoit2017numerical}. Our numerical experiments lead to similar conclusions for the point-mass approximation employed in Algorithm~\ref{algo:linear_program}. Nevertheless, accurate recovery of $F^{\Sigma_0}$ is not the primary objective of the present work. Despite the limitations in estimating the spectral distribution itself, the proposed procedure achieves high accuracy in estimating the target functionals $\calH_j(\cdot)$. Supporting numerical evidence is provided in Section~\ref{sec:simulation}.

We also investigated an alternative approach based on replacing $F^{\Sigma_0}(\tau)$ by the QuEST estimator proposed in \citet{ledoit2015spectrum,ledoit2017numerical}. Although the QuEST approach generally provides a more accurate estimation of $F^{\Sigma_0}$, we find that it is less accurate for estimating the functionals $\calH_j(\cdot)$. Empirically, when the aspect ratio $q$ is moderate, for example $q\lesssim 5$, the two methods yield comparable estimation accuracy for $\calH_j(\cdot)$. In this regime, the proposed procedure is computationally more efficient and is therefore preferred in practice. Moreover, establishing the sharper $o_p(n^{-2/3})$ convergence rate for QuEST-based estimators appears technically challenging.

\section{Additional Simulation Results}\label{sec:addition_simulation_result}
Figures \ref{fig:power_AR_exp_decay}--\ref{fig:power_AR_low_rank} display additional power curves 
\begin{figure}[htbp]
    \centering
    \begin{subfigure}[t]{0.23\linewidth}
        \centering
        \includegraphics[width=\linewidth]{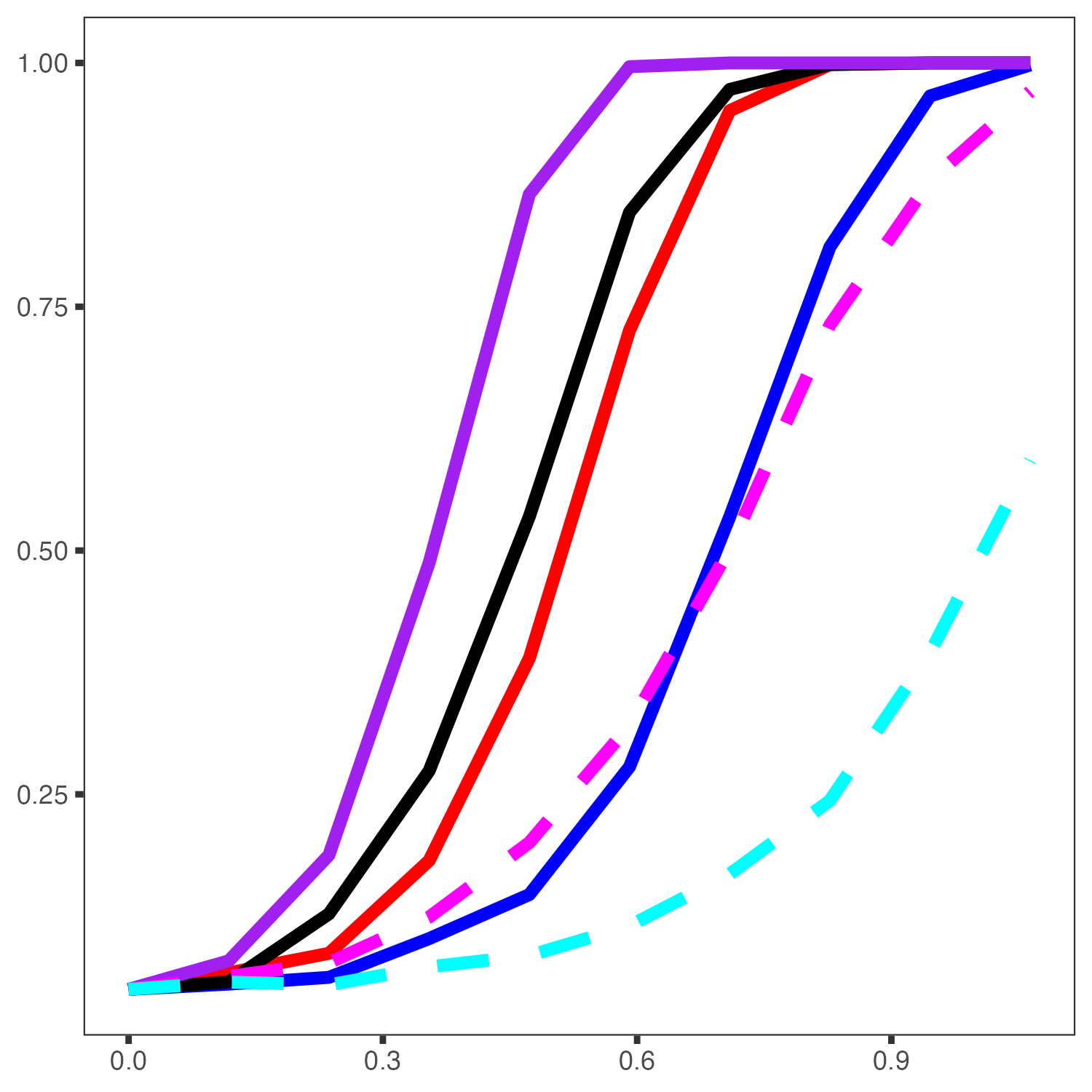}
    \end{subfigure}%
    \begin{subfigure}[t]{0.23\linewidth}
        \centering
        \includegraphics[width=\linewidth]{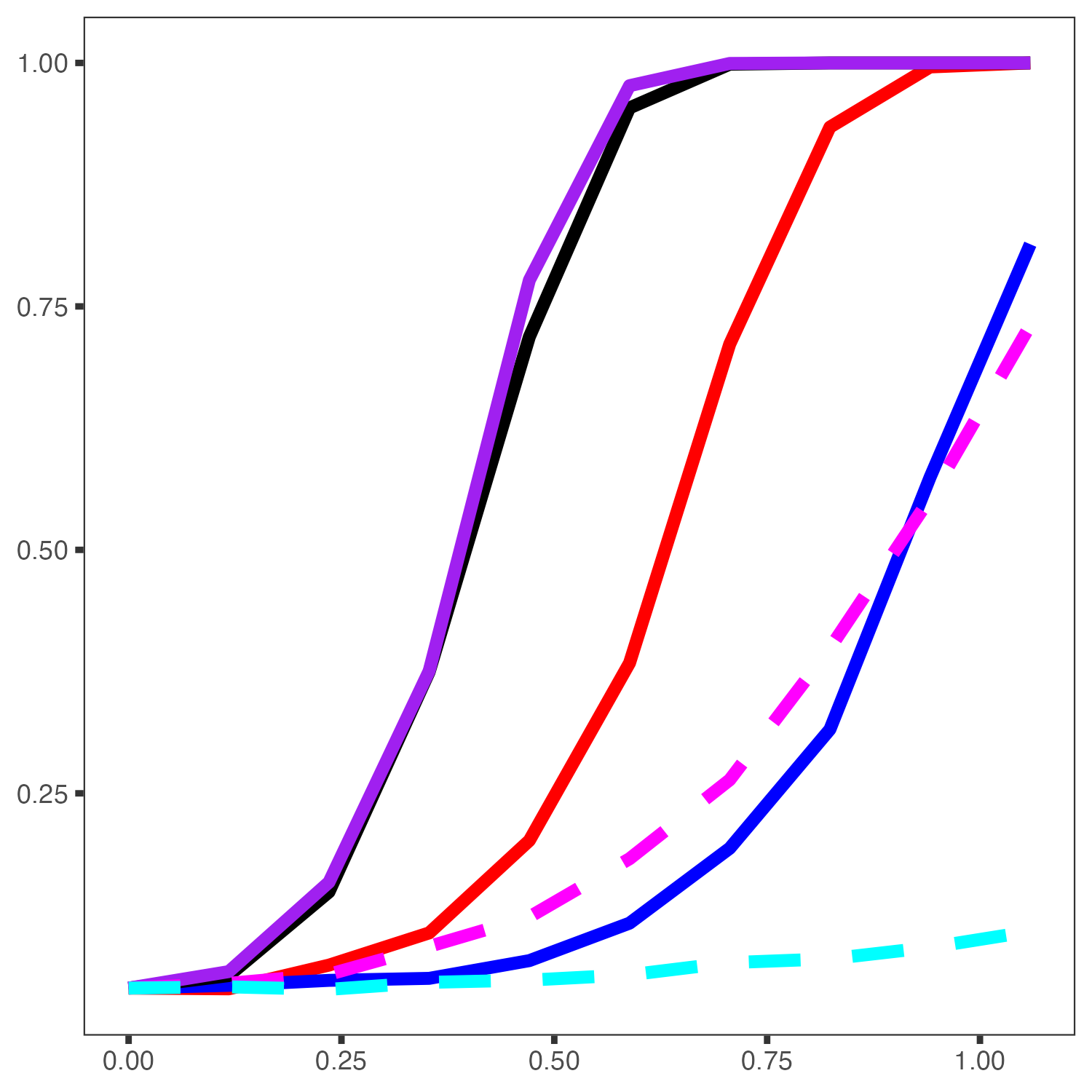}
    \end{subfigure}
    \begin{subfigure}[t]{0.23\linewidth}
        \centering
        \includegraphics[width=\linewidth]{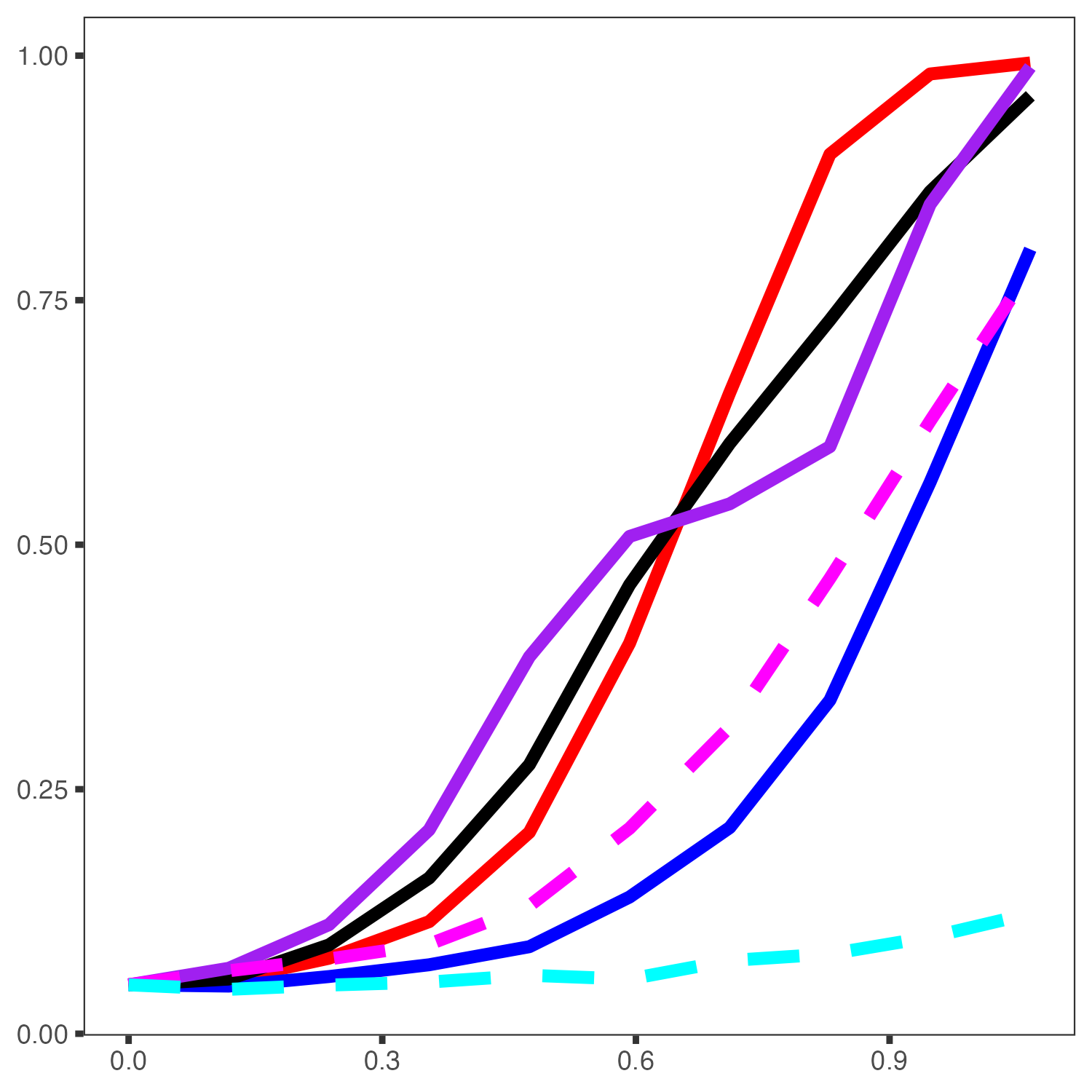}
    \end{subfigure}
    \begin{subfigure}[t]{0.23\linewidth}
        \centering
        \includegraphics[width=\linewidth]{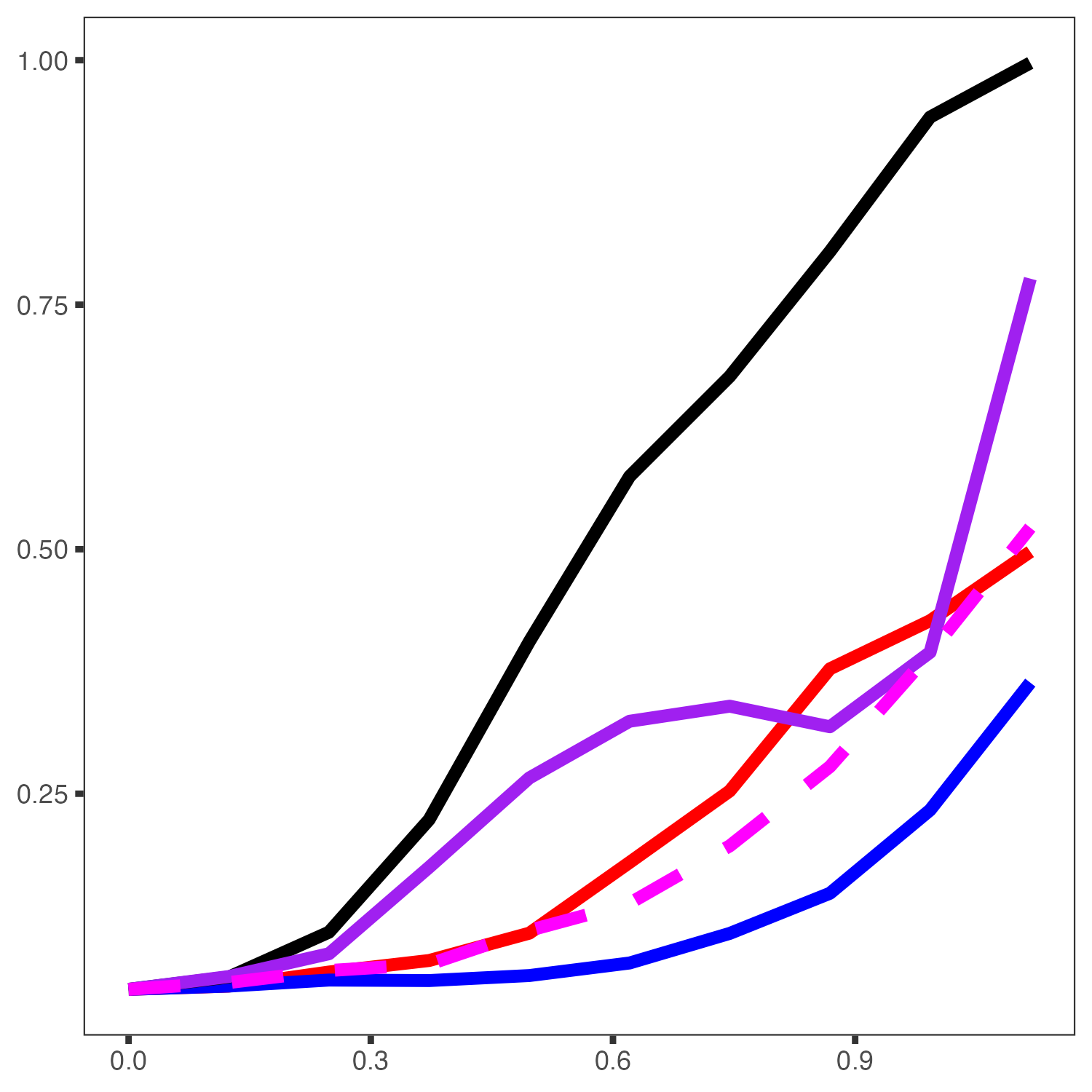}
    \end{subfigure}
    \vfill
    \begin{subfigure}[t]{0.23\linewidth}
        \centering
        \includegraphics[width=\linewidth]{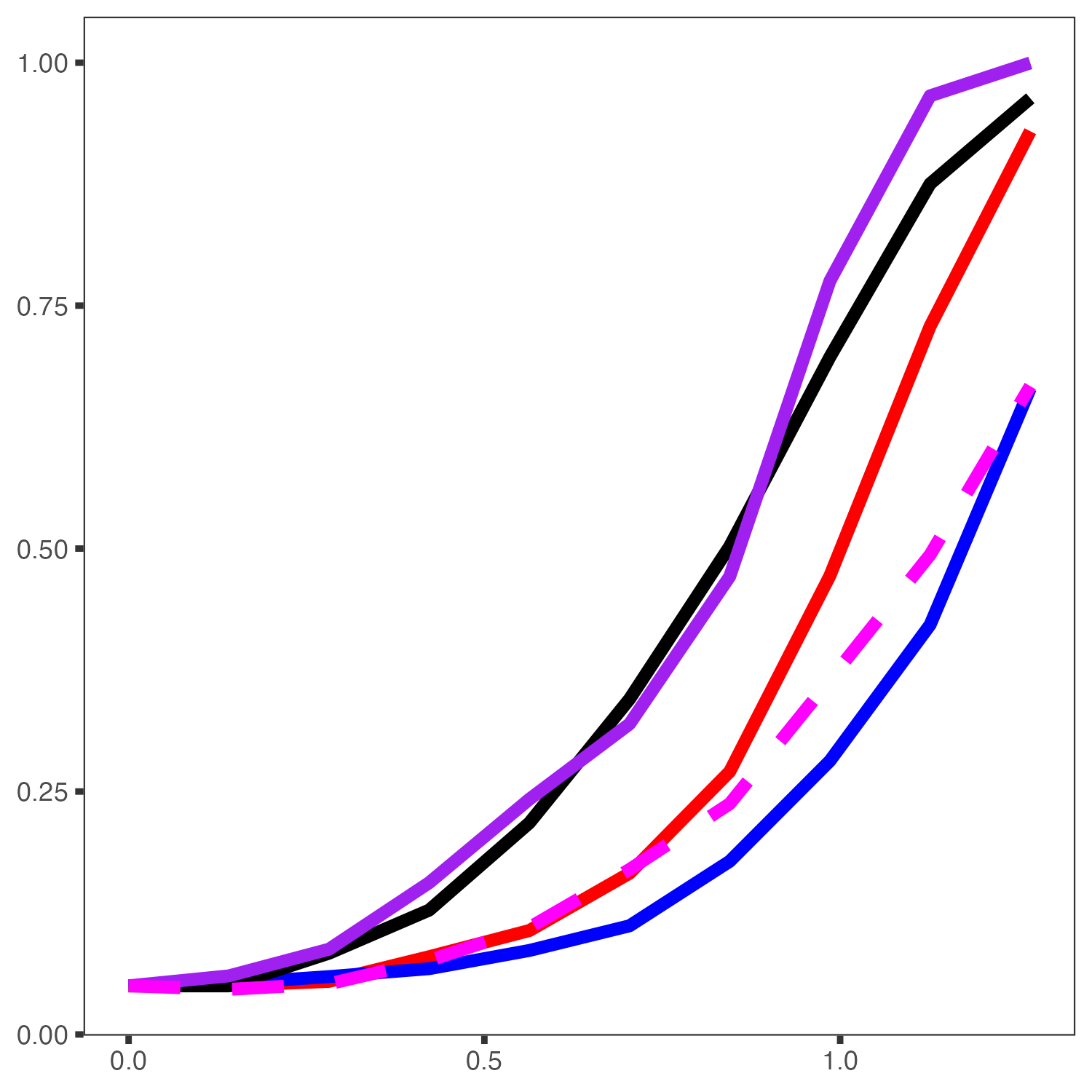}
    \end{subfigure}%
    \begin{subfigure}[t]{0.23\linewidth}
        \centering
        \includegraphics[width=\linewidth]{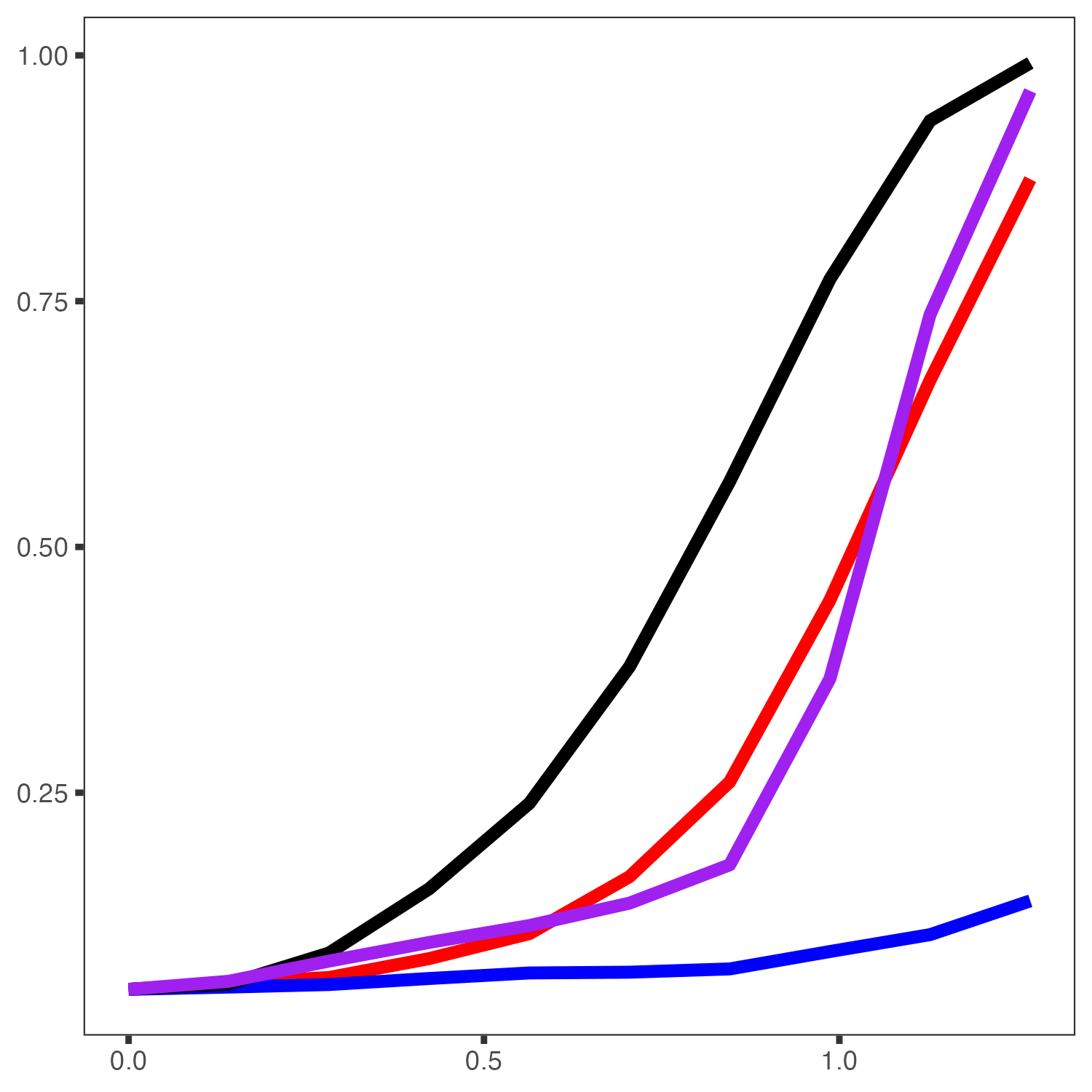}
    \end{subfigure}
    \begin{subfigure}[t]{0.23\linewidth}
        \centering
        \includegraphics[width=\linewidth]{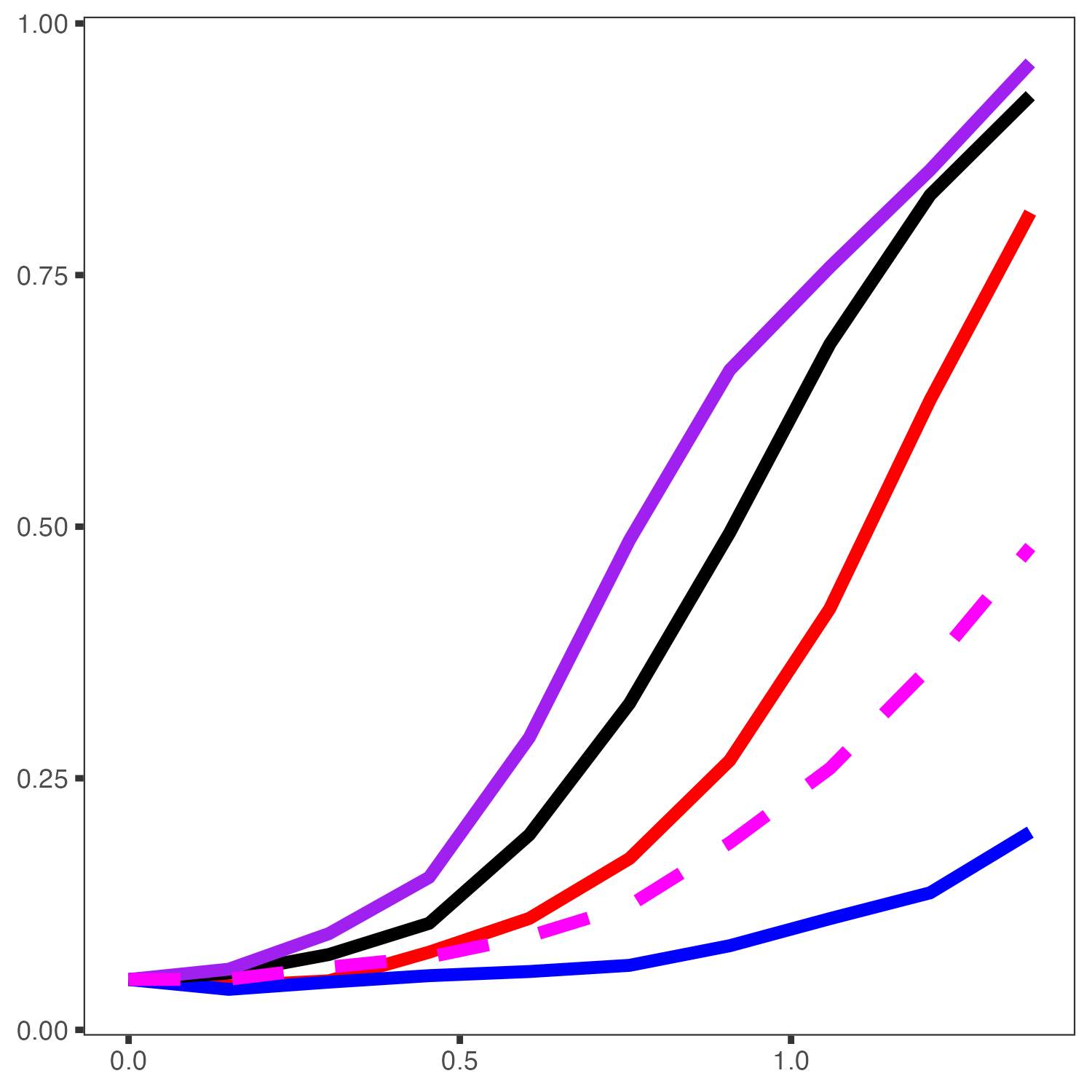}
    \end{subfigure}
    \begin{subfigure}[t]{0.23\linewidth}
        \centering
        \includegraphics[width=\linewidth]{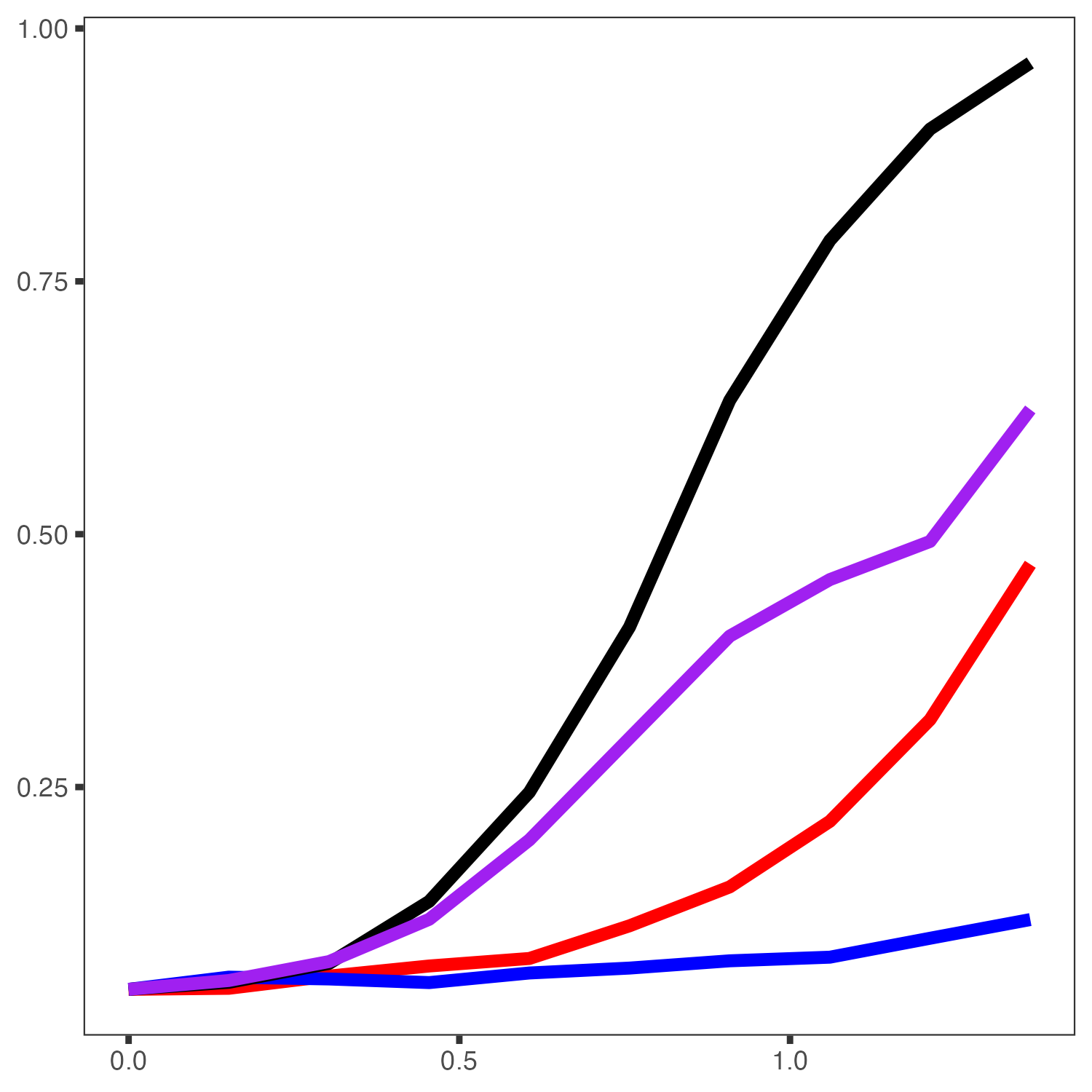}
    \end{subfigure}
    \caption{
    Same as Figure \ref{fig:power_poly_exp_decay} but under $\Sigma_0=\Sigma_{\rm AR}$, and the Exp-Decay correlation model.
    }
    \label{fig:power_AR_exp_decay}
\end{figure}

\begin{figure}[htbp]
    \centering
    \begin{subfigure}[t]{0.23\linewidth}
        \centering
        \includegraphics[width=\linewidth]{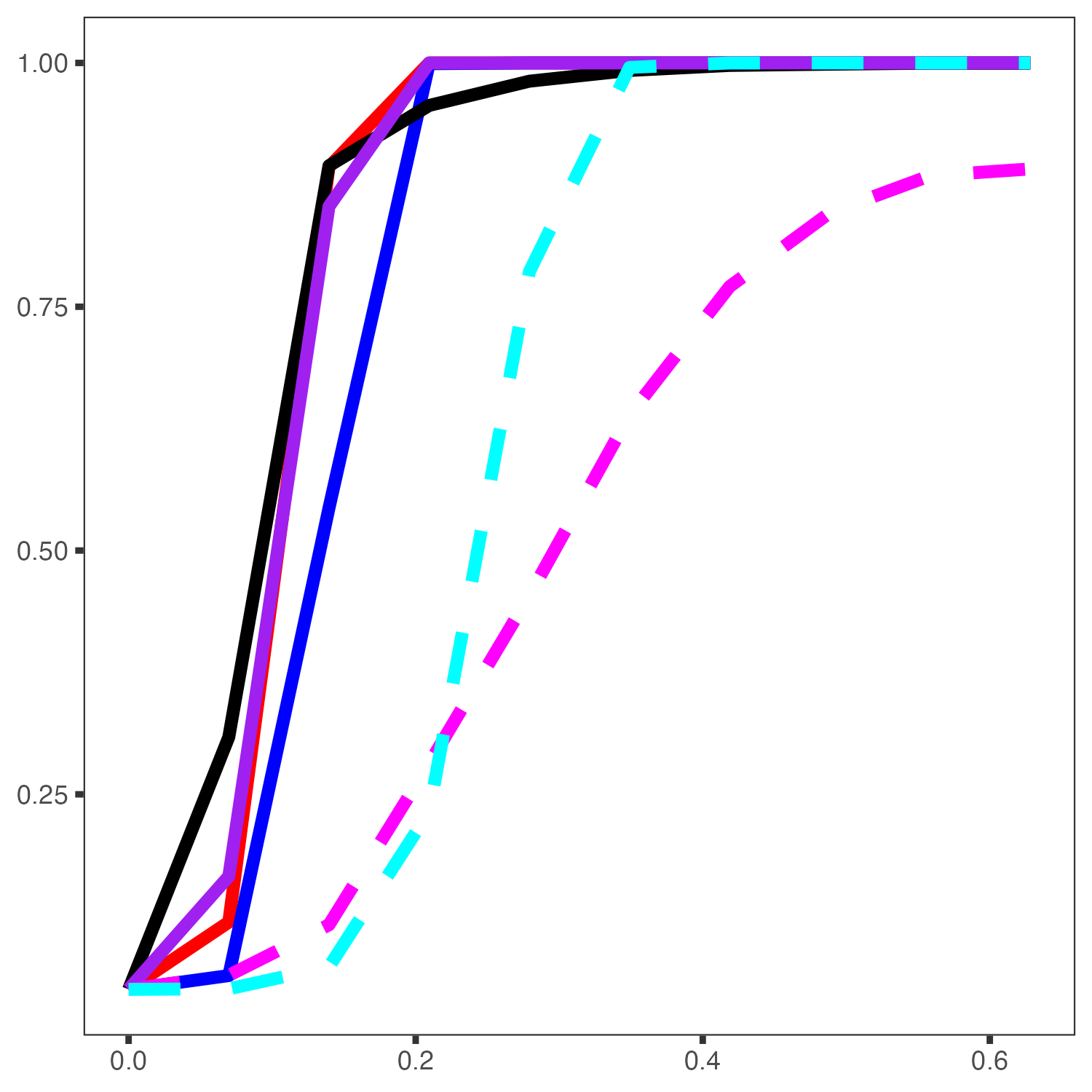}
    \end{subfigure}%
    \begin{subfigure}[t]{0.23\linewidth}
        \centering
        \includegraphics[width=\linewidth]{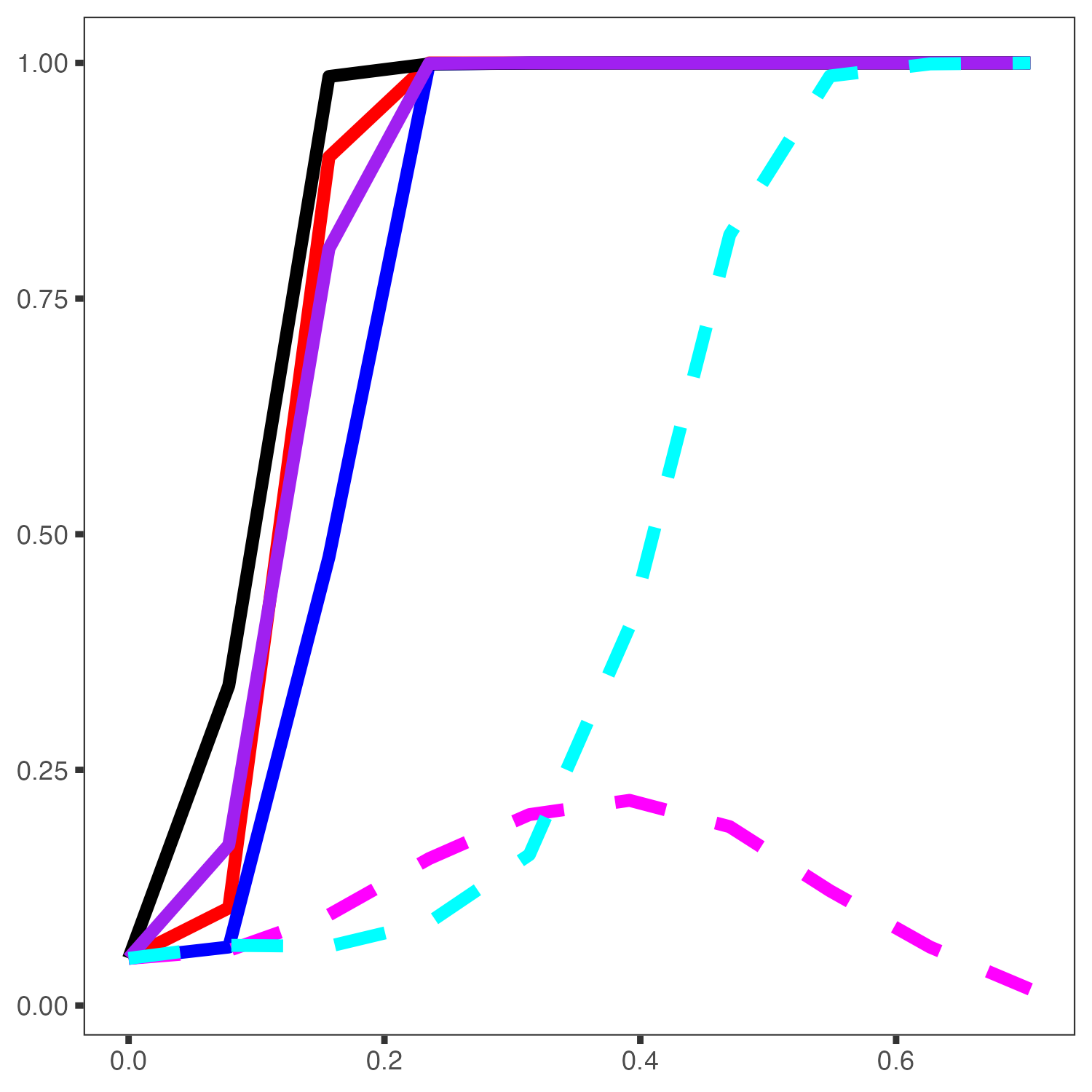}
    \end{subfigure}
    \begin{subfigure}[t]{0.23\linewidth}
        \centering
        \includegraphics[width=\linewidth]{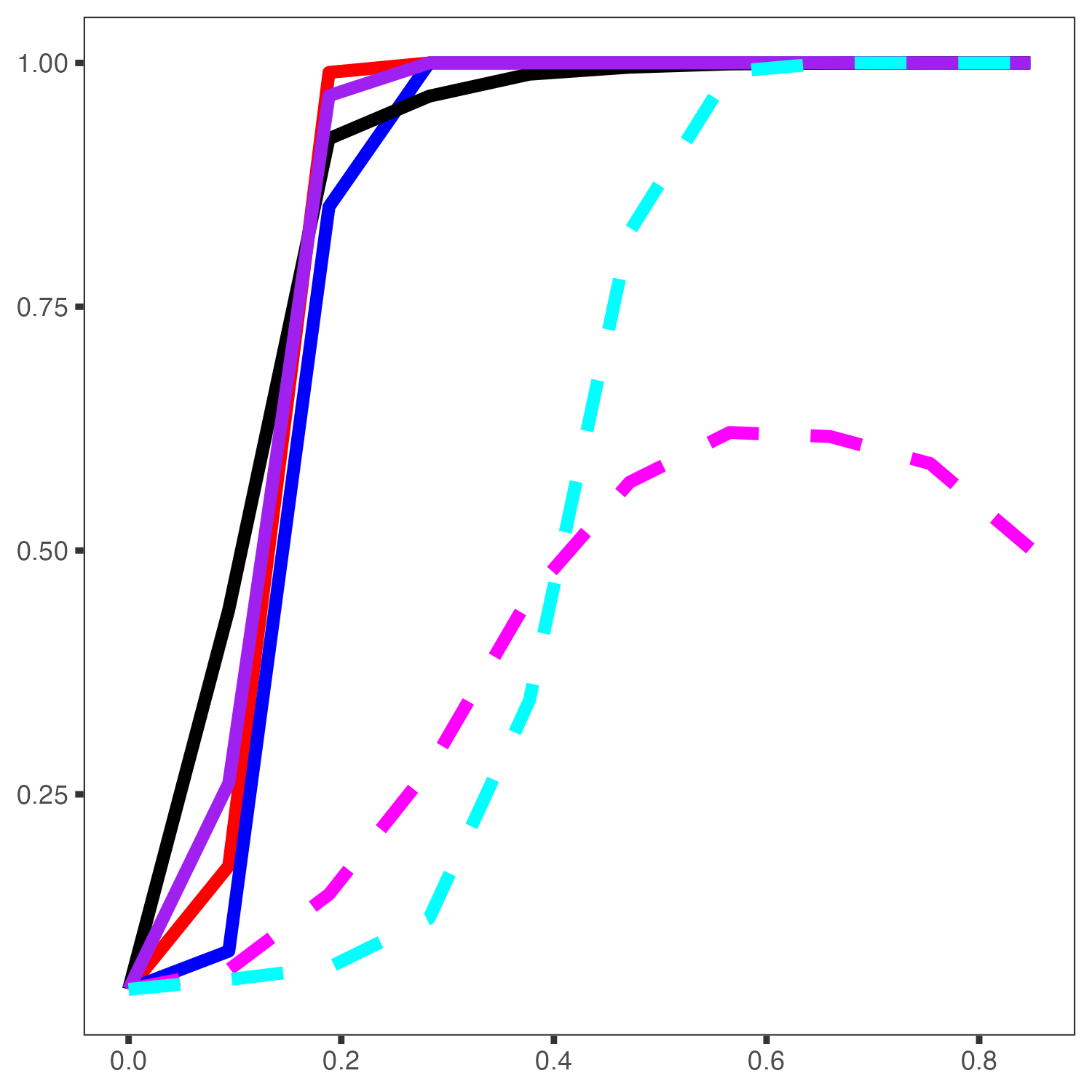}
    \end{subfigure}
    \begin{subfigure}[t]{0.23\linewidth}
        \centering
        \includegraphics[width=\linewidth]{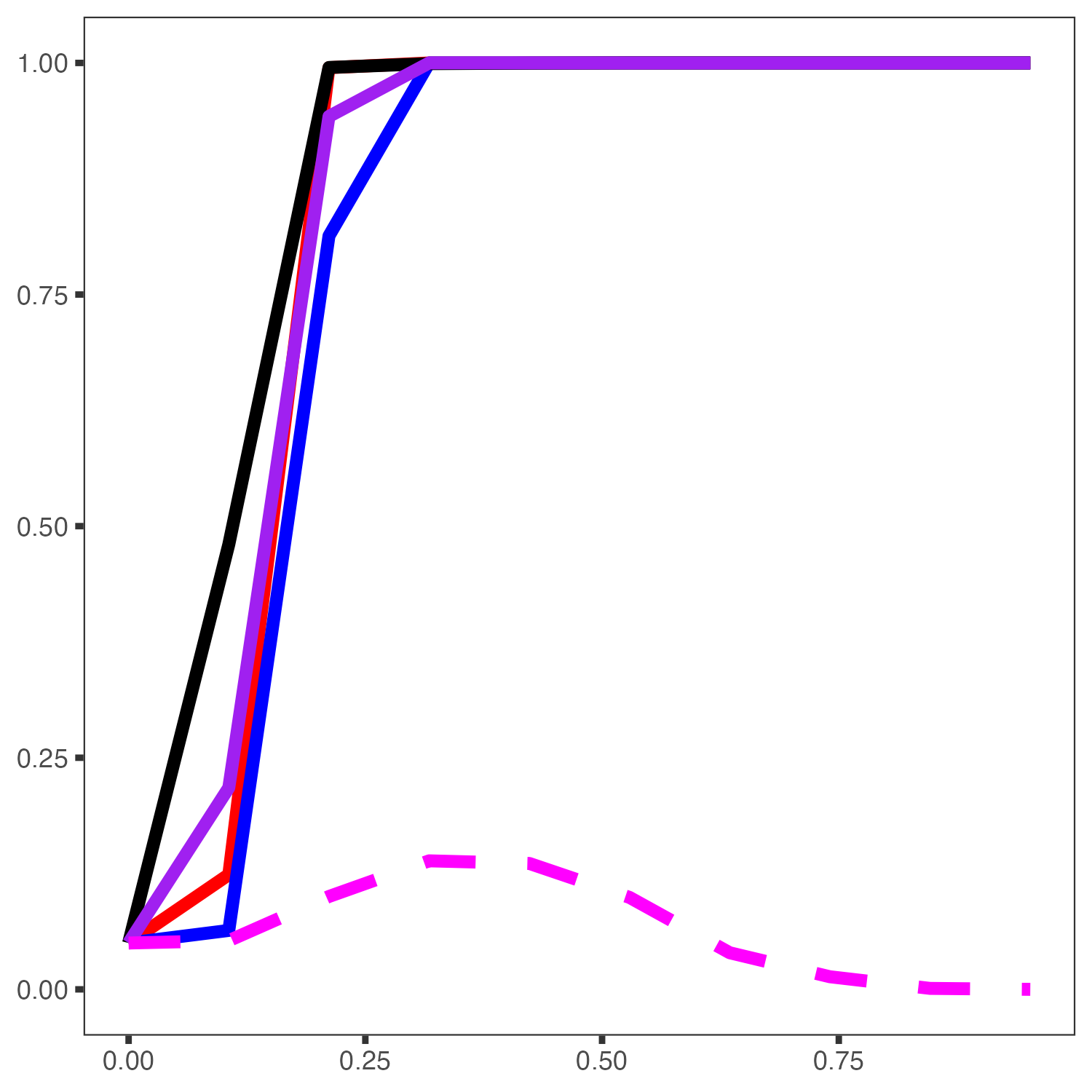}
    \end{subfigure}
    \vfill
    \begin{subfigure}[t]{0.23\linewidth}
        \centering
        \includegraphics[width=\linewidth]{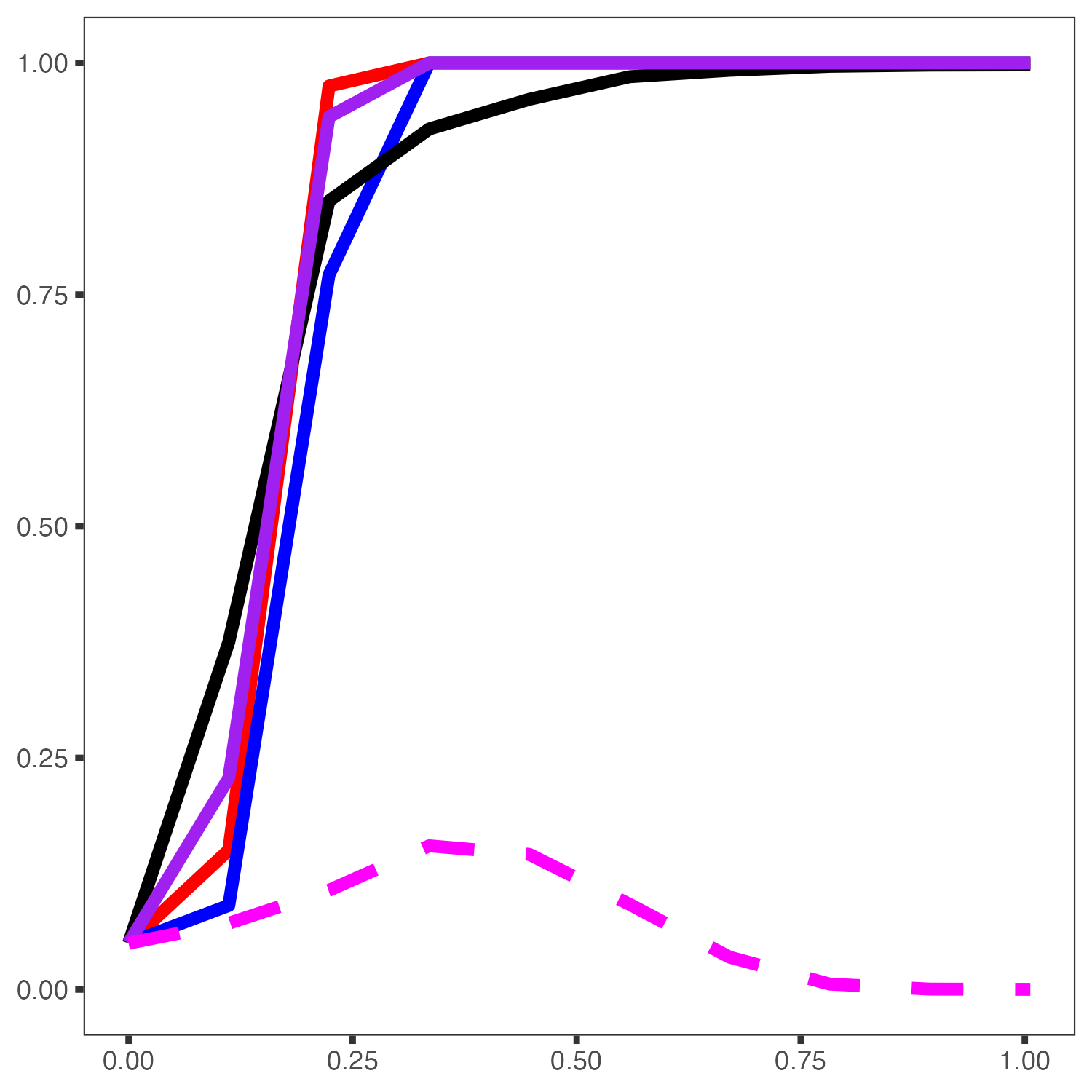}
    \end{subfigure}%
    \begin{subfigure}[t]{0.23\linewidth}
        \centering
        \includegraphics[width=\linewidth]{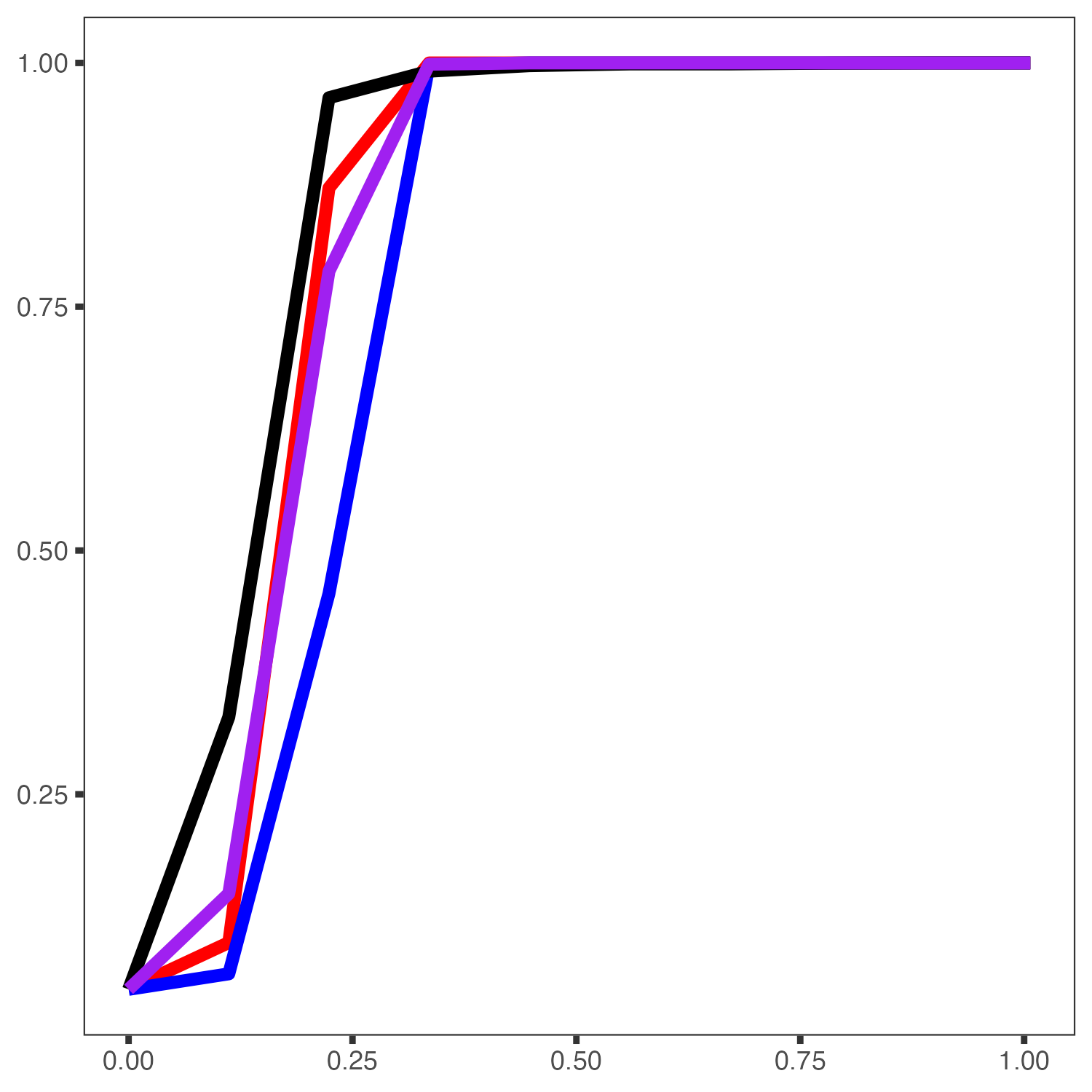}
    \end{subfigure}
    \begin{subfigure}[t]{0.23\linewidth}
        \centering
        \includegraphics[width=\linewidth]{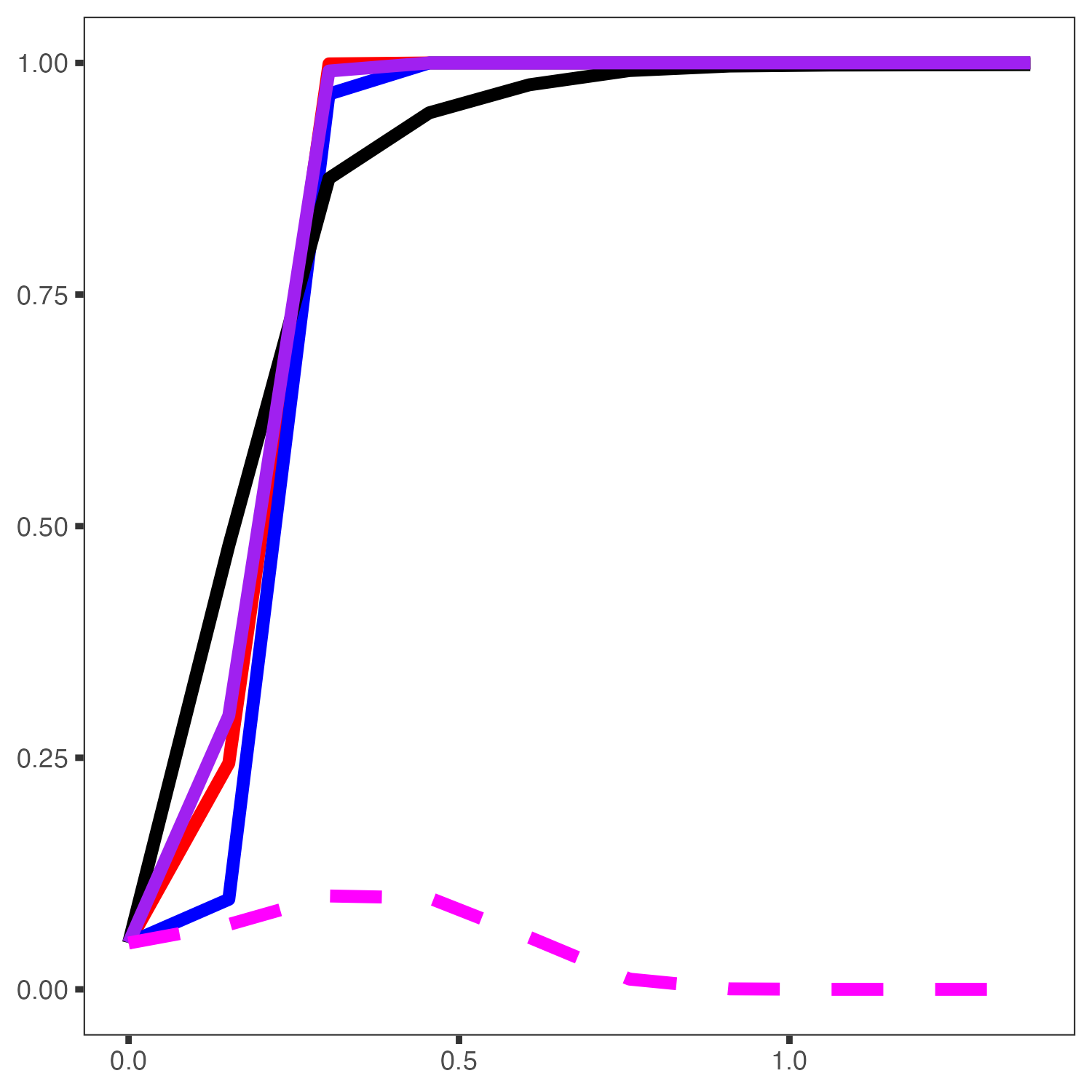}
    \end{subfigure}
    \begin{subfigure}[t]{0.23\linewidth}
        \centering
        \includegraphics[width=\linewidth]{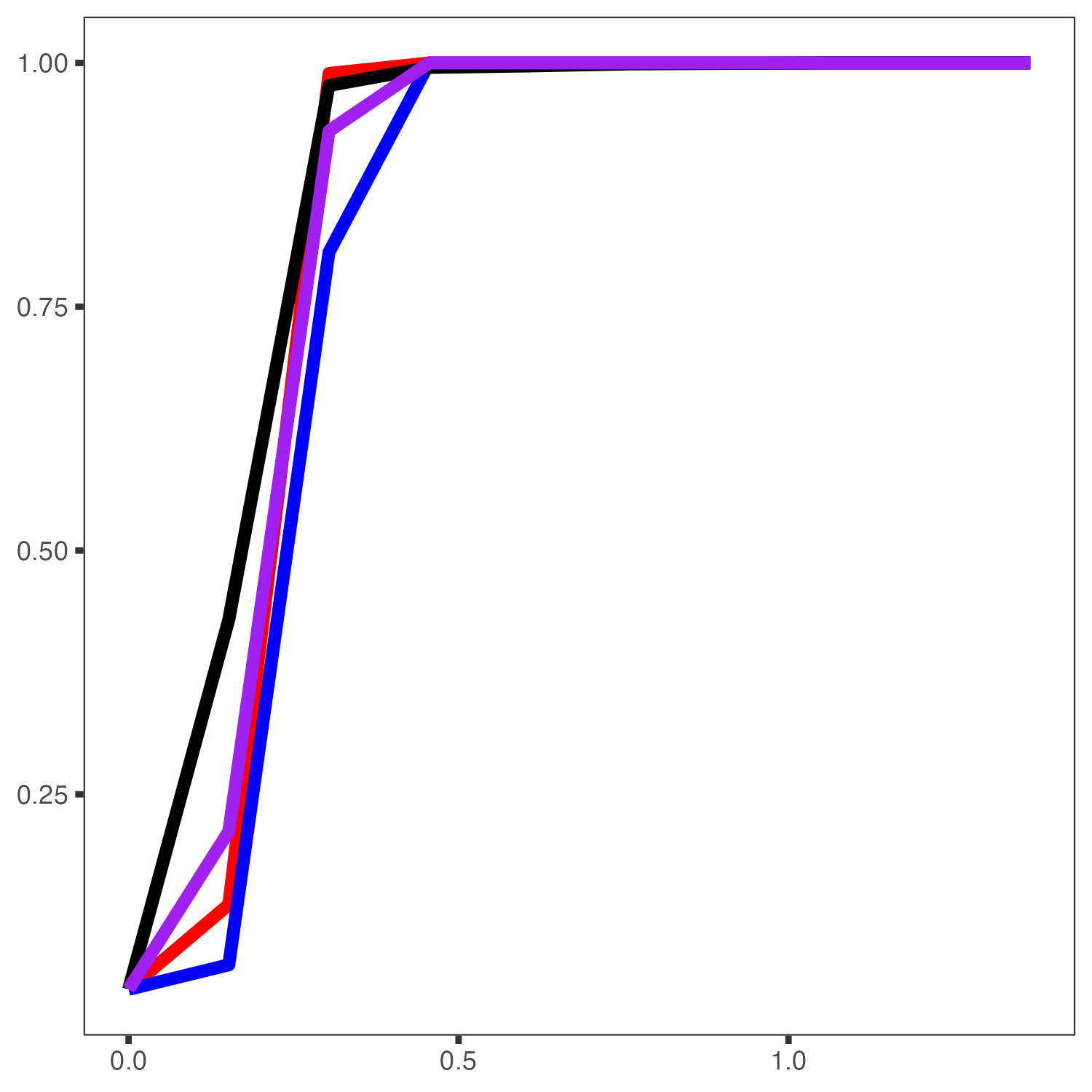}
    \end{subfigure}
    \caption{
    Same as Figure \ref{fig:power_poly_exp_decay} but under $\Sigma_0=\Sigma_{\rm ID}$, and the Low-rank correlation model.
    }
    \label{fig:power_ID_low_rank}
\end{figure}

\begin{figure}[htbp]
    \centering
    \begin{subfigure}[t]{0.23\linewidth}
        \centering
        \includegraphics[width=\linewidth]{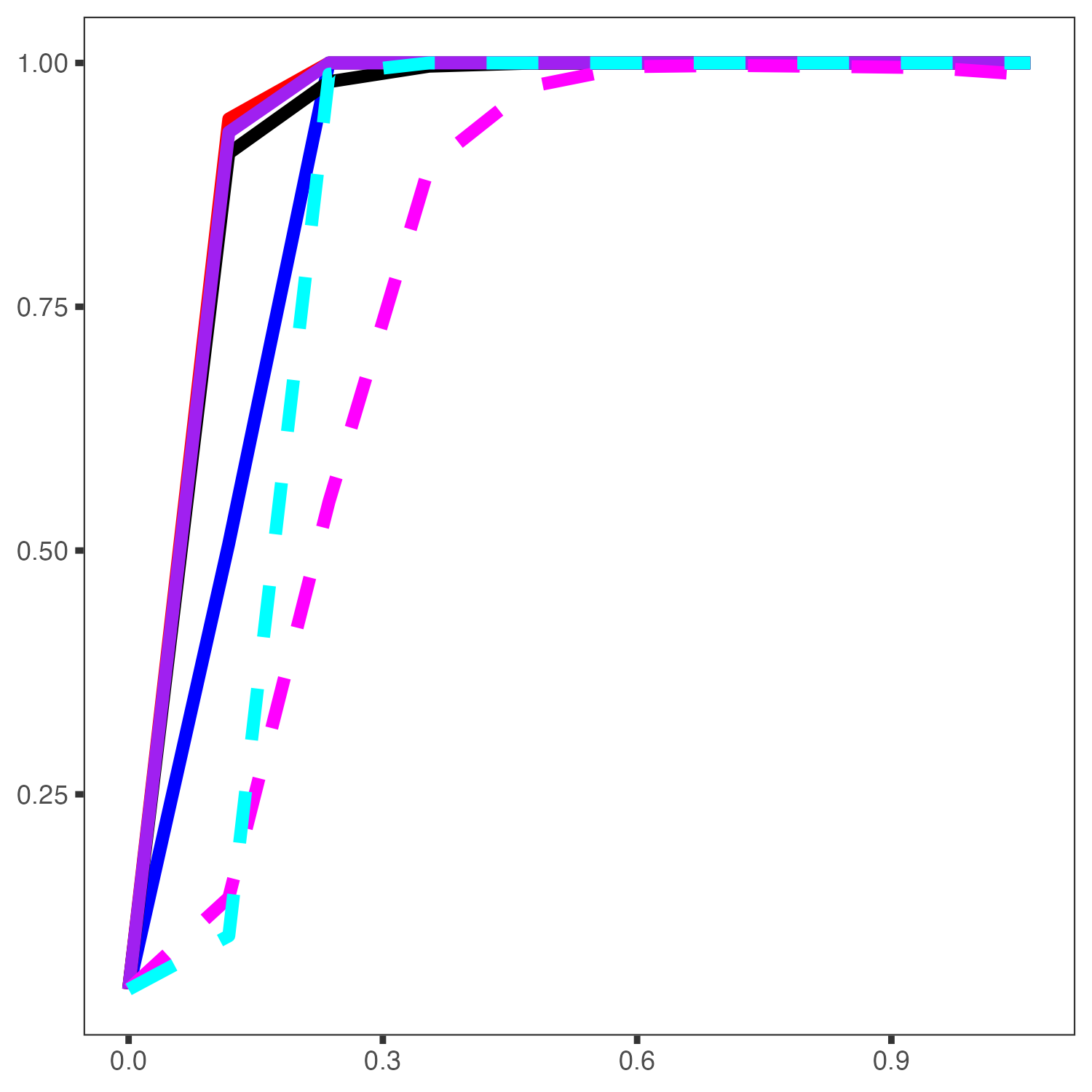}
    \end{subfigure}%
    \begin{subfigure}[t]{0.23\linewidth}
        \centering
        \includegraphics[width=\linewidth]{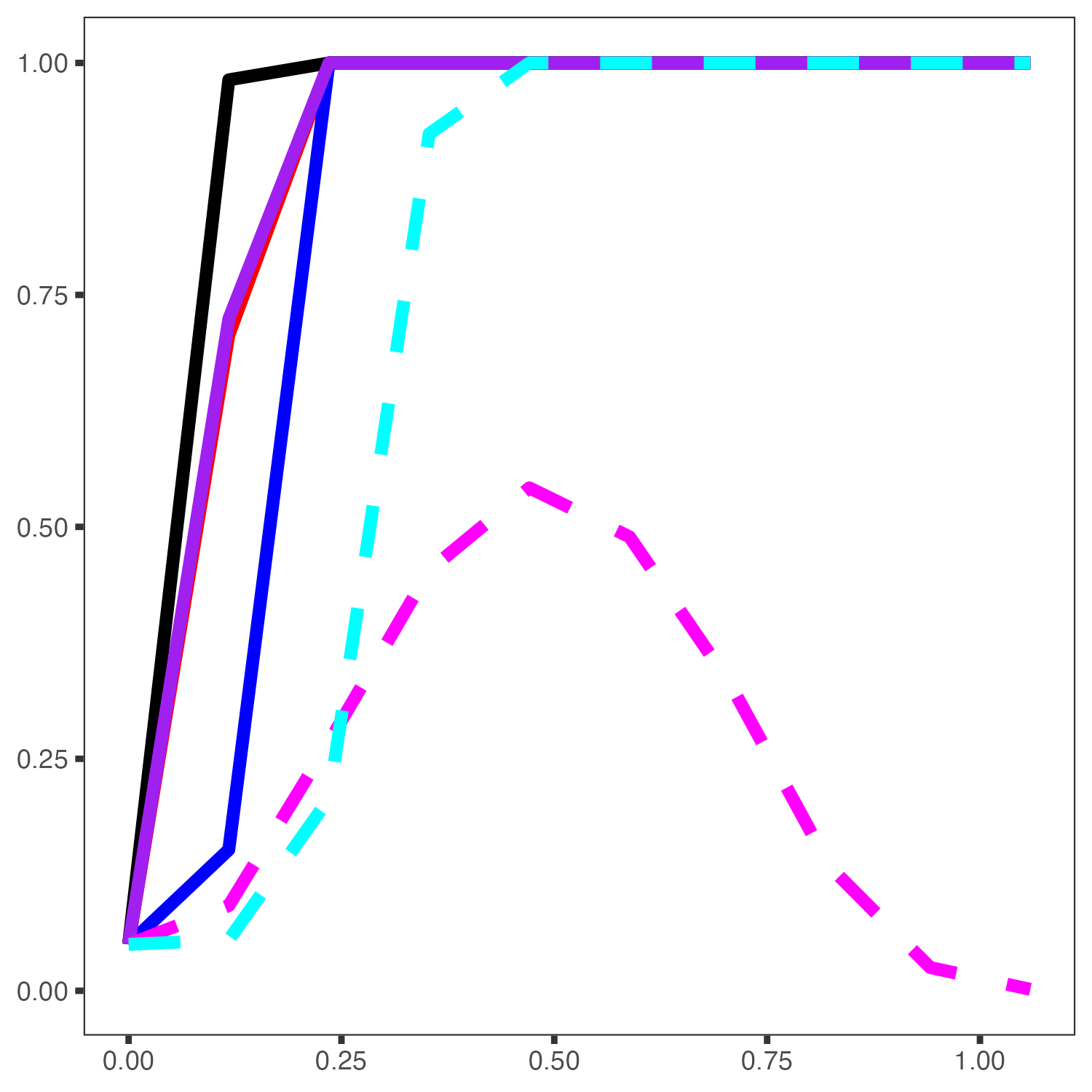}
    \end{subfigure}
    \begin{subfigure}[t]{0.23\linewidth}
        \centering
        \includegraphics[width=\linewidth]{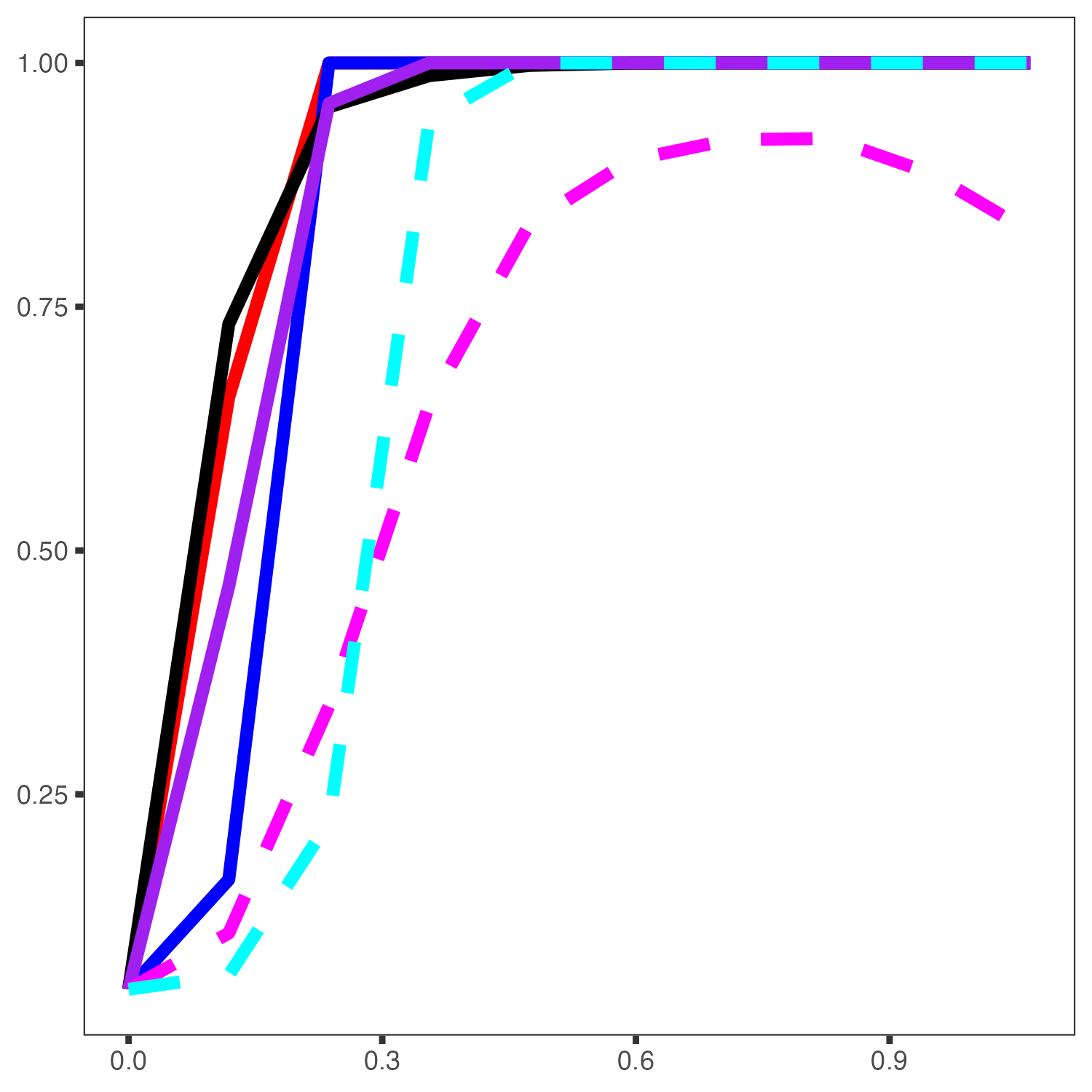}
    \end{subfigure}
    \begin{subfigure}[t]{0.23\linewidth}
        \centering
        \includegraphics[width=\linewidth]{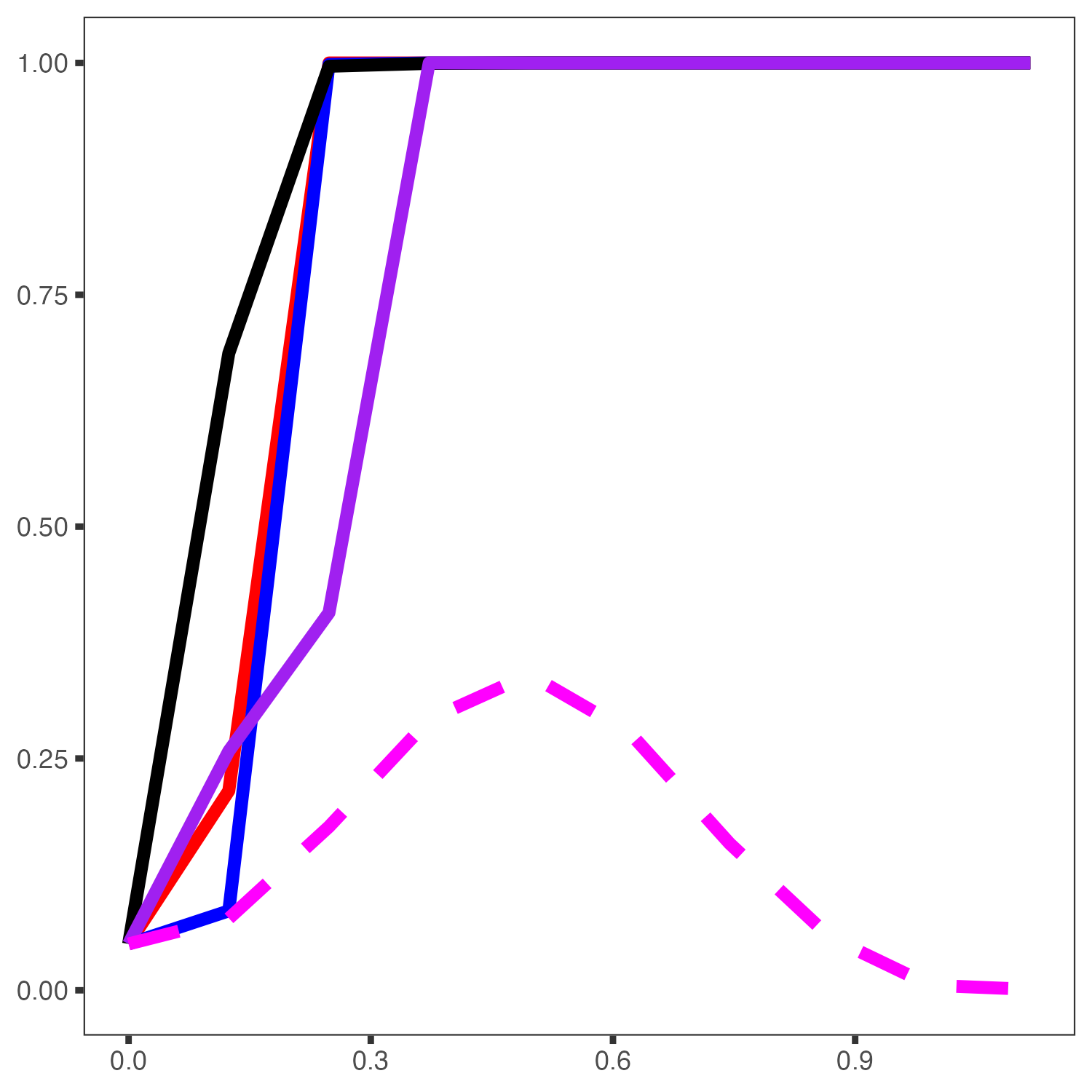}
    \end{subfigure}
    \vfill
    \begin{subfigure}[t]{0.23\linewidth}
        \centering
        \includegraphics[width=\linewidth]{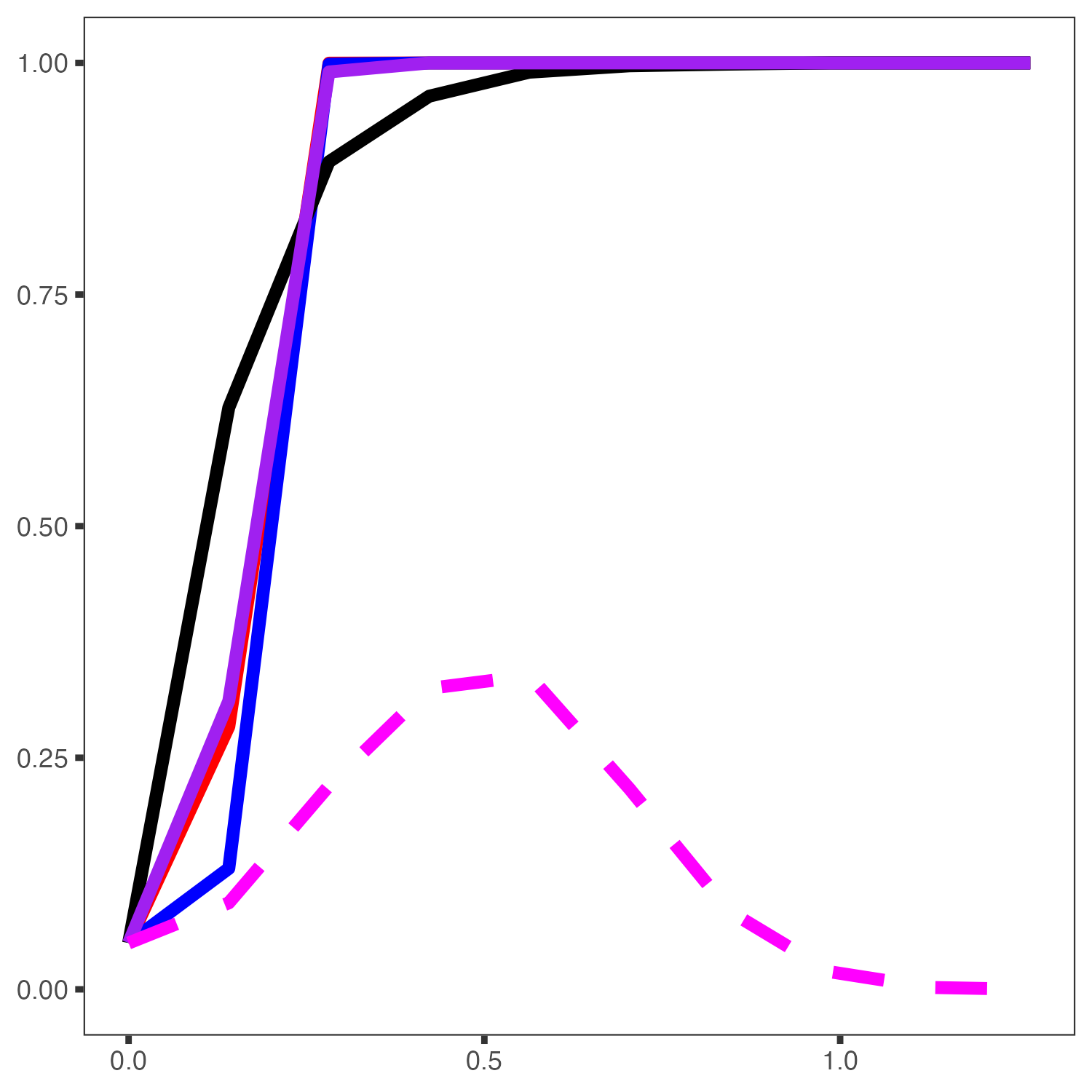}
    \end{subfigure}%
    \begin{subfigure}[t]{0.23\linewidth}
        \centering
        \includegraphics[width=\linewidth]{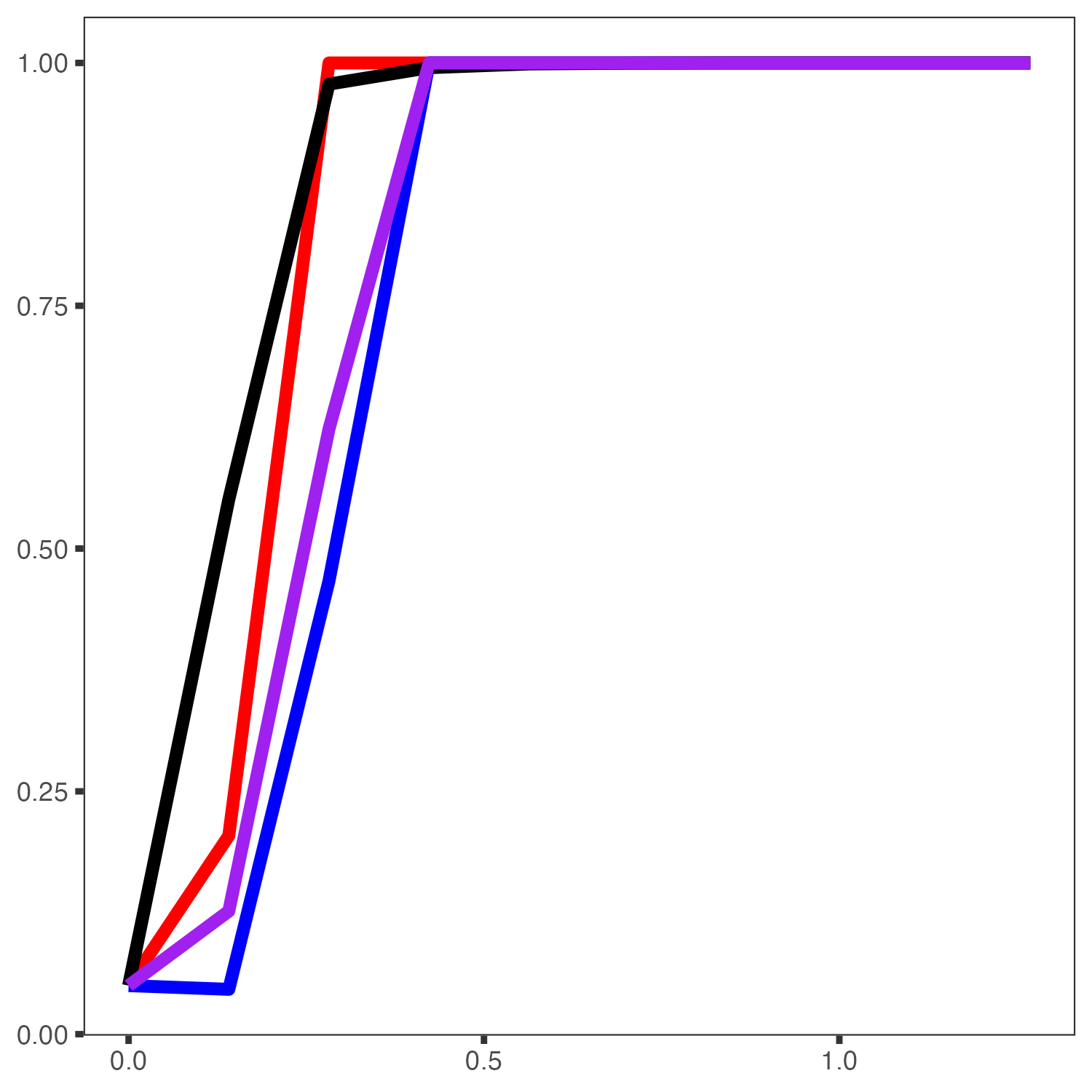}
    \end{subfigure}
    \begin{subfigure}[t]{0.23\linewidth}
        \centering
        \includegraphics[width=\linewidth]{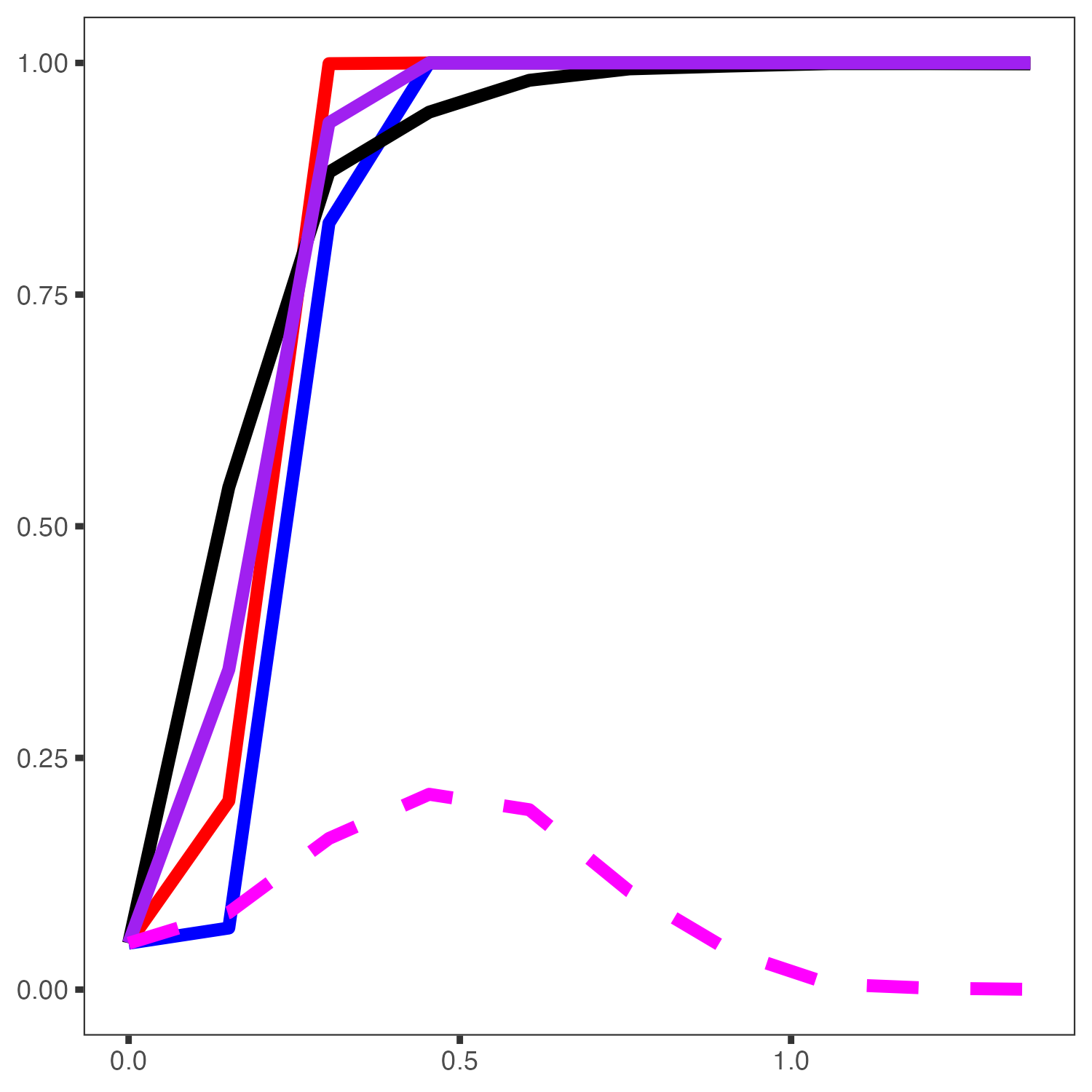}
    \end{subfigure}
    \begin{subfigure}[t]{0.23\linewidth}
        \centering
        \includegraphics[width=\linewidth]{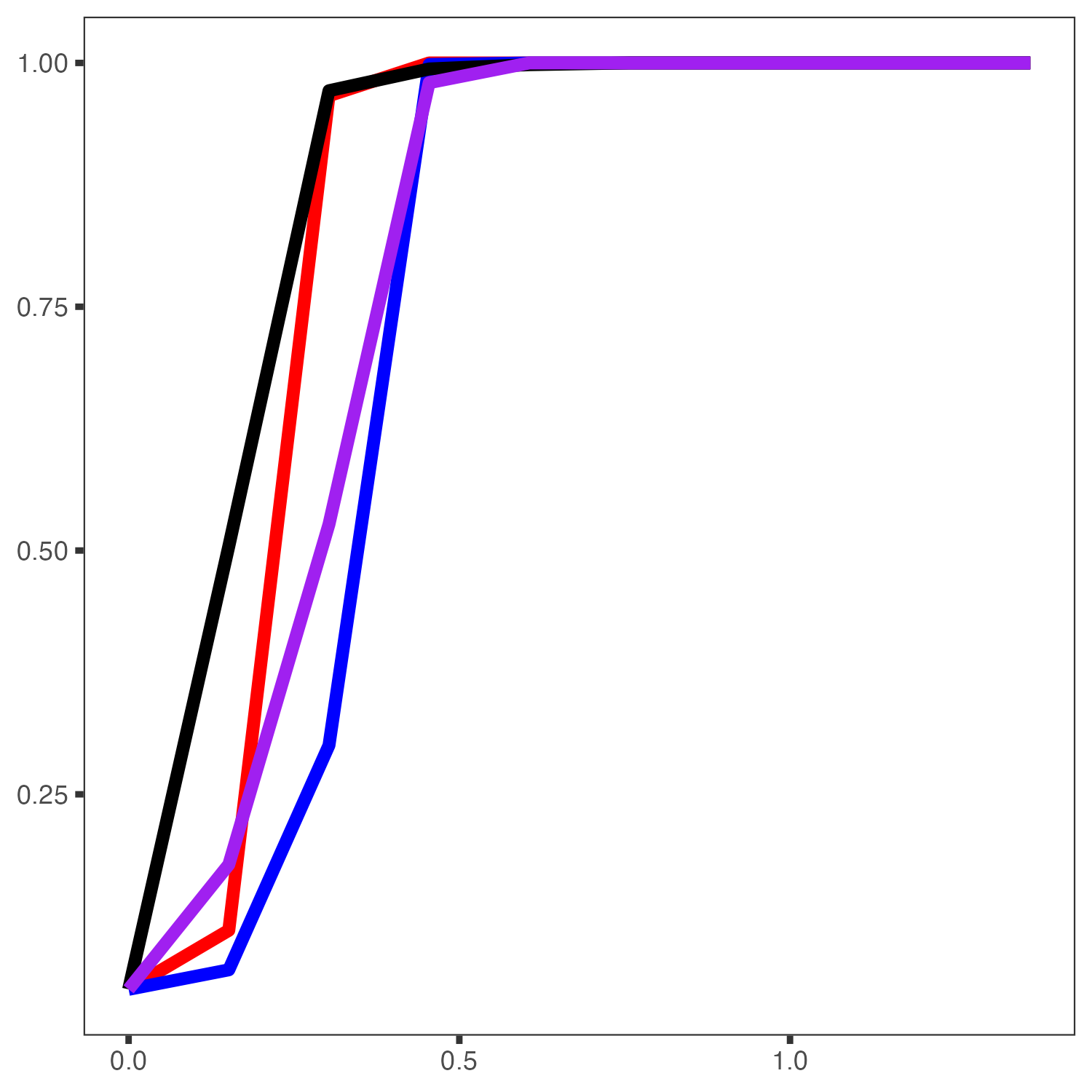}
    \end{subfigure}
    \caption{
    Same as Figure \ref{fig:power_poly_exp_decay} but under $\Sigma_0=\Sigma_{\rm AR}$, and the Low-rank correlation model.
    }
    \label{fig:power_AR_low_rank}
\end{figure}

\section{Technical Proof}\label{sec:technical-proof}
In this section, we present the proofs of the main technical results. The arguments build upon several established results in RMT, including those of \cite{knowles2017anisotropic,han2018unified,li2024analysis,li2025ridge}. To avoid unnecessary repetition, we focus on the main ideas and key steps of the proofs while omitting routine technical details.

First, consider the analysis under the null hypothesis. Since $\bF_{k\lambda}$ is asymmetric, we apply a 
transformation as follows to simplify the analysis. Recall the definition of the regularized $F$-matrix and the projection matrix $\bP_k$ as
\[
\bF_{k\lambda}
=
\bW_{k1}
(\bW_{k2}+\lambda I_{p_1})^{-1}, \qquad \bP_k = \bY^T \hat{\Lambda}_k (\hat{\Lambda}_k^T \bY \bY^T \hat{\Lambda}_k)^{-1} \hat{\Lambda}^T_k \bY.  
\]
Let $P_{\perp} = I_n - n^{-1} 1_n 1_n^T$. Then, the centered data matrices are such that 
\[\bX = M \Lambda_m^T\Sigma_y^{1/2} \bZ_y P_\perp  +   \Sigma_0^{1/2} \bZ_x P_\perp, \qquad \bY =\Sigma_y^{1/2} \bZ_y P_\perp,\] 
where $\bZ_x = (Z_{x1}, \dots, Z_{xn})$ and $\bZ_y = (Z_{y1}, \dots, Z_{yn})$ are two matrices whose columns $Z_{xi}$ and $Z_{yi}$ are iid copies of $Z_x$ and $Z_y$ in Model \eqref{eq:model_PCR}, respectively. Moreover, it is straightforward to verify that 
$P_\perp \bP_k =\bP_k P_\perp = \bP_k$.

Under Condition \ref{enum:lower_bound_eigen_Y}, there exists a universal constant $C>0$ such that the event 
\[\mathcal{E}_{C} = \{\ell_k(\bS_y) > C\} \]
hold with overwhelming probability in the sense that $\mP(\mathcal{E}_C) \geq 1-\exp(-n)$.  In the following analysis, we shall always assume that this event holds so that $\bP_k$ is well-defined. 

Under the null hypothesis $H_0$, we can re-express $\bW_{k1}$ and $\bW_{k2}$ as:
\[ \bW_{k1} = \frac{1}{k} \bX \bP_k \bX^T = \frac{1}{k} \Sigma_0^{1/2} \bZ_x \bP_k \bZ_x^T \Sigma_0^{1/2},\]  
\[\bW_{k2} = \frac{1}{n-1-k} \bX (I_n - \bP_k)\bX^T =  \frac{1}{n-1-k}\Sigma_0^{1/2} \bZ_x (P_\perp -\bP_k) \bZ_x^T\Sigma_0^{1/2}.\]
For any $\bY$, there exists an orthogonal matrix $\bU_k = \bU_k(\bY)$ of dimension $n\times k$ such that  
\[ \bP_{k} = \bU_{k} \bU_k^T.\]
Note that $\bU_k$ depends on $\bY$ only. 

Denote
\[
\bG_{k\lambda}
=
\Sigma_0^{-1/2}
(\bW_{k2}+\lambda I_{p_1})
\Sigma_0^{-1/2}
=\frac{1}{n-1-k}
\bZ_x
(P_\perp-\bU_k\bU_k^T)
\bZ_x^T
+
\lambda\Sigma_0^{-1}
.
\]
Define 
\[
\tilde{\bF}_{k\lambda}
=
\frac{1}{k}
\bU_k^T
\bZ_x^T
\Sigma_0^{1/2}
(\bW_{k2}+\lambda I_{p_1})^{-1}
\Sigma_0^{1/2}
\bZ_x
\bU_k
=
\frac{1}{k}
\bU_k^T
\bZ_x^T
\bG_{k\lambda}^{-1}
\bZ_x
\bU_k.
\]
The matrix $\tilde{\bF}_{k\lambda}$ may be viewed as the symmetric companion matrix of $\bF_{k\lambda}$ under $H_0$ in the sense that $\tilde{\bF}_{k\lambda}$ and $\bF_{k\lambda}$ share the same nonzero eigenvalues. Consequently, since our interest lies in $\sum_{j=1}^k \ell_j(\bF_{k\lambda})$ and $\ell_{\max}(\bF_{k\lambda})$, it is equivalent to analyze $\tilde{\bF}_{k\lambda}$ instead.

\subsection{Basic definitions}
Throughout, we use $a$, $a'$, $c$, $c'$, $\calK$, and $\calK'$ to denote generic constants whose values may change from line to line. Unless stated otherwise, these constants do not depend on $k$, $p_1$, $p_2$, or $n$. 
We use $E = E(z)$ and $\I = \I(z)$ to denote the real and imaginary parts of a complex number $z$, respectively. The argument $z$ is frequently omitted from expressions when no confusion may arise.


The following notion of a high-probability bound has been used in a number of works on random matrix theory. It provides a simple way of systematizing and making precise statements of the form ``$\xi$ is bounded with high probability by $\zeta$ up to small powers of $n$''. 
\begin{definition}[Stochastic domination]\label{def:stochastic_domination}~
\begin{itemize}
\item[(i)] Consider two families of nonnegative random variables
$$
\xi=\left(\xi^{(n)}(u): m \in \mathbb{N}, u \in U^{(n)}\right), \quad \zeta=\left(\zeta^{(n)}(u): n \in \mathbb{N}, u \in U^{(n)}\right),
$$
where $U^{(n)}$ is a possibly $n$-dependent parameter set. We say that $\xi$ is stochastically dominated by $\zeta$, uniformly in $u$, if for all (small) $\varepsilon>0$ and (large) $D>0$ we have
$$
\sup _{u \in U^{(n)}} \mathbb{P}\left[\xi^{(n)}(u)>n^{\varepsilon} \zeta^{(n)}(u)\right] \leqslant n^{-D}
$$
for large enough $n \geqslant n_0(\varepsilon, D)$. Throughout this paper, the stochastic domination will always be uniform in all parameters (such as matrix indices, deterministic vectors, and spectral parameters $z$) that are not explicitly fixed. 
If $\xi$ is stochastically dominated by $\zeta$, uniformly in $u$, we use the notation $\xi \prec \zeta$. Moreover, if for some complex family $\xi$ we have $|\xi| \prec \zeta$ we also write $\xi=O_{\prec}(\zeta)$.
\item[(ii)] We extend the definition of $O_{\prec}(\cdot)$ to matrices in the weak operator sense as follows. Let $A$ be a family of complex square random matrices and $\zeta$ a family of nonnegative random variables. Then we use $A=O_{\prec}(\zeta)$ to mean $|\langle\mathbf{v}, A \mathbf{w}\rangle| \prec \zeta\|\mathbf{v}\|_2\|\mathbf{w}\|_2$ uniformly for all deterministic vectors $\mathbf{v}$ and $\mathbf{w}$.
\item[(iii)] If there exists a positive constant $\calK$ such that $\xi \leq \calK \zeta$, then we write $\xi \lesssim \zeta$. If further there exists a positive constant $\calK'$ such that $\zeta \leq \calK' \xi$, then we write $\xi \asymp \zeta$. 
\end{itemize}
\end{definition}

\begin{definition}\label{def:high_probability}
We say that an event $\Lambda$ holds with high probability if for any large positive constant $D$, there exists $n_0(D)$ such that 
\[\mP(\Lambda^\complement)\leq n^{-D}, ~\mbox{ for any }n\geq n_0(D).\]
\end{definition}

\subsection{Proof of Theorem \ref{thm:main_fix_k}} \label{subsec:proof_main_thm_fix_k}

In this section, we establish Theorem~\ref{thm:main_fix_k} by adapting the technical results of \cite{li2020high} to the present setting. 
For completeness, we restate the main result of \cite{li2020high} in a form tailored to our framework.

\begin{theorem}[Adapted from Theorem~2.1 of \cite{li2020high}]
\label{thm:li2020}
Suppose that Conditions~\ref{enum:high_dimension_regime}, \ref{enum:moments_conditions}, and \ref{enum:boundedness_spectral_norm} hold. Let $k$ be fixed, and let $U_k\in\mathbb{R}^{n\times k}$ be a deterministic matrix with orthonormal columns, that is, $U_k^TU_k=I_k$.
Define
\[
M=
\frac{1}{n-k}
U_k^T\bZ_x^T
\Sigma_0^{1/2}
\left(
\frac{1}{n-k}
\Sigma_0^{1/2}
\bZ_x
(I_n-U_k U_k^T)
\bZ_x^T
\Sigma_0^{1/2}
+\lambda I_{p_1}
\right)^{-1}
\Sigma_0^{1/2}
\bZ_xU_k.
\]
Then,
\[\frac{\sqrt{p_1}}{\Omega_2(\lambda,q)} \Big( M-\Omega_1(\lambda,q)I_k \Big) \stackrel{D}{\longrightarrow} {\rm GOE}(k).\]
\end{theorem}

Theorem~\ref{thm:li2020} is established in \cite{li2020high} using martingale difference techniques. We do not reproduce the proof here; instead, we briefly highlight the connection between the framework considered in \cite{li2020high} and the present setting. 
Specifically, \cite{li2020high} considers a general spectral shrinkage framework in which the eigenvalues of
\[\frac{1}{n-k}\Sigma_0^{1/2}\bZ_x(I_n-U_kU_k^T)\bZ_x^T\Sigma_0^{1/2},\]
denoted by $\ell_1,\ldots,\ell_{p_1}$, are transformed through a general analytic shrinkage function $f$, while the corresponding eigenvectors are retained. Writing $Q$ for the associated eigenvector matrix, the resulting matrix takes the form
\[ \frac{1}{n-k} U_k^T\bZ_x^T \Sigma_0^{1/2} Q \operatorname{Diag} \bigl( f(\ell_1),\ldots,f(\ell_{p_1})\bigr)Q^T \Sigma_0^{1/2}\bZ_x U_k.\]
In the present work, ridge regularization in $M$ corresponds to the particular choice
$f(x)= {1}/{(x+\lambda)}$. Consequently, Theorem~\ref{thm:li2020} follows as a special case of Theorem~2.1 in \cite{li2020high}.

Second, for notational simplicity and clarity, the asymptotic centering and scaling parameters in \cite{li2020high} are presented in their limiting forms, which do not explicitly depend on the aspect ratio $q$. Specifically, \cite{li2020high} additionally assumes that the empirical spectral distribution $F^{\Sigma_0}$ converges in $\ell_\infty$-norm to a limiting distribution $F^\infty$ and that $q=q(n,k)$ converges to a constant. Under these assumptions, the quantities $\Omega_1(\lambda,q)$ and $\Omega_2(\lambda,q)$ admit deterministic limits. In the present work, we do not impose these additional assumptions. Instead, we retain the explicit dependence of the centering and scaling parameters on $(n,k)$ through $q=q(n,k)$ and the empirical spectral distribution $F^{\Sigma_0}$, extracted directly from the proof of Theorem~2.1 of \cite{li2020high}. 

Next, we prove Theorem~\ref{thm:main_fix_k} using Theorem~\ref{thm:li2020}. Fix $\bY$ and consider the event $\mathcal{E}_C$ holds. Let $U_{k+1}$ be chosen such that
\[
U_{k+1}U_{k+1}^T
=\bP_k + \frac{1}{n}{1}_n{1}^T_n.
\]
Without loss of generality, we may take
\[
U_{k+1}
=
[\bU_k, n^{-1/2} 1_n].
\]
With this choice, the leading principal $k\times k$ submatrix of $M$ is
$k/(n-1-k)\tilde{\bF}$.

Applying Theorem~\ref{thm:li2020} conditionally on $\bY$, we obtain
\[
\mathbb I(\mathcal E_C)
\frac{\sqrt{p_1}}{\Omega_2(\lambda,q)}
\Big(
M-\Omega_1(\lambda,q)I_{k+1}
\Big)
\,\Big|\, \bY
\stackrel{D}{\longrightarrow}
{\rm GOE}_{k+1}.
\]
Taking the leading principal $k\times k$ submatrix yields
\[
\mathbb I(\mathcal E_C)
\frac{\sqrt{p_1}}{\Omega_2(\lambda,q)}
\Big(
\frac{k}{n-1-k}\tilde{\bF}
-
\Omega_1(\lambda,q)I_k
\Big)
\,\Big|\, \bY
\stackrel{D}{\longrightarrow}
{\rm GOE}(k).
\]
Therefore, for any fixed open set $\mathcal C\subset\mathbb R^{k\times k}$,
\[
\mathbb P\!\left(
\mathbb I(\mathcal E_C)
\frac{\sqrt{p_1}}{\Omega_2(\lambda,q)}
\left(
\frac{k}{n-1-k}\tilde{\bF}
-
\Omega_1(\lambda,q)I_k
\right)
\in\mathcal C
\,\middle|\,
\bY
\right)
\longrightarrow
\mathbb P\big({\rm GOE}(k)\in\mathcal C\big).
\]
Integrating with respect to the marginal distribution of $\bY$ and applying the dominated convergence theorem gives
\[
\mathbb P\!\left(
\mathbb I(\mathcal E_C)
\frac{\sqrt{p_1}}{\Omega_2(\lambda,q)}
\left(
\frac{k}{n-1-k}\tilde{\bF}
-
\Omega_1(\lambda,q)I_k
\right)
\in\mathcal C
\right)
\longrightarrow
\mathbb P\big({\rm GOE}(k)\in\mathcal C\big).
\]
Since Condition~\ref{enum:lower_bound_eigen_Y} implies that $\mathbb P(\mathcal E_C)\ge 1-\exp(-n)$
for all sufficiently large $n$, the indicator $\mathbb I(\mathcal E_C)$ can be removed asymptotically. Hence,
\[
\mathbb P\!\left(
\frac{\sqrt{p_1}}{\Omega_2(\lambda,q)}
\left(
\frac{k}{n-1-k}\tilde{\bF}
-
\Omega_1(\lambda,q)I_k
\right)
\in\mathcal C
\right)
\longrightarrow
\mathbb P\big({\rm GOE}(k)\in\mathcal C\big).
\]
Finally, since $\bF$ and $\tilde{\bF}$ share the same nonzero eigenvalues, the conclusions of Theorem~\ref{thm:main_fix_k} follow.

\subsection{Proof of Theorem \ref{thm:main_diverge_k}}\label{subsec:proof_main_diverge_k}
In this section, we prove Theorem~\ref{thm:main_diverge_k}. The proof relies on the technical results developed in \cite{knowles2017anisotropic}, \cite{li2024analysis}, and \cite{li2025ridge}. The argument proceeds in three steps.

\begin{itemize}
    \item[]\textbf{Step 1.} We first summarize key properties of a generalized Mar\v{c}enko--Pastur equation, following \cite{li2024analysis}. This equation characterizes the asymptotic spectral behavior of $\bG_{k\lambda}$.

    \item[]\textbf{Step 2.} In addition to Conditions~\ref{enum:high_dimension_regime}--\ref{enum:edger_regularity}, we temporarily assume that the entries of $\bZ_x$ are Gaussian. Under this assumption, we establish the conclusion of Theorem~\ref{thm:main_diverge_k}.

    \item[]\textbf{Step 3.} We then extend the result to the general moment condition specified in Condition~\ref{enum:moments_conditions} by controlling the effect of relaxing the Gaussian assumption on the behavior of $\tilde{\bF}_{k\lambda}$.
\end{itemize}

Such a multi-step strategy is standard in the RMT literature, particularly in the analysis of random matrices with complex dependence structures. Under Gaussianity, the analysis is substantially simplified by the conditional independence structure underlying $\tilde{\bF}_{k\lambda}$. Indeed, since $\bU_k$ is an orthogonal matrix and $P_\perp-\bU_k\bU_k^T$ is orthogonal to $\bU_k\bU_k^T$, it follows that, conditional on $\bY$,
\[
\bZ_x\bU_k ~\mid ~\bY
\quad\text{is independent of}\quad
\bZ_x(P_\perp-\bU_k\bU_k^T) ~\mid ~\bY.
\]
Consequently, conditional on $\bY$, the companion matrix $\tilde{\bF}_{k\lambda}$ may be viewed as a sample covariance matrix with population covariance matrix $\bG_{k\lambda}$.

In the sequel, to simplify the notation, we suppress the dependence on $n$, $k$, and $\lambda$ of various quantities defined in the previous sections whenever no ambiguity arises.

\subsubsection{Step 1: Properties of a generalized Mar\v{c}enko-Pastur equation}\label{subsec:proof_properties_MP_G}
For the population ESD $F^{\Sigma_0}$ and the aspect ratio $q = q(n,k) = p_1/(n-1-k)$, define the following generalized Mar\v{c}enko--Pastur equation.  
\begin{equation}
\label{eq:MP_G}
\phi(z)
=
\int
\frac{\tau\, dF^{\Sigma_0}(\tau)}
{\tau\{[1+q\phi(z)]^{-1}-z\}+\lambda}, \qquad z\in\mathbb{C}^+.
\end{equation}
Under the assumption on $k$, $p_1$ and $n$ in Theorem \ref{thm:main_diverge_k}, we have that $q \asymp 1$. Moreover, $F^{\Sigma_0}$ does not degenerate to the Dirac measure at zero due to Condition \ref{enum:boundedness_spectral_norm}. Throughout the analysis, we assume that the two conditions hold.

Following from the classical works of \cite{marchenko1967distribution} and \cite{silverstein1995empirical}, the equation \eqref{eq:MP_G} admits a unique solution
\[
\phi(z)
=
\phi(z;q,F^{\Sigma_0})
\in
\mathbb{C}^+, \qquad  \text{for every } z \in \mathbb{C}^+.
\]
Moreover, $\phi(z)$ is the Stieltjes transform of a probability measure, denoted by
\[\calG =\calG(q,F^{\Sigma_0}).\]
However, except in a few special cases where $F^{\Sigma_0}$ has a particularly simple form, neither $\phi(z)$ nor the associated measure $\calG$ admits an explicit expression. Nevertheless, many properties of $\calG$ can be deduced from \eqref{eq:MP_G}, studied in \cite{li2024analysis}. 

Recall that $\eta={\lambda}/{\ell_{\max}(\Sigma_0)}$.  Assume that the spectrum of $\Sigma_0$ is asymptotically regular near $\ell_{\max}(\Sigma_0)$ in the sense of Definition~\ref{def:edge_regularity}. Namely, $|g'(h)|> \epsilon_1$ for $h\in (\eta -\epsilon_2, h)$. Then, $\calG$ is compactly supported, whose leftmost support edge is identified as $\rho$, defined in Section \ref{subsec:asymptotic_k_diverge}. The behavior of $\calG$ near $\rho$ is characterized as in Lemma \ref{lemma:leftmost_support_edge}. 
\begin{lemma}
\label{lemma:leftmost_support_edge}
The following properties of $\calG$ are summarized from \cite{li2024analysis}.
\begin{itemize}
\item[(a)] If $g'(h)< -\epsilon_1$, $h\in(\eta-\epsilon_2,\eta)$, then there exists a unique $h_0 <\eta-\epsilon_2$ such that $g'(h_0) = 0$. We have that the leftmost edge $\rho=\lim_{h\uparrow h_0} g(h) > \eta$. Moreover, $\calG$ is continuous near $\rho$ and admits a density, say $f_{\calG}$, in a neighborhood of $\rho$. In particular, there exists a neighborhood of $\rho$, say $(\rho, \rho + a)$, such that for all sufficiently large $n$
\[ f_{\calG}(x) \asymp \sqrt{x-\rho}, \quad x\in (\rho, \rho+a).\]
\item[(b)] If $g'(h) > \epsilon_1$, $h\in (\eta -\epsilon_2, \eta)$, we set $h_0 = \eta$. Then, the leftmost edge $\rho = \lim_{h\uparrow h_0}g(h) = \eta$. Moreover, $\calG$ has an atom at $\rho$ and we can find a lower bound $a>0$ such that $\calG(\{\rho\}) >a$, for all sufficiently large $n$. 
\item[(c)] Moreover, an upper bound on the rightmost support edge of $\calG$ can be found as $ (1+ \sqrt{q})^2 + \lambda/\ell_{\min}(\Sigma_0)$. 
\end{itemize}
\end{lemma}

The connection of $\phi(z)$ and $\calG$ to the analysis in Section \ref{sec:analysis_null} is as follows. Recall the definition of the function $s(g)$ defined for $g\in [0, \rho)$ as in \eqref{eq:def_s_nlambda_g}. We have the following relationship
\[ s(g) = \lim_{z\in\mathbb{C}^+ \to g} \phi(z)  = \lim_{z\in\mathbb{C}^+ \to g} \int \frac{d\calG(\tau)}{\tau -z} = \int \frac{d\calG(\tau)}{\tau -g},  \quad  g\in [0,\rho).\]
Moreover, take differentiation on both sides,
\[ s'(g) =\phi'(g) = \int \frac{d\calG(\tau)}{(\tau -g)^2} >0 \qquad   s''(g) =\phi''(g) = 2\int\frac{d\calG(\tau)}{(\tau -g)^3}>0.\]
Importantly,  since we either have that $\calG$ has an atom at $\rho$ with a lower bound on $\calG(\{\rho\})$ or $d\calG(x) \asymp \sqrt{x-\rho}$, we conclude that 
for any large constant $\calK'$, we can find a small constant $a'>0$ such that 
\[s'(g) \geq \calK', \qquad \text{if } g \in (\rho-a', \rho),\]
for all sufficiently large $n$. 

Consider a function that plays a foundational role in the following analysis as 
\[ A(x) =  -\frac{1}{x}  + \frac{p_1}{k} \int \frac{d\calG(\tau)}{\tau +x} = -\frac{1}{x} + \frac{p_1}{k} s(-x), \qquad x\in (-\rho, 0)\]
Consider the critical point $x = -\beta$ of the function as the root of the equation
\[ A'(-\beta) = \frac{1}{\beta^2} - (p_1/k)s'(\beta) = 0.\]
Equivalently, $\beta$ is such that 
\[ \beta^2 s'(\beta)  =  k/p_1.\]
Note that this definition $\beta$ is consistent with that in Theorem \ref{thm:main_diverge_k}. Such a root exists and is unique because $\beta^2s'(\beta)$ is a monotonic increasing function whose image set is $(0,\infty)$.  

Since under the assumption of Theorem \ref{thm:main_diverge_k}, $k\asymp p_1$, the quantity $\beta$ remains bounded away from $0$ and $\rho$, in the sense that $\beta \asymp 1$ and $\rho-\beta\asymp 1$.


\subsubsection{Step 2: Proof under Gaussianity}\label{subsec:proof_normality}

In this section, we establish \ref{thm:main_diverge_k} under the Gaussian assumption. The proof proceeds in several stages. First, we present the local laws of the resolvent of $\bG$ established in \cite{li2025ridge}. Second, we identify a suitable high-probability region to which $(\bG,\bY)$ belongs asymptotically. Third, conditioning on $(\bG,\bY)$ lying in this region, we establish weak convergence of the conditional distribution of
$\ell_{\max}(\tilde{\bF})\mid(\bG,\bY)$. 
Finally, we show that the corresponding conditional limit is asymptotically stable throughout this region, which allows us to deduce the unconditional convergence. 

For convenience, we call $\bU_{\perp} \in \mathbb{R}^{n\times (n-1-k)}$ to be an orthogonal matrix such that 
\[ \bU_{\perp} \bU_{\perp}^T = P_\perp - \bP_k.\]
Note that $\bU_\perp$ depends on $\bY$ only.

{\bf Stage 1.} Define a linearization matrix and its inverse as 
\[\bK(z) = \begin{bmatrix} \lambda \Sigma_0^{-1} - z I_{p_1}  & \bZ_x \bU_\perp \\ \bU_\perp^T \bZ_x^T & - I_{n-1-k}\end{bmatrix}, \]
\[\bJ(z) = \bK^{-1}(z) = \begin{bmatrix}
(\bG - z I_{p_1})^{-1} & (\bG - z I_{p_1})^{-1} \bZ_x \bU_{\perp} \\
\bU_\perp^T \bZ_x^T (\bG - z I_{p_1})^{-1} &  \bU_\perp^T \bZ_x^T (\bG - z I_{p_1})^{-1} \bZ_x \bU_\perp - I_{n-1-k}
\end{bmatrix}. 
\]
The reason to discuss $\bK(z)$ is that $\bJ(z)$ contains the resolvent $(\bG - zI_{p_1})^{-1} $ of $\bG$ as one of its block and $\bK(z)$ is linear in the entries of $\bZ_x$. 
The main arguments in the proof are the study of the behavior of $\bJ(z)$ when $z$ is near the real axis, known as ``local laws'' in RMT literature. We are particularly interested in characterizing the behavior of $\bJ(z)$ when $z$ is near $\rho$.

We introduce the following control parameter 
\[ \Psi(z) = \sqrt{\frac{\Im \phi(z)}{n\I(z)}} + \frac{1}{n\I(z)}.\]
Define
\[  Q= Q(a,a', n) = \Big\{z \in \mathbb{C}^+ ~:~ |E(z)| \leq 1/a,~~ E\leq \rho + a,~~n^{-1+a'} \leq \I(z) \leq 1/a'  \Big\}.\]
Define a  \emph{deterministic equivalent} matrix 
\[ \bOmega(z) = \begin{bmatrix}
    (\lambda \Sigma_0^{-1} - w(z) I_{p_1})^{-1} & 0 \\
    0 & (z- w(z)) I_{n-1-k}
\end{bmatrix},\]
where $w(z) = z - [1 + q \phi(z)]^{-1}$.

\begin{theorem}
    \label{thm:local_laws_G}
Suppose that Conditions \ref{enum:high_dimension_regime}--\ref{enum:edger_regularity} hold and $\bZ_x$ is normally distributed. Assume that $k$ is such that $k/n\to\gamma \in (0, \alpha)$ as in Theorem \ref{thm:main_diverge_k}. We have the following statements
\begin{itemize}
    \item[(a)]     Consider the case when $\rho > \eta $ first. For arbitrary $a'>0$, there exists a constant $a>0$ such that the following local laws hold.
    \begin{itemize}
        \item[(i)] The entrywise local law holds as 
        \[ \bJ(z) - \bOmega(z) = O_{\prec} (\Psi(z)), \quad z \in Q. \]
        \item[(ii)] The averaged local law holds as 
       \[ \hat{\phi}(z) - \phi(z) = O_\prec \big(\frac{1}{n\I(z)} \big), \quad z \in Q,\]
       where $\hat{\phi}(z) = p_1^{-1} \tr [(\bG -z I_{p_1})^{-1}]$.
    \end{itemize}
    \item[(b)] Consider the case when $\rho = \eta$. In this case, we have $\ell_{\min}(\bG)$ is $\lambda/\ell_{\max}(\Sigma_0)$ with probability 1. For any $z$ bounded away from $[\rho, (1+\sqrt{q})^2 +\lambda/\ell_{\min}(\Sigma_0)]$, we have 
    \[\hat{\phi}(z) -\phi(z) =O_\prec\big(\frac{1}{n} \big).\]
\end{itemize}
\end{theorem}

While Theorem~\ref{thm:local_laws_G} is established in Section~S.6 of \cite{li2025ridge} for deterministic $\bU_\perp$, the extension to the present setting, in which $\bU_\perp$ is random through its dependence on $\bY$, is straightforward. We briefly describe the main ideas. 

For any rank-$k$ projection matrix $\bP_k$,
\[
\frac{1}{n-1-k}
\bZ_x
(P_\perp-\bP_k)
\bZ_x^T
\stackrel{d}{=}
\frac{1}{n-1-k}
\widetilde{\bZ}_x
\widetilde{\bZ}_x^T,
\]
where $\widetilde{\bZ}_x$ is a $p_1\times(n-1-k)$ matrix with independent $N(0,1)$ entries. Consequently, the conditional distribution of $\bG$ given $\bY$ does not depend on $\bY$. Therefore, all bounds established in Section~S.6 of \cite{li2025ridge} hold uniformly over $\bY$. For example, Theorem~S.6.1 of \cite{li2025ridge} implies that, for any $\varepsilon>0$ and $D>0$, there exists $n_0=n_0(\varepsilon,D)$ such that, for all $n\ge n_0$, any deterministic vectors $\mathbf v$ and $\mathbf w$, and any realization of $\bY$ satisfying the event $\mathcal E_C$, we have
\[
\sup_{z\in Q}
\mathbb P\left(
\mathbb I(\mathcal E_C)
\bigl|
\langle
\mathbf v,
(\bJ(z)-\bOmega(z))
\mathbf w
\rangle
\bigr|
>
n^\varepsilon \Psi(z)
\;\middle|\;
\bY
\right)
\le n^{-D}.
\]

Integrating both sides with respect to the marginal distribution of $\bY$ yields
\[
\sup_{z\in Q}
\mathbb P\!\left(
\mathbb I(\mathcal E_C)
\bigl|
\langle
\mathbf v,
(\bJ(z)-\bOmega(z))
\mathbf w
\rangle
\bigr|
>
n^\varepsilon \Psi(z)
\right)
\le n^{-D}.
\]
Moreover, Condition~\ref{enum:lower_bound_eigen_Y} implies that $\mathbb P(\mathcal E_C)\ge1-\exp(-n)$, for all sufficiently large $n$. Consequently, the indicator $\mathbb I(\mathcal E_C)$ can be removed asymptotically without affecting the stochastic domination bounds. This completes the proof of Part~(a)(i). The remaining assertions follow by analogous arguments.

{\bf Stage 2.} As a consequence of Theorem~\ref{thm:local_laws_G}, we can characterize the asymptotic behavior of both the edge and bulk eigenvalues of $\bG$. For any $\varepsilon\in(0,2/3)$, define
\[
\mathfrak N(\varepsilon) = \left[
\rho-n^{-2/3+\varepsilon}
~~,~~
(1+\sqrt q)^2
+
\frac{\lambda}{\ell_{\min}(\Sigma_0)}
+
n^{-2/3+\varepsilon}
\right].
\]
The following lemma shows that, with high probability, the entire spectrum of $\bG$ is contained in $\mathfrak N(\varepsilon)$.

\begin{lemma}
\label{lemma:convergence_edge_G}
Under the same conditions as Theorem \ref{thm:local_laws_G}, for any $\varepsilon\in(0,2/3)$ and $D>0$, there exists a constant $n_0(D,\varepsilon)$ such that
\[
\mathbb P\!\Big(
[\ell_{\min}(\bG),\,\ell_{\max}(\bG)]
\subset
\mathfrak N(\varepsilon)\Big)
\ge
1-n^{-D},
\]
for all $n\ge n_0(D,\varepsilon)$. 
\end{lemma}
The lemma is obtained by adapting Theorem~S.6.2 of \cite{li2025ridge} and Theorem~3.12 of \cite{knowles2017anisotropic}. Again, we briefly describe the main ideas. 
First, applying Theorem~S.6.2 of \cite{li2025ridge} yields that there exists $n_1 (D, \varepsilon)$ such that for all $n\geq n_1(D,\varepsilon)$
\[
\mP \big(\ell_{\min}(\bG)
\ge
\rho-n^{-2/3+\varepsilon}, ~~\ell_k (\bS_y)\geq C \mid \bY  \big) \geq 1-n^{-D}/3.
\]
Next, Theorem~3.12 of \cite{knowles2017anisotropic} implies that there exists $n_2(D,\varepsilon)$ for all $n\geq n_2(D, \varepsilon)$
\[
\mP\Big(
\frac{1}{n-1-k}\ell_{\max}
\!\big(
\bZ_x
(P_\perp-\bP_k)
\bZ_x^T
\big)
\le
(1+\sqrt q)^2+n^{-2/3+\varepsilon},~~\ell_k (\bS_y)\geq C \mid \bY \Big) \geq 1- n^{-D}/3.
\]
The bounds are uniform over $\bY$. Taking expectations with respect to the marginal distribution of $\bY$ yields that for sufficiently large $n$
\[
\mathbb P\!\left(
[\ell_{\min}(\bG),\,\ell_{\max}(\bG)]
\subset
\mathfrak N(\varepsilon),
\;
\ell_k(\bS_y)>C
\right)
\ge
1-(2/3)n^{-D}.
\]
Consequently,
\[
\mathbb P\!\left(
[\ell_{\min}(\bG),\,\ell_{\max}(\bG)]
\subset
\mathfrak N(\varepsilon)
\right)
\ge
1-(2/3)n^{-D}-e^{-n},
\]
and hence
\[
\mathbb P\!\left(
[\ell_{\min}(\bG),\,\ell_{\max}(\bG)]
\subset
\mathfrak N(\varepsilon)
\right)
\ge
1-n^{-D}
\]
for all sufficiently large $n$.

While Lemma~\ref{lemma:convergence_edge_G} concerns the extreme eigenvalues of $\bG$, we next consider its bulk spectral behavior. The following result is adapted form Theorem~S.6.3 of \cite{li2025ridge}.

\begin{lemma}
\label{lemma:linear_spectral_G}
Assume the same conditions as in Theorem \ref{thm:local_laws_G}. Fix $\varepsilon\in(0,2/3)$, and let $f$ be any function analytic on an open interval $\mathfrak N(\infty)$ satisfying $\mathfrak N(\varepsilon)\subset\mathfrak N(\infty)$ for all sufficiently large $n$. 
Let
\[
\mathfrak M^{(1)}_\varepsilon
=
\Big\{
[\ell_{\min}(\bG),\ell_{\max}(\bG)]
\subset
\mathfrak N(\varepsilon)
\Big\}.
\]
Note that, by Lemma \ref{lemma:convergence_edge_G}, $\mathfrak{M}^{(1)}_\varepsilon$ holds with high probability.  Then, 
\[
\mathbb I(\mathfrak M^{(1)}_\varepsilon)
\left|
\frac{1}{p_1}
\sum_{j=1}^{p_1}
f\bigl(\ell_j(\bG)\bigr)
-
\int f(x)\,d\calG(x)
\right|
\prec
\frac{1}{n}.
\]
\end{lemma}
In particular, Theorem~S.6.3 of \cite{li2025ridge} shows that the result holds when $\bP_k$ is a deterministic sequence of projection matrices of rank $k$. Again, the bound in Lemma \ref{lemma:linear_spectral_G} is obtained by integrating over the marginal distribution of $\bY$ using analogous arguments in the proof of Lemma \ref{lemma:convergence_edge_G}.

Lemma~\ref{lemma:linear_spectral_G} shows that, on the event $\mathfrak M^{(1)}_\varepsilon$, the empirical spectral distribution of $\bG$ is well approximated by its deterministic limit $\calG$ when tested against analytic functions. Equivalently, linear spectral statistics of $\bG$ converge to their deterministic counterparts at rate $O_{\prec}(n^{-1})$.  


{\bf Stage 3.} Now, we are in a position to identify a suitable high-probability region to which $(\bG_{k\lambda}, \bY)$ belongs asymptotically.  Let $\varepsilon \in (0, 1/6)$ be fixed and consider the event $\mathfrak M^{(1)}_\varepsilon$ as in Lemma \ref{lemma:linear_spectral_G}. 


Recall the definition of $A(x)$ and $\beta$ in Section \ref{subsec:proof_properties_MP_G}. We replace $\calG(\tau)$ by $F^{\bG}(\tau)$ and set $\tilde{\beta}$ as the root in $(0, \ell_{\min}(\bG))$ of
\[\tilde{\beta}^2 \int \frac{dF^{\bG}(\tau)}{(\tau - \tilde{\beta})^2} = k/p_1.\]
Note that for any $\bG$, the root exists because that $F^{\bG}$ is discrete and as $\beta \to \ell_{\min}(\bG)$,  
\[\beta^2 \int\frac{dF^{\bG}(\tau)}{(\tau -\beta)^2} \to \infty.\] 
By Lemma \ref{lemma:linear_spectral_G}, we can show that
\[\tilde{\beta} - \beta \prec n^{-1+\varepsilon}.\]
Since $\beta$ is bounded away from $\rho$ as stated in Section \ref{subsec:proof_properties_MP_G}, we have $\tilde{\beta}$ is bounded away from $\ell_{\min}(\bG)$ with high probability. Moreover, define 
\[\widetilde{\Theta}_1 = \frac{1}{\tilde{\beta}} + \frac{1}{k} \sum_{j=1}^{p_1} \frac{1}{\ell_j(\bG) - \tilde{\beta}}, \qquad \widetilde{\Theta}_2 = \left[ (\frac{p_1}{k})^3 \int\frac{dF^{\bG}(\tau)}{(\tau - \tilde{\beta})^3}  + \frac{(p_1/k)^2}{\tilde{\beta}^3}  \right]^{1/3} \]
Then, $\widetilde{\Theta}_1$ and $\widetilde{\Theta}_2$ are $\bZ_x\bU_\perp$-dependent counterpart of $\Theta_1$ and $\Theta_2$. Since when $\tilde{\beta}$ is away from the support of $F^{\bG}$, which holds with high probability, $\widetilde{\Theta}_1$ and $\widetilde{\Theta}_2$ are linear spectral statistics of $\bG$, by Lemma \ref{lemma:linear_spectral_G}, 
\[ \mathbb{I}\big(\widetilde{\beta} < (\rho +\beta)/2\big) \mathbb{I}(\mathfrak{M}_\varepsilon^{(1)}) |\widetilde{\Theta}_1 - \Theta_1| \prec 1/n,\]
\[ \mathbb{I}\big(\widetilde{\beta} < (\rho +\beta)/2\big) \mathbb{I}(\mathfrak{M}_\varepsilon^{(1)}) |\widetilde{\Theta}_2 - \Theta_2| \prec 1/n.\]

Motivated by this, let $\mathfrak{M}^{(2)}_\varepsilon$ be the event that 
\[ \mathfrak{M}^{(2)}_\varepsilon = \Big\{ \tilde{\beta} \leq (\rho +\beta)/2, \quad |\widetilde{\Theta}_1 -\Theta_1| \leq n^{-1+\varepsilon}, \quad |\widetilde{\Theta}_2 - \Theta_2| \leq n^{-1+\varepsilon}  \Big\} \]
When $k$ diverges proportionally with $n$, we have that $\mathfrak{M}^{(1)}_\varepsilon\cap \mathfrak{M}^{(2)}_\varepsilon$ holds with high probability. 

{\bf Stage 4.} To show Theorem \ref{thm:main_diverge_k} under Gaussianity, we use the main theorem of \cite{lee2016tracy}. In particular, assume that $(\bG, \bY)$ belong to $\mathfrak{M}^{(1)}_\varepsilon\cap \mathfrak{M}^{(2)}_\varepsilon$. Conditioning on $(\bG,\bY)$, the main theorem of \cite{lee2016tracy} indicates that, for any $t\in\mathbb{R}$,
as $n\to\infty$,
\[\mP\Big(\Big\{\frac{p_1^{2/3}}{\widetilde{\Theta}_2}  \big(\ell_{\max}(\tilde{\bF}) - \widetilde{\Theta}_1 \big) \leq t\Big\} \cap \mathfrak{M}^{(1)}_\varepsilon\cap \mathfrak{M}^{(2)}_\varepsilon \mid \bG, \bY\Big) \longrightarrow \TW_1(t).\]
Again, we integrate both sides with respect to the distribution of  $(\bG, \bY)$. Using the dominated convergence theorem, we obtain that for any $t\in\mathbb{R}$,
\[ \mP\Big(\Big\{\frac{p_1^{2/3}}{\widetilde{\Theta}_2}  \big(\ell_{\max}(\tilde{\bF}) - \widetilde{\Theta}_1 \big) \leq t\Big\} \cap \mathfrak{M}^{(1)}_\varepsilon\cap \mathfrak{M}^{(2)}_\varepsilon \Big) \longrightarrow \TW_1(t)  \]
On the other hand, under the conditions of Theorem \ref{thm:main_diverge_k}, $\mathfrak{M}^{(1)}_\varepsilon\cap \mathfrak{M}^{(2)}_\varepsilon$ holds with high probability and on $\mathfrak{M}^{(1)}_\varepsilon\cap \mathfrak{M}^{(2)}_\varepsilon$, 
\[ p_1^{2/3} \Big|\widetilde{\Theta}_1 -\Theta_1 \Big| \leq p_1^{2/3} n^{-1+\varepsilon} \to 0,\qquad \Big|\widetilde{\Theta}_2 - \Theta_2 \Big| \leq n^{-1+\varepsilon} \to 0.\]
Applying Slutsky's theorem, the statements of Theorem \ref{thm:main_diverge_k} hold.

\subsubsection{Step 3: Extension to non-Gaussianity}

While Theorem \ref{thm:main_diverge_k} is proved under Gaussianity in Section \ref{subsec:proof_normality}, in this section, we show that the results are not sensitive to the normality assumption but generally hold when the moment conditions in \ref{enum:moments_conditions} hold. This phenomenon is known as edge universality in the RMT literature, meaning that the asymptotic distribution of edge eigenvalues is typically unaffected by the specific distribution of the matrix entries, provided that suitable moment conditions are satisfied. 
Recall that $E = E(z)$ and $\I = \I(z)$ are the real and imaginary parts of $z$. We frequently omit $z$ in the expressions. 

Edge universality is typically established via a Green function comparison theorem. Define the resolvent (Green function) of the matrix $\tilde{\bF}$ by
\[
\bL(z) =\big(\frac{1}{k}\bU_k^{T} \bZ_x^{T} \bG^{-1} \bZ_x \bU_k - z I_{k}\big)^{-1}, 
\qquad z \in \mathbb{C}^+ .
\]
We consider two settings: one in which $\bZ_x$ has a general distribution satisfying Condition~\ref{enum:moments_conditions}, and one in which $\bZ_x$ is Gaussian. To distinguish the latter from the former, we denote the Gaussian matrix by $\bZ_x^0$. Accordingly, the resolvent corresponding to $\bZ_x$ is denoted by $\bL_{\bZ_x}$, while the resolvent corresponding to $\bZ_x^0$ is denoted by $\bL_{\bZ_x^0}$.
The averaged trace of $\bL_{\bZ_x}(z)$ and $\bL_{\bZ_x^0}(z)$ is taken as 
\[ \overline{L}_{\bZ_x}(z) = \frac{1}{k} \tr (\bL_{\bZ_x}(z)) \quad \text{and} \quad \overline{L}_{\bZ_x^0}(z) = \frac{1}{k} \tr (\bL_{\bZ_x^0} (z)).\]
Moreover, the transformed F matrix corresponding to $\bZ_x^0$ is denoted by $\tilde{\bF}^0$. 

We claim that to extend the results to non-Gaussianity, it suffices to show the following Green function comparison theorem (Theorem \ref{thm:green_function_comparison}) and rigidity of the largest root (Lemma \ref{lemma:rough_order_largest_root}). 
\begin{theorem}[Green function comparison]\label{thm:green_function_comparison}
Let $\epsilon>0$ and set $\I = n^{-2/3-\epsilon}$. Let $E_1, E_2 \in \mathbb{R}$ satisfy $E_1 < E_2$ and
\[
|E_1 - \Theta_1|,\; |E_2 - \Theta_1| \leq n^{-2/3+\epsilon}.
\]
Let $K:\mathbb{R}\rightarrow \mathbb{R}$ be a smooth function such that
\[
\max_x |K^{(l)}(x)| \leq C, \qquad l=1,2,3,4,5,
\]
for some constant $C>0$. Then there exists a constant $c>0$ such that, for all sufficiently large $n$ and sufficiently small $\epsilon$, 
\[
\left|
\mathbb{E} \, K\!\left(n \int_{E_1}^{E_2} \Im \overline{L}_{\bZ_x}(x+i \I )\, dx\right)
-
\mathbb{E} \, K\!\left(n \int_{E_1}^{E_2} \Im \overline{L}_{\bZ_x^0}(x+i \I)\, dx\right)
\right| 
\leq n^{-c}.
\]
\end{theorem}

\begin{lemma}[Rigidity of largest eigenvalue]\label{lemma:rough_order_largest_root}
    Under the conditions of Theorem \ref{thm:main_diverge_k}, we have
    \[ \ell_{\max}(\tilde{\bF}) - \Theta_1  = O_\prec(n^{-2/3}). \]
\end{lemma}

Before proceeding to the proof of Theorem \ref{thm:green_function_comparison} and Lemma \ref{lemma:rough_order_largest_root}, we briefly discuss how the results are connected to the asymptotic Tracy-Widom limit of the largest root under non-Gaussianity.

First, given Lemma \ref{lemma:rough_order_largest_root}, we can locate to the domain of $E$ such that $|E-\Theta_1|\prec n^{-2/3}$. Fix $E^* \prec n^{-2/3}$ such that it suffices to consider $\ell_{\max}(\tilde{\bF}) \leq \Theta_1 + E^*$. Choose $|E - \Theta_1| \prec n^{-2/3}$, $\I  = n^{-2/3- 9\epsilon}$ and $l = \frac{1}{2} n^{-2/3-\epsilon}$. Then, for some sufficiently small constant $\epsilon>0$ and sufficiently large constant $M$, there exists a function $K$ satisfying the condition in Theorem \ref{thm:green_function_comparison}, a constant $n_0 = n_0(\epsilon, M)$, such that 
\begin{align}
& \mathbb{E} K\left(\frac{k}{\pi} \int_{E-l}^{\Theta_1+E^*} \Im \overline{L}_{\bZ_x}(x+i \I ) d x\right) \leq \mathbb{P}\left( \ell_{\max}(\tilde{\bF}) \leq E\right)\nonumber\\
&\phantom{ssssss} \leq \mathbb{E} K\left(\frac{k}{\pi} \int_{E+l}^{\Theta_1+E^*} \Im \overline{L}_{\bZ_x}(x+i \I) d x\right)+k^{-M}, \label{eq:bound_probability_by_ST_1}\\
&\mathbb{E} K\left(\frac{k}{\pi} \int_{E-l}^{\Theta_1+E^*} \Im \overline{L}_{\bZ_x^0}(x+i \I) d x\right) \leq \mathbb{P}\left( \ell_{\max}(\tilde{\bF}^0) \leq E\right) \nonumber\\
&\phantom{ssssss} \leq \mathbb{E} K\left(\frac{k}{\pi} \int_{E+l}^{\Theta_1+E^*} \Im \overline{L}_{\bZ_x^0}(x+i \I) d x\right)+k^{-M}, \label{eq:bound_probability_by_ST_2}
\end{align}
whenever $n>n_0$. We omit the details in the function $K$ and the proof of \eqref{eq:bound_probability_by_ST_1} and \eqref{eq:bound_probability_by_ST_2}  because it is a standard procedure in RMT. Similar works include Corollary 5.1 of Lemma 6.1 of \citet{erdHos2012rigidity},  \cite{BaoPanZhou2013LocalEdge}, and Lemma 4.1 of \cite{PillaiYin2014}. 

Following \eqref{eq:bound_probability_by_ST_1} and \eqref{eq:bound_probability_by_ST_2},  setting $E$ in \eqref{eq:bound_probability_by_ST_2} as $E\pm 2l$,  we immediately get 
\[ \mP(\ell_{\max}(\tilde{\bF}) \leq E) \geq \mP(\ell_{\max}(\tilde{\bF}^0) \leq E -2l) - k^{-M},\]
\[ \mP(\ell_{\max}(\tilde{\bF}) \leq E) \leq \mP(\ell_{\max}(\tilde{\bF}^0) \leq E + 2l) + k^{-M}.\]
Since $l = O(n^{-2/3 -\epsilon})$, we have then
\[  \mP(\ell_{\max}(\tilde{\bF}^0) \leq E \pm 2l) - \mP(\ell_{\max}(\tilde{\bF}^0) \leq E ) \to 0.\]
We then conclude that 
\[ \mP(\ell_{\max}(\tilde{\bF}) \leq E)  - \mP (\ell_{\max}(\tilde{\bF}^0) \leq  E) \to 0.\]
Since we already established the asymptotic Tracy-Widom distribution of $\ell_{\max}(\tilde{\bF}^0)$ in Section \ref{subsec:proof_normality}, $\ell_{\max}(\tilde{\bF})$ has the same weak limit. This will complete the proof of Theorem \ref{thm:main_diverge_k}.


It remains to prove Lemma~\ref{lemma:rough_order_largest_root} and Theorem~\ref{thm:green_function_comparison}. To this end, we analyze the behavior of the resolvent $\bL(z)$ in a neighborhood of $\Theta_1$. The objective is to show that the resolvent $\bL_{\bZ_x}(z)$ has the same asymptotic limit as that of $\bL_{\bZ_x^0}(z)$. However, it is difficult to work directly with $\bL(z)$ due to its complicated dependence on the individual entries of $\bZ_x$.

Following the strategy of \citet{han2016tracy} and \citet{han2018unified}, we instead express the resolvent $\bL(z)$ as a submatrix of the inverse of a linearization matrix.
In particular, define the following $3\times 3$ block matrix 
\begin{equation}\label{eq:bH}
\bH(z) = \bH(z,\bZ_x,\bY) =\left( \begin{matrix} 
- z I_{k} &  \bU_k^T \bZ_x^T & 0 \\[10pt]
 \bZ_x \bU_k & -\lambda \Sigma_0^{-1} & \bZ_x \bU_{\perp} \\[10pt]
0 & \bU_{\perp}^T \bZ_x^T & I_{n-1-k} 
\end{matrix} \right),
\end{equation}
where $\bU_{\perp} = \bU_{\perp}(\bY) \in \mathbb{R}^{n\times (n-1-k)}$ is an orthogonal matrix such that $\bU_\perp \bU_\perp^T = P_\perp - \bU_k \bU_k^T$. 

Note that $\bH(z)$ is linear in $\bZ_x$. Using the Schur complement formula, it can be verified easily that the upper-left block of the matrix $\bH^{-1}(z)$ is the resolvent $\bL(z)$. Therefore, it suffices to study the properties of $\bH^{-1}(z)$. 

It is important to note that when $\lambda = 0$, the matrix $\bH(z)$ reduces to the linearization matrix (also denoted by $\bH$) used in \citet{han2016tracy} and \citet{han2018unified}. In other words, our linearization matrix differs from that in \citet{han2016tracy} and \citet{han2018unified} only in a single deterministic block (the middle block). Owing to this structural similarity, the arguments developed in \citet{han2018unified} can be reused with only minor modifications, which substantially simplifies our proof.

The explicit expression of $\bH^{-1}(z)$ can be obtained using the standard block matrix inversion formula
\[
\begin{pmatrix}
\mathbf{K} & \mathbf{B} \\
\mathbf{C} & \mathbf{D}
\end{pmatrix}^{-1}
=
\begin{pmatrix}
0 & 0 \\
0 & \mathbf{D}^{-1}
\end{pmatrix}
+
\begin{pmatrix}
\mathbf{I} \\
-\mathbf{D}^{-1}\mathbf{C}
\end{pmatrix}
\left(\mathbf{K}-\mathbf{B}\mathbf{D}^{-1}\mathbf{C}\right)^{-1}
\begin{pmatrix}
\mathbf{I} ~&~\mathbf{B}\mathbf{D}^{-1}
\end{pmatrix}.
\]
As an illustration, the three diagonal blocks of $\bH^{-1}$ are given by
\begin{align*}
(\bH^{-1})_{11} &= \bL(z) = (\bU_k^{T}\bZ_x^{T}\bG^{-1}\bZ_x \bU_k - z I_{k})^{-1},\\[4pt]
(\bH^{-1})_{22} &= (-\bG + z^{-1}\bZ_x \bU_k \bU_k^{T}\bZ_x^{T})^{-1},\\[4pt]
(\bH^{-1})_{33} &= I_{n-1-k} + \bU_{\perp}^{T}\bZ_x^{T}(\bH^{-1})_{22}\bZ_x \bU_{\perp}.
\end{align*}
We omit the explicit expressions for the off-diagonal blocks.

Importantly, when $\bZ_x \bU_{\perp}$ is regarded as fixed, all blocks of $\bH^{-1}$ can be expressed as linear combinations of the block matrices appearing in equation~(4.3) of \citet{knowles2017anisotropic}, provided that the matrix
\[
\bG^{-1} = (\bZ_x \bU_{\perp} \bU_{\perp}^{T}\bZ_x^{T} + \lambda \Sigma_0^{-1})^{-1}
\]
is taken as the population covariance matrix ``$\Sigma$'' in the formulation of \citet{knowles2017anisotropic}. For example, among the diagonal blocks, $(\bH^{-1})_{11}$ corresponds to the last block in (4.3) of \citet{knowles2017anisotropic} (denoted there by $G_n$), while $(\bH^{-1})_{22}$ corresponds to the first block in (4.3) (denoted there by $G_M$). Finally, $(\bH^{-1})_{33}$ can be written as $I_{n-1-k} + \bU_{\perp}^{T}\bZ_x^{T}G_M \bZ_x \bU_{\perp}$.

Heuristically, based on equation~(3.5) of  \cite{knowles2017anisotropic}, the following behaviors of the blocks of $\bH^{-1}$ are expected. 
Recall the definition of the Mar\v{c}enko-Pastur equation solution in \eqref{eq:M_P_equation}, we evaluate the equation at $F^{\bG}$ and $p_1/k$.

Let $\hat{\varphi}_{\bG}(z) = \varphi(z, p_1/k, F^{\bG})$, that is, $\hat{\varphi}_{\bG}(z)$ be the solution to 
\[ z = \frac{-1}{\hat{\varphi}_{\bG}(z)} + \frac{p_1}{k} \int \frac{\tau dF^{\bG}(\tau)}{\tau \hat{\varphi}_{\bG}(z) +1}.  \]
Moreover, let ${\varphi}_{\calG}(z) = \varphi(z, p_1/k, \calG)$ be the limiting counterpart of $\hat{\varphi}_{\bG}(z)$, that is, ${\varphi}_{\calG}(z)$ be the solution to
\[ z = \frac{-1}{\varphi_{\calG}(z)} + \frac{p_1}{k} \int \frac{\tau d\calG(\tau)}{\tau \varphi_{\calG} (z) +1}.  \]

Then, 
\begin{align*}
(\bH^{-1})_{11} &\approx \hat{\varphi}_{\bG}(z)\, I_{k},\\[4pt]
(\bH^{-1})_{22} &\approx -\, \bG^{-1} \bigl(1+ \hat{\varphi}_{\bG}(z)\bG^{-1}\bigr)^{-1}
= -\, (\bG + \hat{\varphi}_{\bG}(z) I_{p_1})^{-1},\\[4pt]
(\bH^{-1})_{33} &\approx I_{n-1-k} - \bU_{\perp}^{T} \bZ_x^{T} (\bG + \hat{\varphi}_{\bG}(z) I_{p_1})^{-1} \bZ_x \bU_{\perp},\\[4pt]
(\bH^{-1})_{ij} & \approx 0 ,\quad \text{if } i\neq j. 
\end{align*}



Moreover, $(\bG + \hat{\varphi}_{\bG}(z) I_{p_1})^{-1}$ is precisely the resolvent of $\bG$ evaluated at $-\hat{\varphi}_{\bG}(z)$. 
In addition, the matrix
\[
I_{n-1-k} - \bU_{\perp}^{T} \bZ_x^{T} (\bG + \hat{\varphi}_{\bG}(z) I_p)^{-1} \bZ_x \bU_{\perp}
\]
can be written as $-\bJ_{(22)}(-\hat{\varphi}_{\bG}(z))$, where $\bJ_{(22)}$ denotes the lower-right block of the matrix $\bJ$. Combining the deterministic equivalent $\bOmega(z)$ of $\bJ(z)$  with the concentration of $\hat{\varphi}_{\bG}(z)$ to ${\varphi}_{\calG}(z)$, we obtain
\begin{align*}
(\bH^{-1})_{22} 
&\approx -\, (\bG + \hat{\varphi}_{\bG}(z) I_{p_1})^{-1} 
\approx \big( \lambda \Sigma_0^{-1} - w(-{\varphi}_{\calG}(z)) I_{p_1} \big)^{-1},\\[6pt]
(\bH^{-1})_{33} 
&\approx I_{n-1-k} - \bU_{\perp}^{T} \bZ_x^{T} (\bG + \hat{\varphi}_{\bG}(z) I_{p_1})^{-1} \bZ_x \bU_{\perp} 
\approx \big(w(-{\varphi}_{\calG}(z)) - z \big) I_{n-1-k}.
\end{align*}

We define the following two block diagonal matrices
\begin{align*}
    \widehat{\Pi}(z) = \left(
    \begin{matrix}
        \hat{\varphi}_{\bG}(z) I_{k} & 0 & 0\\
        0 & (\bG + \hat{\varphi}_{\bG}(z) I_{p_1})^{-1} & 0\\
        0 & 0 & I_{n-1-k} - \bU_{\perp}^{T} \bZ_x^{T} (\bG + \hat{\varphi}_{\bG}(z) I_{p_1})^{-1} \bZ_x \bU_{\perp}
    \end{matrix}
    \right).
\end{align*}

\begin{align*}
    \Pi(z)  =\begin{pmatrix}
        {\varphi}_{\calG}(z) I_{k} & 0 & 0 \\
        0 & \big( \lambda \Sigma_0^{-1} - w(-{\varphi}_{\calG}(z)) I_{p_1} \big)^{-1} & 0\\
        0 & & \big(w(-{\varphi}_{\calG}(z)) - z \big) I_{n-1-k}
    \end{pmatrix}.
\end{align*}

It is worth noting that when $\lambda=0$, the matrix $\widehat{\Pi}(z)$ coincides with the ``limit'' matrix $\Pi(z)$ in \citet{han2018unified} (see equation~(8.5) therein). In particular, for $\lambda=0$, \citet{han2016tracy} and \citet{han2018unified} established that
\[
\bH^{-1}(z) - \widehat{\Pi}(z) 
\prec 
\sqrt{\frac{\Im {\varphi}_{\calG}(z)}{k\I(z)}} + \frac{1}{k\I(z)},
\]
whenever $z$ lies in a neighborhood of $\Theta_1$.  In this section, we establish the corresponding concentration result for $\lambda>0$, where $\widehat{\Pi}(z)$ is replaced by the deterministic equivalent matrix $\Pi(z)$.





Define the following domain near $\Theta_1$ 
\[
D = D(a,a',k) = \Big\{ z = E+\mathrm{i}y \in \mathbb{C}^+ : |E-\Theta_1| \leq a,\; k^{-1+a'} \leq y \leq 1/a' \Big\}
\]
We define the control parameter 
\[ \Phi(z)  =  \Phi(z,\lambda) =  \sqrt{\frac{\Im {\varphi}_{\calG}(z) }{n\I(z)} } + \frac{1}{n\I(z)}.\]

By the convergence of linear spectral statistics established in Lemma~\ref{lemma:linear_spectral_G}, we obtain
\[\bigl|\hat {\varphi}_{\bG}(z)-{\varphi}_{\calG}(z)\bigr|\prec\frac{1}{n\I(z)},\qquad
z\in D(a,a',k).
\]
A similar argument is used in the proof of Theorem~S.7.1 of \cite{li2025ridge}.



Again, we proceed by first assuming Gaussianity. The following theorem follows from a combination of Theorem 3.14 in \cite{knowles2017anisotropic}. 
\begin{theorem}
    \label{thm:local_laws_bH}
    Suppose that the conditions of Theorem \ref{thm:main_diverge_k} hold. Additionally, assume that the observations are Gaussian.  Fix any $\lambda>0$. 
For any $a' >0$,  there exist constants $a>0$ such that the following results hold. 
\begin{itemize}
    \item[(i)] (Entrywise local law)  We have
\[ \bH^{-1}(z) - \widehat{\Pi}(z) = O_\prec(\Phi(z)),\]
uniformly in ${D}(a, a', k)$. 
\item[(ii)] (Averaged local law) We have 
\[| \overline{L}(z) - \hat{\varphi}_{\bG}(z)| \prec \frac{1}{k\I(z)},\]
uniformly in $D(a,a',k)$. Here $\overline{L}(z)$ is the averaged trace of $\bL(z) = (\bH^{-1}(z))_{11}$. 

\item[(iii)] (Local law for $\widehat{\Pi}(z)$) Moreover,
\[  \widehat{\Pi}(z) - \Pi(z) = O_\prec(\frac{1}{n\I(z)}),\]
uniformly in ${D}(a, a', k)$. It immediately follows that part (i) still holds if $\widehat{\Pi}(z)$ is replaced by $\Pi(z)$.
\end{itemize}
\end{theorem}

Before proceeding to the proof of Theorem~\ref{thm:local_laws_bH}, as a connection to the literature, \cite{han2016tracy} and \cite{han2018unified} established part (i) and (ii) of Theorem \ref{thm:local_laws_bH} when $\lambda =0$ under both Gaussian and Condition \ref{enum:moments_conditions}.  

\begin{proof}[Proof of Theorem~\ref{thm:local_laws_bH}]
Recall that with high probability, $\tilde{\beta}$ is bounded away from $\rho$ in the sense that there exists a constant $c>0$ such that $|\tilde{\beta} -\rho| \geq c$. Consequently, when $\bG^{-1}$ is regarded as the ``population covariance matrix'', the regularity condition in Definition~2.7 of \citet{knowles2017anisotropic} is satisfied with high probability.

Under this regularity condition, each block of $\bH^{-1}$ can be expressed as a linear combination of the blocks of the linearization matrix $G$ studied in \citet{knowles2017anisotropic} (see equation~(4.3) therein). Applying the entrywise local laws from Theorem~3.14 of \citet{knowles2017anisotropic}, we obtain that for any $a'>0$ there exists $2a>0$ such that
\[
\mathbb{I}(|\tilde{\beta}-\rho|>c)\,
\bigl|\bH^{-1}(z)-\widehat{\Pi}(z)\bigr|
\; \Big| \;
\bZ_x \bU_{\perp}
\prec 
\sqrt{\frac{\Im \hat{\varphi}_{\bG}(z)}{n\I(z) }}
+\frac{1}{n\I(z) },
\]
uniformly for $z\in \tilde{D}(2a,a',k)$. Here, $\tilde{D}(2a,a',k)$ denotes the $\bZ_x \bU_{\perp}$–dependent counterpart of $D(2a,a',k)$ obtained by replacing $\Theta_1$ with $\tilde{\Theta}_1$. 

Due to the convergence of $\tilde{\Theta}_1$ to $\Theta_1$, for any $a>0$, we have $D(a,a',k)\subset \tilde{D}(2a,a',k)$ with high probability. Together with the fact that $|\hat{\varphi}_{\bG}(z)-{\varphi}_{\calG}(z)|\prec 1/(n\I(z))$. Combining these results, 
\[
\big|\bH^{-1}(z)-\widehat{\Pi}(z)\big|
\,\Big|\, \bZ_x \bU_{\perp}
\prec 
\sqrt{\frac{\Im {\varphi}_{\calG}(z)}{n \I(z) }}
+\frac{1}{n\I(z) },
\]
uniformly for $z\in D(a,a',k)$. Finally, integrating with respect to the distribution of $\bZ_x \bU_{\perp}$, we conclude that
\[
\big|\bH^{-1}(z)-\widehat{\Pi}(z)\big|
\prec 
\sqrt{\frac{\Im {\varphi}_{\calG}(z)}{n\I(z)}}
+\frac{1}{n\I(z)},
\]
uniformly in $D(a,a',k)$. 

The averaged local law in part (ii) follows directly from the averaged version of Theorem~3.14 of \citet{knowles2017anisotropic} together with analogous arguments.

To show part (iii), we only need to consider the second and third diagonal blocks, since the first block is such that $|\hat{\varphi}_{\bG}(z) - {\varphi}_{\calG}(z)|\prec 1/(n\I(z) )$. 
These two blocks in $\widehat{\Pi}(z)$ are linear combinations of the blocks in $\bJ(-\hat{\varphi}_{\bG}(z))$.  Accordingly, the two diagonal blocks in $\Pi(z)$ are the same linear combinations of the blocks in $\bOmega(-{\varphi}_{\calG}(z))$. Given Theorem \ref{thm:local_laws_bH_nonGauss}, to show the concentration of $\widehat{\Pi}(z)$ around $\Pi(z)$, it suffices to verify that $-\hat{\varphi}_{\bG}(z)$ is away from the support of $\calG$ with high probability. By Lemma 4.1 of \cite{li2024analysis}, there exits a constant $c'>0$ such that 
\[\inf_{\tau \in \operatorname{supp}(\calG)} |\tau + {\varphi}_{\calG}(z)|\geq c'\] 
for all $z\in D(a,a',k)$. Therefore, $\inf_{\tau \in \operatorname{supp}(\calG)} | \tau + \hat{\varphi}_{\bG}(z) |\geq c'$ with high probability.  It completes the proof.
\end{proof}

The last step in the proof of Theorem \ref{thm:main_diverge_k} is to extend the results in Theorem \ref{thm:local_laws_bH} to non-Gaussianity. We aim to establish the following results.
\begin{theorem}
    \label{thm:local_laws_bH_nonGauss}
    Suppose that the conditions of Theorem \ref{thm:main_diverge_k} hold. Fix any $\lambda>0$. For any $a' >0$,  there exist constants $a>0$ such that the following results hold. 
\begin{itemize}
    \item[(i)] (Entrywise local law)  We have
\[ \bH^{-1}(z) - \Pi(z) = O_\prec(\Phi(z)),\]
uniformly in ${D}(a, a', k)$. 
\item[(ii)] (Averaged local law) We have 
\[| \overline{L}(z) - \varphi_{\calG}(z)| \prec \frac{1}{n\I(z) },\]
uniformly in $D(a,a',k)$. Here $\overline{L}(z)$ is the averaged trace of $\bL(z) = (\bH^{-1}(z))_{11}$. 
\end{itemize}
\end{theorem}

Once Theorem~\ref{thm:local_laws_bH_nonGauss} is established, it is standard in RMT to deduce Lemma~\ref{lemma:rough_order_largest_root} (edge rigidity) and Theorem~\ref{thm:green_function_comparison} (the Green function comparison theorem). Consequently, edge universality for $\tilde{\bF}$ follows, completing the proof of Theorem~\ref{thm:main_diverge_k}. Such arguments are standard in the literature; see, for example, Lemma~4 of \citet{han2016tracy} and Lemma~1 of \citet{han2018unified}. We therefore omit the details.

Given Theorem~\ref{thm:local_laws_bH}, Theorem \ref{thm:local_laws_bH_nonGauss}, and Lemma~\ref{lemma:rough_order_largest_root}, it is likewise routine in RMT to derive the bound required for the Green function comparison theorem. Similar arguments can be found in Section~6 of \citet{erdHos2012rigidity} and Section~12 of \citet{han2018unified}. In particular, owing to the close similarity between our linearization matrix $\bH(z)$ and that in \citet{han2018unified}, the arguments in Section~12 of \citet{han2018unified} can be applied essentially verbatim. We therefore omit the details.

\begin{proof}[Proof of Theorem \ref{thm:local_laws_bH_nonGauss}]

The proof of Theorem~\ref{thm:local_laws_bH_nonGauss} is technically considerably more demanding. Nevertheless, closely related analyses have been carried out in Sections~9.1, 9.2, and~10 of \citet{han2016tracy} and \citet{han2018unified}, by adapting the arguments originally developed in \citet{knowles2017anisotropic}. 

Recall that our linearization matrix $\bH$ coincides with that of \citet{han2018unified} in the special case $\lambda=0$. Owing to this structural similarity, most of the arguments in \citet{han2018unified} can be directly adapted to our setting. We therefore focus only on describing the similarities and differences, and on identifying the modifications required to tailor those arguments to the present context. With these modifications, Theorem \ref{thm:local_laws_bH_nonGauss} is proved using arguments in Sections 9.1, 9.2 and 10 of \cite{han2018unified} in a verbatim manner.

For any deterministic unit vector $v$ and $w$, any observation matrix $\bZ_x$ satisfying \ref{enum:moments_conditions}, and $z\in D$, define the inner product function 
\[ F (\bZ_x) = F(\bZ_x, z, v, w)  = |v^T (\bH^{-1} (\bZ_x, z) - \Pi(\bZ_x, z))w|. \]
To show Theorem \ref{thm:local_laws_bH_nonGauss}, it suffices to show that 
\[ \mathbb{I}(\mathcal{E}_C) F(\bZ_x) \mid \bY \prec \Phi  \qquad z\in D,\]
and to show the bound is uniform in $\bY$. 

While this argument is established in Theorem \ref{thm:local_laws_bH} under Gaussian observations, say $\bZ_x^0$, we only need to control the difference 
\[ \mE (\mathbb{I}(\mathcal{E}_C)F(\bZ_x) \mid \bY) - \mE (\mathbb{I}(\mathcal{E}_C)F ( \bZ_x^0)\mid \bY).\]

To this end, \citet{knowles2017anisotropic} employ an interpolation method in which a family of distributions $\{f_\theta : \theta \in [0,1]\}$ is introduced such that $f_0$ is $N(0,1)$ and $f_1$ is the distribution of the entries of $\bZ_x$ (see Definition~7.6 of \citet{knowles2017anisotropic} for details). Let $\bZ_x^\theta$ denote the matrix whose entries follow the distribution $f_\theta$. 

To control
\[\mE \big(\mathbb{I}(\mathcal{E}_C)F(\bZ_x^1)\mid \bY\big)  - \mE \big(\mathbb{I}(\mathcal{E}_C) F(\bZ_x^0) \mid \bY\big),\]
one analyzes the change in $\bH^{-1}(\bZ_x^\theta)$ induced by replacing a single entry of $\bZ_x^{\theta}$. Specifically, let $\bZ_{(uv)}^{\theta,s}$ denote the matrix obtained from $\bZ_x^{\theta}$ by replacing its $(u,v)$-th entry with the value $s$. It turns out that we only need to control 
\[ \sum_{uv} \mE \big(\mathbb{I}(\mathcal{E}_C)F^k (\bZ_{(uv)}^{\theta, z_{uv}^{1}}) \mid \bY \big)- \mE\big( \mathbb{I}(\mathcal{E}_C) F^k (\bZ_{(uv)}^{\theta, z_{uv}^{0}})\mid \bY \big),\]
for arbitrary power $k\in2\mathbb{N}$, $\theta\in[0,1]$, and $(u,v)$. See Lemmas 7.9 and 7.10 of \cite{knowles2017anisotropic} for details. The analysis relies on explicitly quantifying the difference in $\bH^{-1}(\bZ_x^{\theta})$ induced by the change  of the $(u,v)$-th entry of $\bZ_x^{\theta}$. Since $A^{-1} - B^{-1} = - A^{-1} (A-B) B^{-1}$, we have 
\[ \bH^{-1}(\bZ_{(uv)}^{\theta, z_{uv}^{1}}) - \bH^{-1}(\bZ_{(uv)}^{\theta, z_{uv}^{0}}) =  - \bH^{-1}(\bZ_{(uv)}^{\theta, z_{uv}^{1}}) [\bH(\bZ_{(uv)}^{\theta, z_{uv}^{1}})  - \bH(\bZ_{(uv)}^{\theta, z_{uv}^{0}}) ] \bH^{-1}(\bZ_{(uv)}^{\theta, z_{uv}^{0}}).\]
\begin{equation}\label{eq:diff_bH_entry_change}
\bH(\bZ_{(uv)}^{\theta, z_{uv}^{1}})  - \bH(\bZ_{(uv)}^{\theta, z_{uv}^{0}})  = -(z_{uv}^{1} - z_{uv}^{0}) [ e_{(p_1+n-1)u} e_{nv}^T \Delta^T +\Delta e_{nv} e_{( p_1+n -1)u}^T].
\end{equation}
Here, $e_{pi}$ is the standard unit vector of dimension $p$ in the coordinate direction $i$ and 
\[\Delta =\Delta(\bY) = \begin{pmatrix}
    \bU_k^T \\[4pt]
    0_{p_1\times n} \\[4pt]
    \bU_{\perp}^T
\end{pmatrix}.\] 
Arguments in Sections 9.1, 9.2, and 10 of \cite{han2018unified} are devoted to the study of the above quantities when $\lambda=0$. 

The similarities between our setting and that of \citet{han2018unified} can be summarized as follows. We first assume that $\bY$ is fixed and $\mathcal{E}_C$ is true. Then, the arguments in \citet{han2018unified} apply. First, the bounds on the norms $\|\bH(z)\|$, $\|\bH^{-1}(z)\|$, and $\|\partial_z\bH^{-1}(z)\|$ by the imaginary part of $z$ are unaffected by the presence of the parameter $\lambda$. In particular, the estimates in (9.16)--(9.19) of \citet{han2018unified} remain valid for all $\lambda>0$. Second, bounds of the form 
\[\mathbb{I}(\mathcal{E}_C)  |v^{T}(\bH^{-1}(\bZ_x^\theta)-\widehat{\Pi}(z))w| \mid \bY \]
such as those established in Lemma~8 of \citet{han2018unified}, continue to hold for 
\[\mathbb{I}(\mathcal{E}_C) |v^{T}(\bH^{-1}(\bZ_x^\theta)-\Pi(z))w| \mid \bY.\]
Third, the analysis of the perturbation induced by modifying a single entry of $\bZ_x^\theta$ can be carried out in exactly the same manner as in \citet{han2018unified}. In particular, equation~\eqref{eq:diff_bH_entry_change} does not depend on $\lambda$, since in the difference 
$\bH(\bZ_{(uv)}^{\theta,z_{uv}^{1}}) - \bH(\bZ_{(uv)}^{\theta,z_{uv}^{0}})$ the middle block cancels.

Importantly, the bounds on 
\[\mathbb{I}(\mathcal{E}_C) |v^{T}(\bH^{-1}(\bZ_x^\theta)-\Pi(z))w|\mid \bY\] 
can be verified to be uniform in $\bY$. It is because $\bY$ affects the difference only through $\Delta(\bY)$, whose spectral norm and Frobenius norm are uniformly bounded for any $\bY$ given that $\mathcal{E}_C$ holds.

The main difference between our setting and that of \citet{han2018unified} lies in the treatment of the deterministic equivalent matrix. In our context, concentration is established with respect to the deterministic equivalent $\Pi(z)$, whereas \citet{han2018unified} employs a $\bZ_x \bU_{\perp}$–dependent equivalent $\widehat{\Pi}(z)$. This distinction leads to a simpler analysis in our case. Indeed, when individual entries of $\bZ_x$ are modified, the matrix $\widehat{\Pi}(z)$ changes accordingly, while $\Pi(z)$ remains unaffected. Consequently, all arguments in \citet{han2018unified} that pertain to the analysis of $\widehat{\Pi}(z)$ can be omitted in our setting. These include, for example, the definition of the groups $g^{(j)}$ in Section~9.1.2, the second term on the right-hand side of equation~(9.42), the quantities $\mathcal{H}_{sti}$ and $\mathcal{H}_{stu}$ appearing between equations~(9.48) and~(9.49), equation~(9.53), as well as the discussion from the fourth paragraph of Section~9.2 up to equation~(9.60), among others.

With these modifications in place, the remaining arguments of \citet{han2018unified} can be applied in a verbatim manner. We therefore omit further details.
\end{proof}

\subsection{Proof of Theorems~\ref{thm:consistency_trace_based} and \ref{thm:consistency_largest_root}}
\label{subsec:proof_consistency}

Let
\[
\bW_{k1}^{(0)}
=
\frac1k
\Sigma_0^{1/2}
\bZ_x
\bP_k
\bZ_x^T
\Sigma_0^{1/2},
\qquad
\bW_{k2}^{(0)}
=
\frac1{n-1-k}
\Sigma_0^{1/2}
\bZ_x
(I_n-\bP_k)
\bZ_x^T
\Sigma_0^{1/2},
\]
which coincide with the corresponding cross-product matrices under the null hypothesis
$M\Lambda_m^T=0$.

Under the alternative,
\[
\bX
=
M\Lambda_m^T\bY
+
\Sigma_0^{1/2}\bZ_x,
\]
we have
\begin{align*}
\bW_{k2}
=&\,
\bW_{k2}^{(0)}
+
\frac1{n-1-k}
\Sigma_0^{1/2}
\bZ_x
(I_n-\bP_k)
\bY^T
\Lambda_m
M^T   \\
&
+
\frac1{n-1-k}
M\Lambda_m^T
\bY
(I_n-\bP_k)
\bZ_x^T
\Sigma_0^{1/2}
+
\frac1{n-1-k}
M\Lambda_m^T
\bY
(I_n-\bP_k)
\bY^T
\Lambda_m
M^T .
\end{align*}

By Condition~\eqref{eq:consistency_eq1},
\[
n^{-1/2}
\bigl\|
M\Lambda_m^T
\bY
(I_n-\bP_k)
\bigr\|_2
=
o_p(1).
\]
Moreover,
\[
\|\bZ_x\|_2
=
O_p(\sqrt n).
\]
Therefore,
\[
\|
\bW_{k2}
-
\bW_{k2}^{(0)}
\|_2
=
o_p(1).
\]

Since the eigenvalues of $\bW_{k2}^{(0)}+\lambda I_{p_1}$ are uniformly bounded above and away from zero with probability tending to one, the same property holds for
$\bW_{k2}+\lambda I_{p_1}$. Consequently, there exists an event
$\mathfrak E_1$ satisfying
\[
\mathbb P(\mathfrak E_1)\longrightarrow1,
\]
on which
\[
\ell_{\min}
(\bW_{k2}+\lambda I_{p_1})
\asymp1,
\qquad
\ell_{\max}
(\bW_{k2}+\lambda I_{p_1})
\asymp1.
\]

Next,
\begin{align*}
\frac{k}{n}\bW_{k1}
=&\,
\frac{k}{n}\bW_{k1}^{(0)}
+
\frac1n
M\Lambda_m^T
\bY
\bP_k
\bY^T
\Lambda_m
M^T
\\
&
+
\frac1n
M\Lambda_m^T
\bY
\bP_k
\bZ_x^T
\Sigma_0^{1/2}
+
\frac1n
\Sigma_0^{1/2}
\bZ_x
\bP_k
\bY^T
\Lambda_m
M^T.
\end{align*}

Since
\[
\Bigl\|
\frac{k}{n}\bW_{k1}^{(0)}
\Bigr\|_2
=
O_p(1),
\]
and the signal-strength assumption implies
\[
\frac{
\bigl\|
n^{-1}
M\Lambda_m^T
\bY
\bP_k
\bZ_x^T
\Sigma_0^{1/2}
\bigr\|_2
}{
\|\mathcal S_{nk}\|_2
}
=o_p(1),
\]
we obtain
\[
\frac{k}{n}\bW_{k1}
=
\mathcal S_{nk}
+
o_p\!\left(
\|\mathcal S_{nk}\|_2
\right)
\]
in spectral norm.

Hence there exists an event $\mathfrak E_2$ satisfying
$\mathbb P(\mathfrak E_2)\to1$ such that
\[
\frac{k}{n}
\|\bW_{k1}\|_2
\ge
\frac12
\|\mathcal S_{nk}\|_2
\]
on $\mathfrak E_2$.

Combining $\mathfrak E_1$ and $\mathfrak E_2$, we obtain
\[
\mathcal T(k,\lambda)
\ge
C
\|\mathcal S_{nk}\|_2,
\]
and
\[
\ell_{\max}(k,\lambda)
\ge
C
\|\mathcal S_{nk}\|_2,
\]
for some deterministic constant $C>0$, with probability tending to one.

Moreover, by Lemma~\ref{lemma:consistency_under_Ha},
\[
\hat\Omega_1(\lambda,q)-\Omega_1(\lambda,q)
=o_p(1),
\qquad
\hat\Omega_2(\lambda,q)-\Omega_2(\lambda,q)
=o_p(1).
\]
Therefore,
\[
\sqrt{p_1}
\frac{
\mathcal T(k,\lambda)
-
k\hat\Omega_1(\lambda,q)
}
{
\sqrt{k}\,
\hat\Omega_2(\lambda,q)
}
\stackrel{P}{\longrightarrow}
\infty,
\]
which proves
\[
\mathbb P\!\left(
\widetilde{\mathcal T}(k,\lambda)
>
z_{1-\alpha}
\right)
\longrightarrow1.
\]

The proof of Theorem~\ref{thm:consistency_largest_root} is analogous. The lower bound for $\ell_{\max}(k,\lambda)$ together with the consistency of $\hat\Theta_1(\lambda,q)$ and $\hat\Theta_2(\lambda,q)$ established in Lemma~\ref{lemma:consistency_under_Ha} yields
\[
\mathbb P\!\left(
\widetilde{\ell}_{\max}(k,\lambda)
>
\operatorname{TW}_{1, \, 1-\alpha}
\right)
\longrightarrow1.
\]
The proof is therefore omitted.

\subsection{Proof of \eqref{eq:characterize_trace_rankone} and \eqref{eq:characterize_largest_rankone}}\label{subsec:proof_characterize_trace_largest}

We now consider the proof of \eqref{eq:characterize_trace_rankone} and \eqref{eq:characterize_largest_rankone}. Under the assumed rank-one alternative model, 
\begin{align*}
\frac{1}{d_n^2} \bW_{k1} = \frac{1}{k} \mu_n \nu_n^T \bZ_y \bP_k  \bZ_y^T \nu_n \mu_n^T + \frac{1}{kd_n} \Sigma_0^{1/2} \bZ_x \bP_k \bZ_y^T \nu_n \mu_n^T\\
+ \frac{1}{kd_n} \mu_n \nu_n^T \bZ_y \bP_k \bZ_x^T \Sigma_0^{1/2}  + \frac{1}{kd_n^2}\Sigma_0^{1/2}\bZ_x \bP_k \bZ_x^T \Sigma_0^{1/2}.
\end{align*}
We first show that the last three terms are asymptotically negligible. First,
\[ \frac{1}{kd_n^2}\Big\| \Sigma_0^{1/2} \bZ_x \bP_k \bZ_x^T \Sigma_0^{1/2}\Big\|_F \leq \frac{1}{d_n^2}\Big\| \Sigma_0^{1/2} \bZ_x \bZ_x^T \Sigma_0^{1/2} \Big\|_F = O_p(\tr(\Sigma_0)/d_n^2) = o_p(1). \]
Second,
\[\frac{1}{k|d_n|} \Big\|\Sigma_0^{1/2} \bZ_x \bP_k \bZ_y^T \nu_n \mu_n^T\Big\|_F \leq \frac{1}{|d_n|}\Big\|\Sigma_0^{1/2} \bZ_x \Big\|_F \Big\| \bZ_y^T \nu_n \Big\|_F = O_p\Big(\sqrt{\tr(\Sigma_0)/d^2_n } \Big) = o_p(1). \]

On the other hand, following the arguments used in the proofs of Theorems~\ref{thm:consistency_trace_based} and \ref{thm:consistency_largest_root}, under Condition \ref{eq:consistency_eq1},
\[ \Big\| \bW_{k2} - \bW_{k2}^{(0)} \Big\|_2 =o_p(1), \]
and
\[ \Big\| \big(\bW_{k2} + \lambda I_{p_1}\big)^{-1} - \big(\bW_{k2}^{(0)} + \lambda I_{p_1}\big)^{-1} \Big\|_2 =o_p(1). \]

Following Theorem 3.6 of \cite{Yang2019EJP} on the local laws of the sample covariance matrix with separable covariance structure, we conclude that 
\[ \mu_n^T\Big(\bW_{k2}^{(0)} + \lambda I_{p_1}\Big)^{-1}\mu_n  = \mu_n^T \calD(-\lambda) \mu_n+ o_p(1).\]
It follows then 
\[ \mu_n^T \Big(\bW_{k2} + \lambda I_{p_1}\Big)^{-1} \mu_n = \mu_n^T \calD(-\lambda) \mu_n + o_p(1).\]

Combining the above estimates yields
\[
\frac1{d_n^2}\bF_{k\lambda}
=
\frac1k
\bigl(
\nu_n^T\bZ_y\bP_k\bZ_y^T\nu_n
\bigr)
\bigl(
\mu_n \mu_n^T\mathcal D(-\lambda)
\bigr)
+
\mathcal M,
\]
where $\|\mathcal M\|_F=o_p(1)$.
Since the leading term is rank one, its trace and largest eigenvalue are identical. Therefore,
\eqref{eq:characterize_trace_rankone} and
\eqref{eq:characterize_largest_rankone} follow immediately.

\subsection{Proof of Lemma~\ref{lemma:consistency_Bayes_lambda}}
\label{subsec:proof_lemma_consistency_Bayes_lambda}

Recall that the estimated risk function is
\[
\hat{\mathfrak{R}}(\lambda,\Pi)
=
\frac{
-\sum_{i=0}^{\theta}\pi_i\hat{\Upsilon}_i(\lambda)
}{
\hat{\Xi}(\lambda)
\sum_{i=0}^{\theta}\pi_i\hat{\mathfrak M}_i
},
\]
where $\hat{\Upsilon}_i(\lambda)$ is defined in \eqref{eq:def_Upsilon} and is a consistent estimator of $\Upsilon_i(\lambda)$. Pointwise consistency for every fixed $\lambda>0$ is established in \cite{li2020adaptable}. Since the candidate parameter space $\mathcal L$ is finite, the convergence is uniform over $\lambda\in\mathcal L$. Together with the assumption
\[
\sup_{\lambda\in\mathcal L}
\bigl|
\hat{\Xi}(\lambda)-\Xi(\lambda,q)
\bigr|
\stackrel{P}{\longrightarrow}
0,
\]
we obtain that, for every fixed prior model $\Pi$,
\[
\sup_{\lambda\in\mathcal L}
\bigl|
\hat{\mathfrak R}(\lambda,\Pi)
-
\mathfrak R(\lambda,\Pi)
\bigr|
\stackrel{P}{\longrightarrow}
0.
\]

Furthermore, both the numerator and denominator of $\mathfrak R(\lambda,\Pi)$ are linear functions of $\Pi$. Since $\mathfrak P$ is a compact subset of a finite-dimensional Euclidean space, the above convergence is also uniform over $\Pi\in\mathfrak P$. That is,
\[
\sup_{\Pi\in\mathfrak P}
\sup_{\lambda\in\mathcal L}
\bigl|
\hat{\mathfrak R}(\lambda,\Pi)
-
\mathfrak R(\lambda,\Pi)
\bigr|
\stackrel{P}{\longrightarrow}
0.
\]

We now prove the two assertions.

\begin{itemize}
\item[(i)]
Under the separation condition in Part~(i), for any fixed $\Pi\in\mathfrak P$,
\[
\mathbb P\!\left(
\min_{\lambda\in\mathcal L\setminus\{\lambda_\infty\}}
\hat{\mathfrak R}(\lambda,\Pi)
>
\hat{\mathfrak R}(\lambda_\infty,\Pi)
+\frac{\varepsilon}{2}
\right)
\longrightarrow
1.
\]
Hence,
\[
\mathbb P\!\left(
\hat{\lambda}_B(\Pi)=\lambda_\infty
\right)
\longrightarrow
1.
\]
Therefore, the weak convergence results in Theorems \ref{thm:main_fix_k} and  \ref{thm:main_diverge_k} still hold when $\lambda$ is replaced by $\hat{\lambda}_B$

\item[(ii)]
Under the separation condition in Part~(ii),
\[
\mathbb P\!\left(
\min_{\lambda\in\mathcal L\setminus\{\lambda_\infty\}}
\sup_{\Pi\in\mathfrak P}
\hat{\mathfrak R}(\lambda,\Pi)
>
\sup_{\Pi\in\mathfrak P}
\hat{\mathfrak R}(\lambda_\infty,\Pi)
+\frac{\varepsilon}{2}
\right)
\longrightarrow
1.
\]
Hence,
\[
\mathbb P(\hat{\lambda}_*=\lambda_\infty)
\longrightarrow
1.
\]
Therefore, the weak convergence results in Theorems \ref{thm:main_fix_k} and  \ref{thm:main_diverge_k} still hold when $\lambda$ is replaced by $\hat{\lambda}_*$
\end{itemize}

\bibliographystyle{apalike} 
\bibliography{references}       

\end{document}